\documentclass[11pt,a4paper]{article}

\usepackage{amsmath}
\usepackage{amssymb}
\usepackage{mathrsfs}
\usepackage{amsthm}
\usepackage{graphicx}
\usepackage{float}
\usepackage{color}
\usepackage{mathcomp}
\usepackage{mathtools}
\usepackage{authblk}

\usepackage{tikz}
\usetikzlibrary{calc,shapes,patterns,decorations.pathreplacing}

\usepackage{bm}
\usepackage{setspace}
\usepackage{hyperref}
\hypersetup{
	colorlinks = true,
	linkcolor = red,
	urlcolor = cyan,
	citecolor = blue,
}

\makeatletter
\DeclareRobustCommand\widecheck[1]{{\mathpalette\@widecheck{#1}}}
\def\@widecheck#1#2{%
    \setbox\z@\hbox{\m@th$#1#2$}%
    \setbox\tw@\hbox{\m@th$#1%
       \widehat{%
          \vrule\@width\z@\@height\ht\z@
          \vrule\@height\z@\@width\wd\z@}$}%
    \dp\tw@-\ht\z@
    \@tempdima\ht\z@ \advance\@tempdima2\ht\tw@ \divide\@tempdima\thr@@
    \setbox\tw@\hbox{%
       \raise\@tempdima\hbox{\scalebox{1}[-1]{\lower\@tempdima\box
\tw@}}}%
    {\ooalign{\box\tw@ \cr \box\z@}}}
\makeatother

\newcommand{\triple}[1]{{\left\vert\kern-0.25ex\left\vert\kern-0.25ex\left\vert #1 
    \right\vert\kern-0.25ex\right\vert\kern-0.25ex\right\vert}}

\def\sy{\textsc{y}} \def\sY{\textsc{Y}}
\def\xx{\textsc{X}} \def\XX{\textsc{X}}
 \def\XXX{\textbf{X}}
 \def\YYY{\textbf{Y}}
\def\yy{\textsc{Y}} \def\YY{\textsc{Y}}

\newcommand{\re}{\operatorname{Re}}
\newcommand{\im}{\operatorname{Im}}
\newcommand{\ii}{\mathrm{i}}
\newcommand{\ee}{\mathrm{e}}
\newcommand{\eps}{\epsilon}
\newcommand{\dd}{\mathrm{d}}
\newcommand{\G}{\mathbb{G}}
\newcommand{\I}{\mathbf{1}}
\newcommand{\s}{\mathsf{S}}
\newtheorem{theorem}{Theorem}[section] % reset theorem numbering for each chapter
\newtheorem{proposition}[theorem]{Proposition} 
\newtheorem{corollary}[theorem]{Corollary}
\newtheorem{lemma}[theorem]{Lemma}

\newtheorem{definition}[theorem]{Definition}  % definition numbers are dependent on theorem numbers

\newtheorem{remark}[theorem]{Remark}
\numberwithin{equation}{section}

\title{Supersymmetric Cluster Expansions and Applications to Random Schr\"odinger Operators}

\author[1,2]{Luca Fresta}
\affil[1]{Institute of Mathematics, University of Zurich, Winterthurerstrasse 190, 8057 Zurich, Switzerland}
\affil[2]{Department of Mathematics, University of T\"ubingen, Auf der Morgenstelle 10, 72076 T\"ubingen, Germany}

\begin{document}

\maketitle

\begin{abstract}
We study discrete random Schr\"odinger operators via the supersymmetric formalism. We develop a cluster expansion that converges  at both strong and weak disorder. 
We prove the exponential decay of the disorder-averaged Green's function and 
the smoothness of the local density of states either at weak disorder and at energies in proximity of the unperturbed spectrum or at strong disorder and at any energy.
As an application, we establish Lifshitz-tail-type estimates for the local density of states and thus localization at weak disorder.
\end{abstract}

\section{Introduction}

In this paper we consider discrete random Schr\"odinger operators on $\ell^{2}(\mathbb{Z}^{\mathrm{D}}) \otimes \mathbb{C}^{\s}$:
\begin{equation}
H_{\omega} = H + \gamma V_{\omega} \;,
\end{equation}
where $H$ is a Hermitian translation-invariant hopping operator, with fast decaying matrix elements, and $V_{\omega}$ is a local random potential, i.e., $\big(V_{\omega}u\big)_{x} = \omega_{x} u_{x}$, $\{\omega_{x} \}_{x \in \mathbb{Z}^{\mathrm{D}}}$ being i.i.d.~random variables with probability distribution $\nu(\dd \omega_{x})$. The set $\s \subset \mathbb{N}$ is a finite set of indices, which can possibly represent, e.g., spin or sub-lattice ``colour''.

Our work focuses on the study of the disorder-averaged Green's function via the supersymmetric (SUSY) formalism. The SUSY approach to random systems was pioneered in the physics literature by Parisi and Sourals \cite{ParisiSourlas,ParisiSourlas2}, and by Efetov \cite{Efetov82} based on the seminal work of Sch\"afer and Wegner \cite{Wegner79,SchaferWegner}. 
In the mathematics literature, the SUSY formalism has been rigorously applied in the study of random Schr\"odinger operators and random matrices, see \cite{KleinPerez,CampaninoKlein,
KleinMartinelli,ConstantinescuFelder,
BovierKleinPerez, Bovier,SjostrandWang2,Wang,
DisertoriSpencerZirnbauer,
DisertoriSpencer,DisertoriPinson,DisertoriLager2,
Shamis,Shcherbina3,Shcherbina1,Shcherbina2}.
%and in the theory of stochastic differential equations \cite{KleinLandau,Gubinelli}.
%ci sarebbero molti lavori matematici nella cosiddetta dimensional reduction, di Brydges Wright, Brydges collaborators, etc.

The analysis carried out in this paper is inspired by \cite{FrestaPorta}, where SUSY and renormalization group have been used to study a massless hierarchical model for disordered three-dimensional semimetals. 
One of the main obstacles in the control of the oscillatory SUSY integrals for disordered systems is represented by the presence of mostly complex reference Gaussian ``measures''. 
The extension of \cite{FrestaPorta} to the non-hierarchical case requires the use of a cluster expansion that exploits this strong oscillatory nature of the SUSY integrals.

The scope of this work is thus to develop a methodology that will be useful in future settings. The novelty of the present paper consists in the construction of SUSY cluster expansions based on the so-called Battle-Brydges-Federbush formula, see \cite{BattleFederbush,Brydges84}.
Notice that cluster expansions in the oscillatory SUSY context were previously considered in \cite{BovierKleinPerez} in the form of a Mayer trick that is applicable only to hopping operators without internal degrees of freedom and only at strong disorder,  therefore not general enough for our purposes.
On the other hand, our technique can be applied to lattice operators with internal degrees of freedom and with long range hopping. Consequently, by means of a dual SUSY representation of the Green's function that we introduce below, we are able to handle the weak disorder regime as well.
As a simple application of our method, we review some known results in the context of random Schr\"odinger operators.
% 
%We construct a SUSY cluster expansion that is applicable in the regimes of both strong disorder and weak disorder with massive covariance, and we review some known results in the context of random Schr\"odinger operators.
\medskip

Let us provide some preliminary definitions.
If $\Lambda \subset \mathbb{Z}^{\mathrm{D}}$ is a finite subset, we denote by $H_{\omega,\Lambda}$ the restriction of $H_{\omega}$ to $\ell^{2}(\Lambda) \otimes \mathbb{C}^{\s}$ with zero boundary conditions outside of $\Lambda$. The disorder-averaged Green's function at finite volume is the following $\mathbb{C}^{\s \times \s}$-valued function:
\begin{equation}
\G_{\Lambda}(x,y;z) : = \mathbb{E}_{\omega} \bigg(\frac{1}{H_{\omega,\Lambda} -z }  \bigg)_{x,y} \;,
\qquad z \in \mathbb{C} \setminus \mathbb{R}\;,
\end{equation}
%
%and possibly extended by continuity to $z \in \mathbb{R}$. 
where $\mathbb{E}_{\omega} $ denotes the expectation with respect to the product measure $\nu_{\mathbb{Z}^{\mathrm{D}}}(\dd \omega) := \bigtimes_{x \in \mathbb{Z}^{\mathrm{D}}} \nu(\dd\omega_{x})$.
The local density of states (LDOS) $\rho(E)$ at energy $E \in \mathbb{R}$ can be defined as:
\begin{equation}
\begin{split}
\rho(E) :=  \lim _{\eps \to 0^{+}}\lim _{\Lambda \nearrow \mathbb{Z}^{d}}\rho_{\eps,\Lambda}(E)\;,
\qquad \rho_{\eps,\Lambda}(E) :=  \frac{1}{\pi |\s|}\im  \,\mathrm{Tr}_{\s} \,\G_{\Lambda}(0,0;E + \ii \eps) \;.
\end{split}
\end{equation}
Since $H_{\omega}$ is ergodic, $\rho(E)$ exists for almost $E \in \mathbb{R}$, 
%The existence of the limit for almost every $E$ is a consequence of Birkhoff's ergodic theorem and of the properties of the Stieltjes transform, 
see \cite{AizenmanWarzel} and references therein for more details. 
%This definition of the LDOS can be naturally extended to a function of $E \in \mathbb{C}$.
\medskip

We make the following assumptions on the disorder distribution:
\begin{enumerate}
\item[(H1)] The measure $\nu$ is Lebesgue absolutely continuous, is even and satisfies the finite-moment condition $\int  |\omega|^{|\s|+1} \nu(\omega) \dd \omega \,< \infty$. This hypothesis is assumed throughout the rest of the script.

\item[(H2)] The Fourier transform of the density $\nu$ satisfies smoothness and decay conditions that will be made precise at separate times, in (H2-$\mathrm{I}$) which is based on Definition \ref{def: bounds_strong_disorder} and in (H2-$\mathrm{II}$) which is based on Definition \ref{def: bounds_weak_disorder}.
The hypotheses (H2-$\mathrm{I}$) and (H2-$\mathrm{II}$) will be assumed only throughout Section \ref{sec: SUSY_cluster_expansion_strong} and Section \ref{sec: SUSY_cluster_expansion_weak} respectively.
%Since we do not formalize this hypothesis now, we will state it in its precise form in the assumptions of the theorems that we are going to prove.
\end{enumerate}
We shall remark that we do not anticipate Definitions \ref{def: bounds_strong_disorder} and \ref{def: bounds_weak_disorder}
because they require the discussion of Section \ref{sec: SUSY_formalism}. In fact, assumption (H2) is formulated in terms of superfunctions rather than of $\nu$: the connection between the two will only be established in the appendix.

As will be clear, assumptions (H1) and (H2) are quite restrictive but yet apply to a large class of disorder distributions. 
This class includes measures with unbounded support like the Gaussian distribution and perturbations of it, but also measures whose density is smooth and compactly supported.
%On the other hand, non-smooth measures like the uniform distribution, are not not covered in this work.

\subsection{Results}

Let us summarize our results based on SUSY cluster expansions and compare them with the literature. Since $\nu$ is even, w.l.o.g.~we shall henceforth restrict to $\gamma > 0$.
In Section \ref{sec: SUSY_cluster_expansion_strong} we prove Theorem \ref{thm: exponential_decay_strong_disorder} and Corollary \ref{thm: DOS_strong_disorder} which respectively establish some properties of $\mathbb{E}_{\omega} G_{\omega,\Lambda}$ and $\rho(E)$ at strong disorder and at any energy. Assuming (H1) and (H2) our claims can be informally stated as follows:
\medskip

\noindent{}
\textbf{Theorem.} 
\textit{\begin{itemize}
\item[\textbf{(i)}] Let $H$ be a hopping operator with exponentially decaying matrix elements and let $E \in \mathbb{R}$. If  $\gamma \geq C$, $C$ depending on $E$, $\nu$ and the decay of the matrix elements of $H$, then  uniformly in $\Lambda$ and in $0 \leq \eps \leq 1$
\begin{equation}
\label{intro_decay_strong}
\sup _{\sigma,\sigma'}\big|\big(\mathbb{E}_{\omega} G_{\omega,\Lambda}(x,y;E \pm \ii \eps) \big)_{\sigma,\sigma'} \big| \leq C' \,\gamma^{-2+\delta_{x,y}} \ee^{-c|x-y|}\;,
\end{equation}
for some constant $C'$ depending on $E$, $\nu$ and the decay of the matrix elements of $H$, and some constant $c$ depending only on the latter.
\item[\textbf{(ii)}] Let $H$ be as above and $\re E$ be in a bounded set. Under the same assumptions on $\gamma$, the LDOS $\rho(E)$ is analytic provided that $|\mathrm{Im}\,E| < c$, $c$ depending on $\gamma$ and $\nu$.
\end{itemize}}
We remark that Theorem \ref{thm: exponential_decay_strong_disorder} also implies the Wegner estimate \cite{Wegner81}, which in turn implies localization via finite volume criteria \cite{AizenmanSchenker}. 

The exponential decay of the disorder-averaged Green's function for a class of disorder distributions that includes the Gaussian one and uniformly as $\eps \to 0^{+}$ is novel. In \cite{SpencerNotes} the author says that this is indeed expected to hold based on simple perturbative arguments, but that the corresponding SUSY model is otherwise difficult to analyze because of oscillations.
Decay estimates similar to \eqref{intro_decay_strong} are implied, e.g., by the probability estimates in \cite{FroehlichSpencer} or by fractional moments \cite{Aizenman}, but the resulting bounds are not uniform down to $\eps = 0$. In \cite{SjostrandWang1,SjostrandWang2} a SUSY representation is used to obtain exponential decay, uniformly in $\Lambda$ and $\eps$; the technique is based on complex deformation of the ``oscillatory measure'' and works for Cauchy distribution of the disorder (or perturbations of it).

Conversely, the regularity properties of the LDOS at strong disorder have been extensively studied, see, e.g., \cite{EdwardsThouless,ConstantinescuSpencer, BovierKleinPerez}. In \cite{BovierKleinPerez} the application of cluster expansion techniques to the SUSY representation of the LDOS is pioneered.
The authors 
consider the Laplacian on $\mathbb{Z}^{\mathrm{D}}$ in the presence of a random potential with uniform distribution of the disorder.
As anticipated, our analysis relies on a different expansion and applies to a larger class of Hamiltonians: the hopping is long-range and due to the presence of internal degrees of freedom, a non-trivial quartic fermionic interaction appears in the SUSY representation of $\mathbb{E}_{\omega}G_{\omega,\Lambda}$, not present in \cite{BovierKleinPerez} by the Pauli exclusion principle. Besides, in \cite{BovierKleinPerez} the constant $C$ is uniform in $E$, which can therefore span the entire real line. The result in (ii) is thus weaker in this regard, but the result in (iv) presented below is complementary and allows us to consider $E$ in unbounded sets at finite $\gamma$.
%Their main focus is to investigate the Anderson Model with uniform disorder distribution; their technique, however, can be applied only to tight-binding Hamiltonians which do not require the cluster-expansion machinery in its full generality.

\medskip

In Section \ref{sec: SUSY_cluster_expansion_weak} we prove Theorem \ref{thm: exp_decay_weak}, Corollary \ref{thm: analyticity_LDOS_weak} and Theorem \ref{thm: Lifshitz_LDOS_weak} which establish some properties of $\mathbb{E}_{\omega} G_{\omega,\Lambda}$ and $\rho(E)$ at weak disorder and at energies outside of the unperturbed spectrum. Assuming (H1) and (H2) our claims can be informally stated as follows:
\medskip

\noindent{}
\textbf{Theorem.} 
\textit{\begin{itemize}
\item[\textbf{(iii)}] Let $E$ be outside of $\sigma(H)$ and define $\delta:= \mathrm{dist}(E,\sigma(H))$. If $\gamma \leq C \delta$ for some small constant $C$ depending on $\nu$, then uniformly in $\Lambda$ and $\eps \geq 0$
\begin{equation}
\sup _{\sigma,\sigma'}\big|\big(\mathbb{E}_{\omega} G_{\omega,\Lambda}(x,y;E \pm \ii \eps) \big)_{\sigma,\sigma'} \big| \leq C'\, \gamma^{-\delta_{x,y}} \ee^{-\theta \sqrt{\delta}|x-y|} \;,
\end{equation}
for some constant $C'$ depending on $\nu$ and for any $\theta \in [0,1)$.
\item[\textbf{(iv)}] Let $\re E$ be in a bounded set outside of $\sigma(H)$. Under the same assumptions on $\gamma$ as in (iii), the LDOS $\rho(E)$ is analytic provided that $|\mathrm{Im}\,E| < c$, $c$ depending on $\gamma$ and $\nu$. 
%Under the same assumptions, the LDOS $\rho(E)$ is analytic in a strip $|\mathrm{Im} E| \leq C$ and for real values of $E$ satisfies Lifshitz-tail-type estimates.
\item[\textbf{(v)}] Under the same assumptions as in (iii), the following bound holds true uniformly in $\Lambda$ and $\eps \geq 0$
\begin{equation}
\big | \rho_{\epsilon,\Lambda}\big | \leq C' \, \gamma^{-1} \, \ee^{-c(\gamma \, \delta^{-1})^{-1/2p}} \;,
\end{equation}
for some constants $C'$, $c$ and $p$ depending on $\nu$.
\end{itemize}}

We remark that Lifshitz-tail-type estimates, presented in Theorem \ref{thm: Lifshitz_LDOS_weak}, are well-known \cite{Klopp} to imply localization via finite-volume criteria, see Remark~\ref{rmk: IDOS_rmk} for details.

The exponential decay of the disorder-averaged Green's function at weak disorder and at energies close to the spectrum was expected to hold true \cite{SpencerNotes}, but no proof was available to the best of our knowledge. One reason for this, is that the methods based on fractional moments or on probability estimates do not allow the direct control of the disorder-averaged Green's function. On the other hand, the SUSY formalism is suitable for studying the disorder-averaged Green's function, but the control of the estimates is cumbersome because of oscillations.

The LDOS was expected to be analytic at weak disorder and away from the unperturbed spectrum \cite{ConstantinescuSpencer}. In \cite{Bovier} Bovier studied the analyticity of the LDOS in a hierarchical model at weak Gaussian-distributed disorder and at energies in proximity of the ``band edge''.
The work is based on SUSY and on the renormalization group analysis of the hierarchical Laplacian. Our result applies to any hopping Hamiltonian with quadratic energy dispersion relation at the band edge, e.g., the discrete Laplacian $-\Delta_{\mathbb{Z}^{\mathrm{D}}}$.

Localization in the Lifshitz-tail regime has already been established in \cite{Aizenman,Wang, Klopp, Elgart}. In \cite{Aizenman}, Aizenman establishes localization up to $\delta \geq C \gamma^{\frac{1}{1+\mathrm{D} +\eps}}$, $\eps >0$ and $C$ universal constant. The result was improved by Wang \cite{Wang} to $\delta \geq \gamma$ and later on boosted by Klopp \cite{Klopp} up to $\delta \geq \gamma^{1 + \frac{\mathrm{D}}{4\mathrm{D} + 4}}$. Finally, in \cite{Elgart} Elgart proved localization up to $\delta \geq C_{\mathrm{edge}} \gamma^{2} +\gamma^{4 -\eps}$, with $\eps >0$ and optimal $C_{\mathrm{edge}} >0$\footnote{The constant $C_{\mathrm{edge}}$ is precisely the one expected for the mobility edge, and corresponds to the extraction of the tadpole diagram, see \cite{Elgart}.}. The proof is based on the systematic resummation of the ``tadpole graph'' in the perturbative expansion of the Green's function, as analysed by Spencer in \cite{Spencer}.
Notice that in \cite{Aizenman} a very general class of disorder distribution is considered and localization is established at energies close to the unperturbed spectrum. On the other hand,
in \cite{Wang,Klopp} and in \cite{Elgart} they consider disorder distributions with semi-bounded and bounded support respectively; furthermore, the result is established at energies close to the spectrum of the random Hamiltonian (in these cases $\delta \equiv \mathrm{dist}(E,\sigma(H_{\omega}))$).
Our Lifshitz-tail-type estimate in Theorem~\ref{thm: Lifshitz_LDOS_weak} applies to disorder measures with unbounded support, e.g., the Gaussian distribution, and allows us to prove localization in the proximity of the unperturbed spectrum, at energies up to
$\delta \geq \gamma \,|\ln \gamma|^{\alpha}$, for some $\alpha$ sufficiently large.
We believe that this is the best achievable result with a single-step SUSY cluster expansion.
\medskip

The paper is organized as follows. In Section~\ref{sec: SUSY_formalism} we describe the machinery of the ``superformalism'' and we provide two SUSY representations of the disorder-averaged Green's function.
In Section~\ref{sec: SUSY_cluster_expansion_strong} we formalise the assumptions on $\hat{\nu}$ as (H2-$\mathrm{I}$), and we prove the results (i) and (ii) above. In Section~\ref{sec: SUSY_cluster_expansion_weak} we introduce the assumption (H2-$\mathrm{II}$) on $\hat{\nu}$ and we prove the results (iii), (iv) and (v).
In appendix we discuss in more detail examples of disorder distributions that satisfy the hypotheses (H2-$\mathrm{I}$) and (H2-$\mathrm{II}$).

\section{SUSY formalism}
\label{sec: SUSY_formalism}
%
%Throughout this section we describe the machinery of the ``superformalism''.
After a brief introduction to
normed Grassmann algebras and superfunctions, we state three main propositions that are crucial in our analysis. We conclude the section with Proposition~\ref{prop: SUSY_representation}: we provide two SUSY representations of $\G_{\Lambda}$ that will be used respectively at strong and weak disorder.
The use of super Fourier transform and the estimation in norm of the SUSY integrals are the novel features of our method.

\subsection{Normed Grassmann Algebras}

Grassmann algebras formalise the algebraic structure of anticommuting variables. They are widely used in statistical mechanics and field theory \cite{SamuelI,SamuelII,Mastropietro}. It is useful to equip these algebras with a suitable norm: this will make the estimates in Sections~\ref{sec: SUSY_cluster_expansion_strong} and \ref{sec: SUSY_cluster_expansion_weak} rather simple and intuitive.
Previous examples of the use of norms in the context of Grassmann integration can be found in \cite{Feldman, BauerschmidtBrydges}.

\begin{definition}
A Grassmann algebra is a complex unital algebra whose generators anticommute.
\end{definition}
We will only consider Grassmann algebras with a finite number of generators. Let $\XX, \YY$ be subsets of $\Lambda$; we use the boldface font to denote the Cartesian product of such sets with $\s$, that is, we write $\XXX := \XX \times \s$, $\YYY := \YY \times \s$ and so on.
We introduce the following Grassmann algebras:
\begin{equation}
\mathscr{G}:= \bigwedge \mathbb{C}^{\s \times \{ \pm\}} \;,
\qquad \qquad
 \mathscr{G}^{\XX}:= \bigwedge\, \mathbb{C}^{\XXX \times \{\pm\}} \;,
\quad \forall \XX \subset \Lambda \;.
\end{equation}
We notice that $\mathscr{G} \cong \mathscr{G}^{\{x\}}$ for any $x \in \Lambda$, so that any discussion on $\mathscr{G}^{\XX}$ includes $\mathscr{G}$ as a special case.
Let $\{ \psi^{\varepsilon}_{x,\sigma}\}^{\varepsilon = \pm}_{(x,\sigma) \in \XXX}$ be the set of generators of $\mathscr{G}^{\XX}$ and let the set $\mathbf{\Lambda} \times \{ \pm\}$ be provided with a total order. It is easy to see that the Grassmann algebra $\mathscr{G}^{\XX}$ is a complex linear space of dimension $2^{2|\XXX|}$, the basis elements being
\begin{equation}
\label{eq:basis_elements_grassmann_algebra}
\psi^{\scriptscriptstyle{\mathcal{X}}} := \sideset{}{'} \prod_{(x,\sigma,\varepsilon) \in {\scriptscriptstyle{\mathcal{X}}}} 
\psi_{x,\sigma}^{\varepsilon} \;,
\qquad \text{for }\quad
{\scriptstyle{\mathcal{X}}} \subset \XXX \times \{ \pm \} \;,
\end{equation}
where the prime in the product operator means that the product is ordered. Accordingly, we can write any element $f \in \mathscr{G}^{\XX}$ as
\begin{equation}
\label{eq: decomposition_Grassmann_element}
f = \sum_{\mathcal{X} \subset \XXX \times \{ \pm \}} f_{\scriptscriptstyle{\mathcal{X}}} \,\psi^{\scriptscriptstyle{\mathcal{X}}} \;,
\qquad
\end{equation}
for some $f_{\scriptscriptstyle{\mathcal{X}}} \in \mathbb{C}$ that will be called coefficients of $f$.
\medskip

The generators of a Grassmann algebra are often referred to as anticommuting variables. It is useful for our purposes to think of $\mathscr{G}^{\XX}$ as a set of functions of such anticommuting variables: if $f \in \mathscr{G}^{\XX}$ and if $\{ \psi^{\varepsilon}_{x,\sigma}\}^{\varepsilon = \pm}_{(x,\sigma) \in \XXX}$ is the set of generators, we will write $f = f(\psi)$. 
Because we think of elements of a Grassmann algebra as functions of anticommuting variables, it is quite natural to introduce a linear operation like integration. Grassmann integration is defined as follows:
\begin{equation}
\begin{split}
&\int \dd \psi _{x,\sigma}^{\varepsilon}   1 = 0 \;,
\\
&\int \dd \psi_{x,\sigma}^{\varepsilon} \int \dd \psi^{\varepsilon'}_{x',\sigma'} f(\psi) = -\int \dd\psi^{\varepsilon'}_{x',\sigma'} \int \dd \psi_{x,\sigma}^{\varepsilon}f(\psi) \;,
\\
&\int \dd \psi_{x,\sigma}^{\varepsilon} \big(\psi^{\varepsilon'}_{x',\sigma'} f(\psi)\big)  = \delta_{\varepsilon,\varepsilon'} \delta_{x,x'}\delta_{\sigma,\sigma'} f(\psi) - \psi^{\varepsilon'}_{x',\sigma'} \int \dd \psi_{x,\sigma}^{\varepsilon} f(\psi)
\;.
\end{split}
\end{equation}
We shall also set
\begin{equation}
\label{def:berezin_integral}
\int \dd \psi_{\XX} \; \cdot \; := \prod_{(x,\sigma) \in \XXX} \int \dd \psi_{x,\sigma}^{+} \int \dd \psi^{-}_{x,\sigma} \; \cdot \;.
\end{equation}
Despite looking very peculiar, Grassmann integration is the cornerstone of the supersymmetric formalism, see Proposition \ref{prop: SUSY_replica_trick}.
Furthermore, we will use Grassmann integration as a tool for representing certain maps between elements of a Grassmann algebra, see, e.g., Lemma \ref{lemma: Grassmann_inversion}. This requires us to work within larger Grassmann algebras generated by two or more sets of independent variables, e.g., $\{\psi^{\varepsilon}_{x,\sigma} \}^{\varepsilon = \pm}_{(x,\sigma) \in \XXX}$ and $\{\eta^{\varepsilon}_{x,\sigma} \}^{\varepsilon = \pm}_{(x,\sigma) \in \XXX}$. In such cases, we will not explicitly refer to the larger algebra and we will only 
say that the Grassmann variables are independent.
\medskip

We now define normed Grassmann algebras, which are introduced already in \cite{Berezin66}.
\begin{definition}
A Grassmann algebra is said to be normed if it is equipped with a norm $\| \cdot \|$ satisfying:
\begin{equation}
\label{eq: Grassmann_norm_property}
\begin{split}
\Vert 1 \Vert = 1 \;, \qquad \Vert f \, g \Vert \leq \Vert f \Vert \Vert g \Vert \;, 
\end{split}
\end{equation}
$1$ denoting the multiplicative identity, $f$ and $g$ being any element of the Grassmann algebra.
\end{definition}
\begin{remark} \label{rmk:banach_algebra}
Notice that a normed Grassmann algebra is a Banach algebra.
\end{remark}
Henceforth, we equip the Grassmann algebra $\mathscr{G}^{\XX}$ with the following $\ell^{1}$-type norm:
\begin{equation}
\label{def: Grassmann_norm}
\|f \| := \sum_{\mathcal{X} \subset \XXX \times \{ \pm \}} | f_{\scriptscriptstyle{\mathcal{X}}} | \;,
\qquad
\forall f \in \mathscr{G}^{\XX}\;.
\end{equation}
\begin{lemma}
The Grassmann algebra $\mathscr{G}^{\XX}$ equipped with $\| \cdot \|$ defined in \eqref{def: Grassmann_norm} is a normed Grassmann algebra.
\end{lemma}
\begin{remark}
The same holds true for the larger Grassmann algebras with two or more sets of independent generators.
In the present work, the norm will be always implicitly associated with the largest Grassmann algebra we work with. 
\end{remark}
\begin{proof}
For any $f, g \in \mathscr{G}^{\XX}$ we have
\begin{equation}
\begin{split}
f \,g & = \sum_{\mathcal{X} \subset \XXX \times \{ \pm \}} \sum_{\mathcal{X}' \subset \XXX \times \{ \pm \}} f_{\scriptscriptstyle{\mathcal{X}}} \,g_{\scriptscriptstyle{\mathcal{X}'}} \,\psi^{\scriptscriptstyle{\mathcal{X}}}   \,\psi^{\scriptscriptstyle{\mathcal{X}'}}
\\
& = \sum_{\mathcal{X} \subset \XXX \times \{ \pm \}} \Big( \sum_{\substack{\mathcal{X}' \cap\mathcal{X}'' = \emptyset \\ \mathcal{X}' \cup \mathcal{X}'' = \mathcal{X} }} \mathrm{sign}({\scriptstyle{\mathcal{X}'}},{\scriptstyle{\mathcal{X}''}}) f_{\scriptscriptstyle{\mathcal{X}'}} \,g_{\scriptscriptstyle{\mathcal{X}''}} \Big) \,\psi^{\scriptscriptstyle{\mathcal{X}}} \;,
\end{split}
\end{equation}
for some $\mathrm{sign}({\scriptstyle{\mathcal{X}'}},{\scriptstyle{\mathcal{X}''}}) \in \{ \pm\}$ which we shall leave unspecified. Hence
\begin{equation}
\begin{split}
\| f\,g\| & = \sum_{\mathcal{X} \subset \XXX \times \{ \pm \}}\Big| \sum_{\substack{\mathcal{X}' \cap\mathcal{X}'' = \emptyset \\ \mathcal{X}' \cup \mathcal{X}'' = \mathcal{X} }} \mathrm{sign}({\scriptstyle{\mathcal{X}'}},{\scriptstyle{\mathcal{X}''}}) f_{\scriptscriptstyle{\mathcal{X}'}} \,g_{\scriptscriptstyle{\mathcal{X}''}} \Big| 
\\
& \leq \sum_{\mathcal{X} \subset \XXX \times \{ \pm \}} \sum_{\substack{\mathcal{X}' \cap\mathcal{X}'' = \emptyset \\ \mathcal{X}' \cup \mathcal{X}'' = \mathcal{X} }} \big| f_{\scriptscriptstyle{\mathcal{X}'}} \,g_{\scriptscriptstyle{\mathcal{X}''}} \big| \leq \| f \| \, \| g \| \;.
\end{split}
\end{equation}
\end{proof}

A normed Grassmann algebra is useful because we can estimate Grassmann integrals and thus avoid their exact computation.
The following property holds:
\begin{equation}
\label{eq:fermionic_norm_integral}
\Big| \int \dd \psi_{\xx} f(\psi)\Big| \leq \| f\| \;, \qquad \forall f \in \mathscr{G}^{\XX}\;,
\end{equation}
and will be used extensively below. Notice  that when we estimate Grassmann integrals, property \eqref{eq: Grassmann_norm_property} is particularly useful because the integrand $f(\psi)$ is usually the product of different terms that we want to control separately.

\subsection{Superfields and Superfunctions}

%Following \cite{Berezin66,Efetov80,Efetov10}, with abuse of terminology, all mathematical objects that involve both commuting and anti-commuting variables will contain the word ``super'' in their name. 

Supervectors are collections of commuting and anticommuting variables:
%\footnote{Notice that in our terminology supervectors are not elements of a super vector space.}
$\Phi := (\phi,\psi)$,
\begin{equation}
\phi = (\phi_{1,1},\phi_{1,2}\dots,\phi_{|\s|,1},\phi_{|\s|,2}) \in \mathbb{R}^{2|\s|} \;,
 \quad
 \psi = (\psi^{+}_{1},\psi^{-}_{1}, \dots,\psi^{+}_{|\s|},\psi^{-}_{|\s|}) \;,
\end{equation}
where $\{\psi^{\pm}_{\sigma}\}_{\sigma \in \s}$ is a set of generators of the normed Grassmann algebra $\mathscr{G}$.
If $\Phi$ is a supervector, we shall write $\Phi \in \mathcal{S}$. 
It is customary to introduce the complex-variable notation
\begin{equation}
\phi^{\pm}_{\sigma} := \phi_{\sigma,1} \pm \ii \phi_{\sigma,2} \;,
\end{equation}
that is, $\phi_{\sigma,1} = \re \phi^{+}_{\sigma}$ and $\phi_{\sigma,2} = \im \phi^{+}_{\sigma}$, and to denote by $\Phi^{+} = (\phi^{+},\psi^{+})$ and $\Phi^{-} = (\phi^{-},\psi^{-})$ the row and column vectors respectively. Thus, 
\begin{equation}
\Phi^{+} \Phi^{-} = \sum_{\sigma \in \s}\phi^{+}_{\sigma}\phi^{-}_{\sigma} + \psi^{+}_{\sigma}\psi^{-}_{\sigma} \;,
\end{equation}
is an element of $\mathscr{G}$, while $\Phi^{-}\Phi^{+}$ is a $\s \times \s$ supermatrix. 

Functions of a supervector are $\mathscr{G}$-valued maps, $f: \Phi \mapsto f(\Phi) \in \mathscr{G}$. For example, any polynomial in $\Phi^{+} \Phi^{-}$ is a function of a supervector; if $f(\cdot)$ is an analytic function, also $ f(\Phi^{+} \Phi^{-})$, defined by its Taylor expansion, is another instance of such functions.
\medskip

Superfields are maps from subsets of $\Lambda$ to supervectors, that is $
\Phi:  \XX \ni x \mapsto \Phi_{x} \in \mathcal{S}$.
If $\Phi$ is a superfield we shall write 
$\Phi \in \mathcal{S}^{\XX}$. Given a superfield $\Phi$, we define $\Phi^{\pm} : x \mapsto \Phi_{x}^{\pm}$. The contraction
\begin{equation}
\Phi^{+}_{x}A_{x,y}\Phi^{-}_{y} = \sum_{\sigma,\sigma' \in \s}  \Big[
\phi_{x,\sigma}^{+} \big(A_{x,y}\big)_{\sigma,\sigma'}\phi^{-}_{y,\sigma'} + \psi_{x,\sigma}^{+}\big(A_{x,y}\big)_{\sigma,\sigma'}\psi^{-}_{y,\sigma'} \Big] \;,
\end{equation}
where $A_{x,y} \in \mathbb{C}^{\s \times \s}$, will be widely used in the rest of the script.

Superfunctions are maps $f: \mathcal{S}^{\XX} \to \mathscr{G}^{\XX}$. It is clear that a superfunction can be decomposed as in Eq.~\eqref{eq: decomposition_Grassmann_element}:
\begin{equation}
f(\Phi) = \sum_{\mathcal{X} \subset \XXX \times \{ \pm \}} f_{\scriptscriptstyle{\mathcal{X}}}(\phi) \,\psi^{\scriptscriptstyle{\mathcal{X}}}
\end{equation}
where the functions $f_{\scriptscriptstyle{\mathcal{X}}}: \mathbb{R}^{2 \XXX} \to \mathbb{C}$ will be called the \textit{coefficients} of $f$. With abuse of notation we will write $f((\phi,0)) :=  f_{\emptyset}(\phi)$ or $ f(\Phi)|_{\psi =0}:=f_{\emptyset}(\phi)$.

We also introduce some useful spaces of superfunctions. We say that a superfunction $f: \mathcal{S}^{\XX} \to \mathscr{G}^{\XX}$ belongs to the space $L^{p}(\mathcal{S}^{\XX},\mathscr{G}^{\XX})$ if all its coefficients belong to the Banach space $L^{p}(\mathbb{R}^{2\XXX})$ of measurable functions $g: \mathbb{R}^{2\XX} \to \mathbb{C}$ such that $|g|^{p}$ is Lebesgue integrable if $p \in [1,\infty)$ or that are essentially bounded if $p = \infty$. Similarly, we introduce the space of Schwartz superfunctions $\mathscr{S}(\mathcal{S}^{\XX},\mathscr{G}^{\XX})$, that is, those superfunctions whose coefficients are elements of $\mathscr{S}(\mathbb{R}^{2\XXX})$. It is natural to consider the following norm on $L^{p}(\mathcal{S}^{\XX},\mathscr{G}^{\XX})$ for $p \in [1,\infty) $:
\begin{equation}
\|f\|_{L^{p}(\mathcal{S}^{\XX},\mathscr{G}^{\XX})} := \left( \int \dd \phi_{\xx}  \, \big(\|f(\Phi) \|\big)^{p}\right)^{1/p} \;,
\end{equation}
having set set
\begin{equation}
\dd \phi_{x} := 
%\bigtimes_{\sigma \in \s} \dd \phi^{+}_{x,\sigma} \dd \phi^{-}_{x,\sigma} = 
\bigtimes_{\sigma \in \s}\pi^{-1}\dd \phi_{x,\sigma, 1}\dd \phi_{x,\sigma,2} \;,
\qquad
\dd \phi_{\XX}:= \bigtimes_{x \in \xx} \dd \phi_{x} \;,
\end{equation}
and for $p = \infty$
\begin{equation}
\|f\|_{L^{\infty}(\mathcal{S}^{\XX},\mathscr{G}^{\XX})} : = \mathrm{ess}\, \sup\limits_{\negthickspace\negthickspace\negthickspace\negthickspace\negthickspace\negthickspace \phi \in \mathbb{R}^{2\XXX}} \, \| f(\Phi)\|\;.  
\end{equation}
%
%The typical case in the present paper is the generalization of $L^{1}(\mathbb{R}^{2\XXX},\mathbb{C})$ to the space of integrable superfunctions $L^{1}(\mathcal{S}^{\XX},\mathscr{G}^{\XX})$, that is, those superfunctions whose coefficients are Lebesgue integrable. Notice that in $L^{1}(\mathbb{R}^{2\XXX},\mathbb{C})$ the norm is with respect to the measure $\dd \phi_{\XX}$, that is, for any $f \in L^{1}(\mathbb{R}^{2\XXX},\mathbb{C})$ we have $\| f\|_{L^{1}(\mathbb{R}^{2\XXX},\mathbb{C})}:= \int \dd \phi_{\XX} |f(\phi)|$.

Superintegration is denoted by:
\begin{equation}
\int \dd \Phi_{\XX} \; \cdot \; := \int \dd \phi_{\XX} \, \int \dd \psi_{\XX}\; \cdot \;.
\end{equation}
Superintegration will be used as a tool for representing certain maps between superfunctions, e.g., the super Fourier transform, see Definition \ref{def:superfourier_transform}. This will require additional superfields and thus other independent Grassmann variables. When this is the case, we always implicitly work in a larger Grassmann algebra, generated by all the Grassmann variables we consider, see also discussion below \eqref{def:berezin_integral}, and say that we have \textit{independent} superfields.

As a simple consequence of \eqref{eq:fermionic_norm_integral}, notice the bounds:
\begin{equation}
\label{L1_bound_norms}
\Big | \int \dd \Phi_{\xx} \, f(\Phi) \Big | 
\leq \int \dd \phi_{\xx} \,\Big|\int \dd \psi_{\xx} f(\Phi) \Big|
\leq \int \dd \phi_{\xx} \, \big\|f(\Phi) \big \| = \| f\|_{L^{1}(\mathcal{S}^{\XX},\mathscr{G}^{\XX})} \;,
\end{equation}
which will be repeatedly used in the rest of the script. 

\subsection{Three main propositions}

The first identity that we present is the so-called supersymmetric replica trick, which is a way to write the entries of a matrix via super Gaussian integrals. This trick was first introduced in the context of random Schr\"odinger operators by Efetov \cite{Efetov82}.
\begin{proposition}
\label{prop: SUSY_replica_trick}
Let $A \in \mathbb{C}^{\XXX \times \XXX}$ be a complex matrix with positive definite Hermitian part. The following representation holds true:
\begin{equation}
(A^{-1})_{x,y} =
\int \dd \Phi_{\XX} \, \ee^{-\sum_{x',y'}\Phi_{x'}^{+}A_{x',y'}\Phi_{y'}^{-}} \, \psi^{-}_{x} \psi^{+}_{y}
\end{equation}
where $(A^{-1})_{x,y} \in \mathbb{C}^{\s \times \s}$ and $\psi^{-}_{x} \psi^{+}_{y}$ is a $\s \times \s$ matrix of Grassmann variables.
\end{proposition}
\begin{proof}
It is well-known, see for example Section 3 in \cite{Wegner16}, that for any invertible complex matrix $A \in \mathbb{C}^{\XXX \times \XXX}$
\begin{equation}
\label{eq: fermionic_gaussian}
(A^{-1})_{x,y} = (\det A)^{-1}
\int \dd \psi_{\xx} \, \ee^{- \sum_{x,y \in \xx} \psi^{+}_{x}A_{x,y} \psi^{-}_{y}} \psi^{-}_{x} \psi^{+}_{y} \;.
\end{equation}
We notice that if $\phi^{\pm} = \phi_{1} \pm \ii \phi_{2}$, $\phi_{i} \in \mathbb{R}^{\XXX}$, we can write
\begin{equation}
\label{manipulation}
\begin{split}
\sum_{x,y \in \xx} \phi^{+}_{x} A_{x,y} \phi^{-}_{y} & =
\sum_{x,y \in \xx} \sum_{\sigma,\sigma' \in \s} \phi^{+}_{x,\sigma} (A_{x,y})_{\sigma,\sigma'} \phi^{-}_{y,\sigma'} 
\\
& =  \begin{pmatrix}
\phi_{1}^{T} & \phi_{2}^{T}
\end{pmatrix}  \underbrace{\begin{pmatrix}
 \frac{A + A^{T}}{2} & - \ii \frac{A - A^{T}}{2} \\ \ii \frac{A - A^{T}}{2} & \frac{A + A^{T}}{2} 
\end{pmatrix}}_{=: \widetilde{A}}
\begin{pmatrix}
\phi_{1} \\ \phi_{2}
\end{pmatrix} \;,
\end{split}
\end{equation}
with $\widetilde{A} \in \mathbb{C}^{2\XXX \times 2\XXX}$ being symmetric.
If $A$ has positive Hermitian part then $\widetilde{A}$ is non-singular and $ \begin{pmatrix}
\phi_{1}^{T} & \phi_{2}^{T}
\end{pmatrix}  \re \widetilde{A}
\begin{pmatrix}\phi_{1} \\ \phi_{2}
\end{pmatrix} \geq 0$ for any $\phi_{i} \in \mathbb{R}^{\XXX}$, therefore
\begin{equation}
\label{eq: bosonic_gaussian}
\int \dd \phi_{\xx} \, \ee^{-\sum_{x,y \in \xx} \phi^{+}_{x} A_{x,y} \phi^{-}_{y}} = (\det A)^{-1} \;,
\end{equation}
see for instance Section 7 in \cite{HormanderI}.
Putting together \eqref{eq: fermionic_gaussian} and \eqref{eq: bosonic_gaussian} proves the claim.
\end{proof}

% It is important to notice that the positive definiteness requirement is necessary in order for the integral to be well-defined.

In the second proposition we state the super Plancherel identity. This identity is the keystone of the dual SUSY cluster expansion that we present in Section~\ref{sec: SUSY_cluster_expansion_weak}.
%
%This identity is the cornerstone of the analysis carried out in section, since it allows the construction of a weak-disorder SUSY cluster expansion.
%
%This identity is crucial in the analysis presented in Section~\ref{sec: SUSY_cluster_expansion_weak}, for it is the cornerstone of the construction of
%the cluster expansion that at weak disorder.
%Thanks to this identity, we are able to carry out a rigorous analysis of $\G_{\Lambda}$ at weak disorder via the supersymmetric formalism.
It is based on the theory of super Fourier transform, which we will briefly cover. We shall point out that Berezin had already considered the Fourier transform on Grassmann algebras in his pioneering work \cite{Berezin66}, see also \cite{BerezinMarinov,Berezin87}.
\begin{definition}[Super Fourier Transform]
\label{def:superfourier_transform}
Let $f \in L^{1}(\mathcal{S}^{\XX},\mathscr{G}^{\XX})$. The super Fourier transform of $f$, denoted by $\widehat{f}$, is the function $\widehat{f}: \mathcal{S}^{\XX} \ni \xi \mapsto \widehat{f}(\xi) \in \mathscr{G}^{\XX}$ defined by:
\begin{equation}
\label{eq: super_Fourier_Transform}
 \widehat{f}(\xi):= \int \dd \Phi_{\xx} \, \ee^{- \ii \sum_{x \in \xx} \big(\xi^{+}_{x} \Phi_{x}^{-}+\Phi_{x}^{+}\xi_{x}^{-}\big)} \, f(\Phi) \;,
\end{equation}
where $\xi = (\kappa,\eta) \in \mathcal{S}^{\XX}$ is another independent superfield and where
\begin{equation}
\xi^{+}_{x} \Phi_{x}^{-}+\Phi_{x}^{+}\xi_{x}^{-} = \sum_{\sigma \in \s} \xi^{+}_{x,\sigma} \Phi_{x,\sigma}^{-}+\Phi_{x,\sigma}^{+}\xi_{x,\sigma}^{-} \;.
\end{equation}
\end{definition}

Some important properties of the Fourier transform on $L^{1}(\mathbb{R}^{2\XXX})$ and $\mathscr{S}(\mathbb{R}^{2\XXX})$ carry over to $L^{1}(\mathcal{S}^{\XX},\mathscr{G}^{\XX})$ and $\mathscr{S}(\mathcal{S}^{\XX},\mathscr{G}^{\XX})$. In particular, we will see that the super Fourier transform is invertible in the latter space, the inversion being the super Fourier transform with flipped sign.  
\begin{proposition}[Super Plancherel identity]
\label{prop: super_Planchere_identity}
Let $f \in \mathscr{S}(\mathcal{S}^{\XX},\mathscr{G}^{\XX})$ and $g \in L^{1}(\mathcal{S}^{\XX},\mathscr{G}^{\XX})$, then
\begin{equation}
\int \dd \Phi_{\xx} \,f(\Phi) \, g(\Phi) = \int \dd \xi_{\xx}\, \widehat{f}(\xi) \, \widehat{g}(-\xi) \;.
\end{equation}
\end{proposition}
The proof of this statement is trivial once the inversion theorem for the Grassmann Fourier transform is established.
\begin{lemma}\label{lemma: Grassmann_inversion}
Let $\{\psi^{\varepsilon}_{x,\sigma} \}^{\varepsilon = \pm}_{(x,\sigma) \in \XXX}$ and $\{\eta^{\varepsilon}_{x,\sigma} \}^{\varepsilon = \pm}_{(x,\sigma) \in \XXX}$ be independent Grassmann variables. For any  $f = f(\psi)$ let
\begin{equation}
\label{eq: Grassmann_fourier}
\widehat{f}(\eta) := \int \dd \psi_{\xx}\,\ee^{-\ii\sum_{x \in \xx} \big(\eta^{+}_{x}\psi^{-}_{x} +  \psi^{+}_{x} \eta^{-}_{x} \big)} \,f(\psi) \;,
\end{equation}
be the Grassmann Fourier transform. Then,
\begin{equation}
\label{eq: Grassmann_fourier_inversion}
f(\psi) = \int \dd \psi_{\xx}\,\ee^{\ii\sum_{x \in \xx} \big( \psi^{+}_{x} \eta^{-}_{x} + \eta^{+}_{x}\psi^{-}_{x} \big)} \, \widehat{f}(\eta) \;.
\end{equation}
\end{lemma}

\begin{proof}
%The proof is obvious once the inversion theorem for the super Fourier transform is established. To this end, it suffices to check the case with anticommuting variables only. 
Let $\{\psi^{\varepsilon}_{x,\sigma} \}^{\varepsilon = \pm}_{(x,\sigma) \in \XXX}$, $\{\psi'^{\varepsilon}_{x,\sigma} \}^{\varepsilon = \pm}_{(x,\sigma) \in \XXX}$ and
 $\{\eta^{\varepsilon}_{x,\sigma} \}^{\varepsilon = \pm}_{(x,\sigma) \in \XXX}$
 be independent Grassmann variables.   Define
\begin{equation}
\delta^{\xx}(\psi):= \int \dd \eta_{\xx} \, \ee^{-\ii\sum_{x \in \xx} \big(\eta^{+}_{x}\psi^{-}_{x} +  \psi^{+}_{x} \eta^{-}_{x} \big)} \;.
\end{equation}
We are going to show that this function plays the role of the Dirac delta function in the anticommuting setting, that is
\begin{equation}
\label{eq: dirac_delta}
\int \dd \psi_{\xx} \, \delta^{\xx}(\psi - \psi') \, f(\psi) = f(\psi') \;.
\end{equation}
By linearity, it suffices to prove \eqref{eq: dirac_delta} in the case $f(\psi) = \psi^{{\scriptscriptstyle{\mathcal{X}}}}$, for any ${\scriptstyle{\mathcal{X}}} \subset (\XXX \times\{\pm\})$.
%Since no convergence problem arises, we accomplish the goal by proving the following two identities:
%%
%\begin{equation}
%\begin{split}
%& \int \dd \eta_{\xx} \, \ee^{-\ii\sum_{x \in \xx} \big(\eta^{+}_{x}\psi^{-}_{x} +  \psi^{+}_{x} \eta^{-}_{x} \big)} = 
%\prod_{(x,\sigma) \in \XXX} \psi^{-}_{x,\sigma}\psi^{+}_{x,\sigma} =:
%\delta^{\xx}(\psi) \;,
%\\ 
%& \int \dd \psi_{\xx} \, \delta^{\xx}(\psi - \psi') \, f(\psi) = f(\psi') \;.
%\end{split}
%\label{eq: delta_dirac_identities}
%\end{equation}
%%
%where $\eta, \psi, \psi'$ are different (independent) Grassmann variables. 
%To prove the first identity, 
First of all, we notice that by nilpotency
\begin{equation}
\ee^{-\ii\sum_{x \in \xx} \big(\eta^{+}_{x}\psi^{-}_{x} +  \psi^{+}_{x} \eta^{-}_{x} \big)} = \prod_{(x,\sigma,\varepsilon) \in \XXX \times \{ \pm\}} \big(1 - \ii\varepsilon \eta^{\varepsilon}_{x,\sigma}\psi^{-\varepsilon}_{x,\sigma}\big) \;.
\end{equation}
The term which contains all the $\eta$'s is by inspection $\prod_{(x,\sigma)} \eta^{-}_{x,\sigma} \eta^{+}_{x,\sigma}\psi^{-}_{x,\sigma} \psi^{+}_{x,\sigma}$, thus
\begin{equation}
 \delta^{\xx}(\psi)= \int \dd \eta_{\xx} \, \ee^{-\ii\sum_{x \in \xx} \big(\eta^{+}_{x}\psi^{-}_{x} +  \psi^{+}_{x} \eta^{-}_{x} \big)} = \prod_{(x,\sigma) \in \XXX} \psi^{-}_{x,\sigma}\psi^{+}_{x,\sigma} \;,
\end{equation}
where we used that $\eta^{-}_{x,\sigma} \eta^{+}_{x,\sigma}$ and $\psi^{-}_{x,\sigma} \psi^{+}_{x,\sigma}$ are even elements and thus commute.
As a consequence, for any set ${\scriptstyle{\mathcal{X}}} \subset (\XXX \times\{\pm\})$
we can write
\begin{equation}
\delta^{\XX}(\psi) = \mathrm{sign}({\scriptstyle{\mathcal{X}}})\psi^{{\scriptscriptstyle{\mathcal{X}}}'}\psi^{{\scriptscriptstyle{\mathcal{X}}} \;,}
\end{equation}
where ${\scriptstyle{\mathcal{X}}}' := (\XXX \times\{\pm\}) \setminus {\scriptstyle{\mathcal{X}}}$ and where $\mathrm{sign}({\scriptstyle{\mathcal{X}}})$ is a permutation sign which we leave unspecified. Again by nilpotency, we notice that the only term in $\delta^{\xx}(\psi - \psi')$ that gives non-vanishing contribution to \eqref{eq: dirac_delta} with $f(\psi) = \psi^{{\scriptscriptstyle{\mathcal{X}}}}$ is $\mathrm{sign}({\scriptstyle{\mathcal{X}}})\psi^{{\scriptscriptstyle{\mathcal{X}}}'}(-\psi')^{{\scriptscriptstyle{\mathcal{X}}}}$, so that
\begin{equation}
\begin{split}
\int \dd \psi_{\xx} \,\delta^{\xx}(\psi - \psi') \psi^{{\scriptscriptstyle{\mathcal{X}}} } & = \int \dd \psi_{\xx} \,\mathrm{sign}({\scriptstyle{\mathcal{X}}})\psi^{{\scriptscriptstyle{\mathcal{X}}}'}(-\psi')^{{\scriptscriptstyle{\mathcal{X}}}}
\psi^{{\scriptscriptstyle{\mathcal{X}}} } 
\\
& = \int \dd \psi _{\xx}\, \Big( \prod_{(x,\sigma) \in \XXX} \psi^{-}_{x,\sigma}\psi^{+}_{x,\sigma} \Big) \psi'^{{\scriptscriptstyle{\mathcal{X}}}} = \psi'^{{\scriptscriptstyle{\mathcal{X}}}}  \;,
\end{split}
\end{equation}
where in the second equality we used that $(-\psi')^{{\scriptscriptstyle{\mathcal{X}}}} \psi^{{\scriptscriptstyle{\mathcal{X}}}} = \psi^{{\scriptscriptstyle{\mathcal{X}}}}(\psi')^{{\scriptscriptstyle{\mathcal{X}}}}$ while in the last one the definition of $\int \dd \psi_{\xx}$. The inversion of the Grassmann Fourier transform is then established:
\begin{equation}
\begin{split}
&\int \dd \eta_{\xx} \, \ee^{\ii \sum_{x \in \xx} \big(\psi^{+}_{x} \eta^{-}_{x} + \eta^{+}_{x}\psi^{-}_{x} \big)} \, \widehat{f}(\eta)
\\
 = & \int \dd \psi'_{\xx} \,\int \dd \eta_{\xx} \, \ee^{-\ii\sum_{x \in \xx} \big(\eta^{+}_{x}(\psi'^{-}_{x}-\psi^{-}_{x}) +  (\psi'^{+}_{x} -\psi^{+}_{x}) \eta^{-}_{x} \big)} f(\psi')
\\
 = &\int \dd \psi'_{\xx} \delta^{\xx}(\psi' - \psi) f(\psi') = f(\psi) \;.
\end{split}
\end{equation}
\end{proof}
\begin{proof}[Proof of Proposition \ref{prop: super_Planchere_identity}]
By swapping the bosonic integration with Grassmann integration and using the invertibility of the Fourier transform in $\mathscr{S}(\mathbb{R}^{2\XXX},\mathbb{C})$ we obtain that
\begin{equation}
f(\Phi) = \int \dd \xi_{\xx} \, \ee^{ \ii \sum_{x \in \xx} \big(\xi^{+}_{x} \Phi_{x}^{-}+\Phi_{x}^{+}\xi_{x}^{-}\big)} \widehat{f}(\xi) \;,
\end{equation}
for any $f \in \mathscr{S}(\mathcal{S}^{\XX},\mathscr{G}^{\XX})$, thus establishing the invertibility of the super Fourier transform on this space. We have,
\begin{equation}
\begin{split}
\int \dd \Phi_{\xx} \,f(\Phi) \, g(\Phi)  & = \int \dd \Phi_{\xx}  \int \dd \xi_{\xx} \, \ee^{ \ii \sum_{x \in \xx} \big(\xi^{+}_{x} \Phi_{x}^{-}+\Phi_{x}^{+}\xi_{x}^{-}\big)} \widehat{f}(\xi)\, g(\Phi) 
\\
& = \int \dd \xi_{\xx}\, \widehat{f}(\xi) \, \widehat{g}(-\xi) \;,
\end{split}
\end{equation}
where we swapped the superintegrals by Fubini-Tonelli theorem, since \linebreak $\| \widehat{f} \|_{L^{1}(\mathcal{S}^{\XX},\mathscr{G}^{\XX})}$, $\| g\|_{L^{1}(\mathcal{S}^{\XX},\mathscr{G}^{\XX})}$ and $ \big \|\ee^{ \ii \sum_{x \in \xx} \big(\xi^{+}_{x} \Phi_{x}^{-}+\Phi_{x}^{+}\xi_{x}^{-}\big)} \big \|_{L^{\infty}(\mathcal{S}^{\XX},\mathscr{G}^{\XX})}$ are finite.
\end{proof}
Before moving to the last proposition of this section, we present a lemma that is the extension of a simple and well-known inequality in the theory of Fourier transform.
\begin{lemma}
\label{lemma: bound_fourier_transform}
Let $f \in L^{1}(\mathcal{S}^{\XX},\mathscr{G}^{\XX})$. The following bound holds true:
\begin{equation}
\|\widehat{f}\, \|_{L^{\infty}(\mathcal{S}^{\XX},\mathscr{G}^{\XX})} \leq \|f \|_{ L^{1}(\mathcal{S}^{\XX},\mathscr{G}^{\XX})} \;.
\end{equation}
\end{lemma}
\begin{proof}
Let $\Phi = (\phi,\psi)\in \mathcal{S}^{\XX}$ and $\xi = (\kappa,\eta) \in \mathcal{S}^{\XX}$ be independent superfields. Clearly we have
\begin{equation}
\label{eq:lemma_superfourier_transform}
\|\widehat{f}(\xi) \| \leq \int \dd \phi_{\xx} \,\Big \| \int \dd \psi_{\xx} \ee^{- \ii \sum_{x \in \xx} (\eta^{+}\psi^{-}_{x} + \psi^{+}_{x}\eta_{x}^{-})} \,f(\Phi) \Big\|\;.
\end{equation}
We shall prove that the Grassmann norm on the r.h.s.~of \eqref{eq:lemma_superfourier_transform} is equal to $\| f(\Phi)\|$ by showing that the Grassmann Fourier transform rearranges the coefficients up to a phase factor. 
We write the exponential term as 
\begin{equation}
\prod_{(x,\sigma,\varepsilon) \in \XXX \times \{ \pm\}} \big(1 - \ii\varepsilon \eta^{\varepsilon}_{x,\sigma}\psi^{-\varepsilon}_{x,\sigma}\big) = \sum_{{\scriptscriptstyle{\mathcal{X}}} \subset \XXX \times \{\pm \}} \mathrm{phase}({\scriptstyle{\mathcal{X}}}) \eta^{{\scriptscriptstyle{\mathcal{X}}}} \, \psi^{{\scriptscriptstyle{\overline{\mathcal{X}}}}}
\;,
\end{equation}
where $\mathrm{phase}({\scriptstyle{\mathcal{X}}}) \in U(1)$, while $\scriptstyle{\overline{\mathcal{X}}}$ is the subset of $\XXX \times \{\pm \}$ obtained from $\scriptstyle{\mathcal{X}}$ by flipping all the $\varepsilon$'s. By noticing that $\delta^{\xx}(\psi) = \mathrm{sign}({\scriptstyle{\overline{\mathcal{X}}'}})\psi^{{\scriptscriptstyle{\overline{\mathcal{X}}}}}\psi^{{\scriptscriptstyle{\overline{\mathcal{X}}'}}}$, where ${\scriptstyle{\overline{\mathcal{X}}'}} = (\XXX \times \{ \pm\}) \setminus {\scriptstyle{\overline{\mathcal{X}}}}$ (see proof of Lemma \ref{lemma: Grassmann_inversion}) we obtain
\begin{equation}
\int \dd \psi_{\xx} \ee^{- \ii \sum_{x \in \xx} (\eta^{+}\psi^{-}_{x} + \psi^{+}_{x}\eta_{x}^{-})} \,f(\Phi) = \sum_{{\scriptscriptstyle{\mathcal{X}}} \subset \XXX \times \{\pm \}} \mathrm{phase}({\scriptstyle{\mathcal{X}}}) \,\mathrm{sign}({\scriptstyle{\overline{\mathcal{X}}'}}) f_{{\scriptscriptstyle{\overline{\mathcal{X}}'}}}(\phi)
\eta^{{\scriptscriptstyle{\overline{\mathcal{X}}'}}} \;,
\end{equation}
and the claim follows since there is a one to one correspondence between ${\scriptstyle{\mathcal{X}}}$ and ${\scriptstyle{\overline{\mathcal{X}}'}}$.
\end{proof}

Last but not least, supersymmetry is a crucial property in the analysis of superintegrals. The last proposition we present is an instance of the well-known localization formula for supersymmetric functions \cite{ParisiSourlas}. 
%We shall accordingly consider superfunctions in $W^{1,1}_{0}(\mathcal{S}^{\XX},\mathscr{G}^{\XX})$, that is superfunctions whose coefficients belong to $W^{1,1}(\mathbb{R}^{2\XXX},\mathbb{C})$ and decay at infinity.
%

Let us first define what supersymmetry is.
\begin{definition}[SUSY]
Introduce the differential operator:
\begin{equation}
Q_{\Phi} = \sum_{(x,\sigma)\in \XXX}\sum_{\varepsilon} \Big[ \psi^{\varepsilon}_{x,\sigma} \frac{\partial}{\partial \phi^{\varepsilon}_{x,\sigma}} - \varepsilon \phi_{x,\sigma}^{\varepsilon} \frac{\partial}{\partial \psi^{\varepsilon}_{x,\sigma}} \Big] \;.
\end{equation}
We say that $f$ is supersymmetric if it is $Q$-closed, that is, it is differentiable and  satisfies
\begin{equation}
Q_{\Phi} f(\Phi) = 0 \;. 
\end{equation}
\end{definition}
For the sake of generality, we shall state the SUSY localization formula under weak decay assumptions. 
\begin{proposition}[SUSY localization formula]
\label{prop: SUSY_localization_theorem}
Let $f$ be supersymmetric and decaying at infinity. If $f$ and $\Phi \mapsto (1+|\phi^{\varepsilon}_{x,\sigma}|)^{-1}\,\psi^{\varepsilon}_{x,\sigma} \,(\partial/\partial \phi^{\varepsilon}_{x,\sigma}) f(\Phi)$ are in $L^{1}(\mathcal{S}^{\XX},\mathscr{G}^{\XX})$ for any $(x,\sigma,\varepsilon) \in \XXX \times \{ \pm\}$ then
\begin{equation}
\int \dd \Phi_{\xx} \, f(\Phi) = f(0) \;.
\end{equation}
\end{proposition}
Even though our assumptions are somewhat weaker than, e.g., \cite{DisertoriSpencerZirnbauer, BauerschmidtBrydges}, the proof outlined in those references carries over and is therefore here omitted.
A more geometrical perspective on this statement can be found in \cite{Berline, SchwarzZaboronsky,BlauThompson}. 

The following lemma will be useful for the application of the SUSY localization formula.
\begin{lemma}
\label{lemma: decomposition_Q_operator}
Let $f \in L^{1}(\mathcal{S},\mathscr{G})$ be even, supersymmetric and invariant under $U(1)^{\times \s}$ fermionic transformations, $\psi^{\varepsilon}_{\sigma} \mapsto \ee^{\ii \varepsilon \theta_{\sigma}}\psi^{\varepsilon}_{\sigma}$ for $\theta \in [0,2\pi)^{\s}$. Then
\begin{equation}
\label{eq: decomposition_Q_operator}
\psi^{\varepsilon}_{\sigma} \frac{\partial}{\partial \phi^{\varepsilon}_{\sigma}} f(\Phi) = -\varepsilon \phi_{\sigma}^{-\varepsilon} \frac{\partial}{\partial \psi^{-\varepsilon}_{\sigma}} f(\Phi) \;,
\qquad
\forall \sigma \in \s, \,\varepsilon \in \{\pm\} \;.
\end{equation}
\end{lemma}
\begin{proof}
Denote by $\mathcal{U}_{\theta}$ the $U(1)^{\times \s}$ fermionic transformation. Since $0$ and $f$ are invariant, we have that $\mathcal{U}_{\theta}\, Q_{\Phi}\,\mathcal{U}_{\theta}^{-1} f(\Phi) = 0$ for any $\theta \in [0,2\pi)^{\s}$. We decompose $Q_{\Phi} = \sum_{\varepsilon,\sigma} \widetilde{Q}_{\varepsilon,\sigma}$, with $\widetilde{Q}_{\varepsilon,\sigma}:= \psi^{\varepsilon}_{\sigma} \frac{\partial}{\partial \phi^{\varepsilon}_{\sigma}} +\varepsilon \phi_{\sigma}^{-\varepsilon} \frac{\partial}{\partial \psi^{-\varepsilon}_{\sigma}}$ satisfying $\mathcal{U}_{\theta} \,\widetilde{Q}_{\varepsilon,\sigma}\,\mathcal{U}_{\theta}^{-1} = \ee^{\ii \varepsilon \theta_{\sigma}} \widetilde{Q}_{\varepsilon,\sigma}$. Fix $\bar{\sigma}$ and $\bar{\varepsilon}$ and choose $\theta_{\sigma} =  \delta_{\bar{\sigma},\sigma}\pi$ and $\theta'_{\sigma} = \delta_{\bar{\sigma},\sigma} \pi/2$. Then, the linear combination
\begin{multline}
0 = \frac{1}{4}Q_{\Phi}f(\Phi) - \frac{1}{4}\mathcal{U}_{\theta} \,Q_{\Phi}\,\mathcal{U}_{\theta}^{-1} f(\Phi) 
\\
-  \frac{\ii}{4} \bar{\varepsilon}\, \mathcal{U}_{\theta'} \,Q_{\Phi}\,\mathcal{U}_{\theta'}^{-1} f(\Phi)+
\frac{\ii}{4} \bar{\varepsilon}\, \mathcal{U}_{\theta+\theta'} \,Q_{\Phi}\,\mathcal{U}_{\theta+\theta'}^{-1} f(\Phi)
= \widetilde{Q}_{\bar{\varepsilon},\bar{\sigma}} f(\Phi)  \;,
\end{multline}
which is exactly Eq.~\eqref{eq: decomposition_Q_operator} for $\bar{\sigma}$ and $\bar{\varepsilon}$.
\end{proof}

\subsection{Disorder-averaged Green's function}
\label{sec:SUSY_repr_Green}

In the proposition below two SUSY representations of the disorder-averaged Green's function are finally discussed. The first representation is well-known and has already been applied to the study of the Anderson model at strong disorder \cite{Efetov82}, see also \cite{BovierKleinPerez}.
The second representation is new, to the best of our knowledge. It is particularly useful at weak disorder and energies close to the spectrum of $H$.
We call these two SUSY representations
respectively ``direct SUSY integral'' and ``dual SUSY integral''.

We provide some preliminary definitions.
By assumption on the disorder distribution, see (H1), we have that $\int  |\omega|^{|\s|+1}\nu(\omega)\dd \omega < \infty$ and thus we can define the following function of a supervector:
\begin{equation}
\label{eq: definition_F_z}
F_{z}(\Phi) := \ee^{ \gamma^{-1}z (\Phi^{+}\Phi^{-})} \hat{\nu}(\Phi^{+} \Phi^{-}) \;,
\end{equation}
where $z \in \mathbb{C}$ and where
\begin{equation}
\hat{\nu}(\Phi^{+}\Phi^{-}) :=
\int \ee^{\ii \omega \Phi^{+} \Phi^{-}}\,\nu(\omega) \,\dd \omega
= \sum_{n = 0}^{|\s|}
\hat{\nu}^{(n)}(\phi^{+}\phi^{-})\frac{ (\psi^{+}\psi^{-})^{n}}{n!}
 \;.
\end{equation}
The condition on the moments of $\nu$ is necessary in order to have $\hat{\nu}^{(n)}$ well-defined for $n = 0,\dots,|\s|$. 
It is important to notice that $F_{z}$ is supersymmetric $Q_{\Phi}F_{z}(\Phi) = 0$,  even $F_{z}(\Phi) = F_{z}(-\Phi)$ and invariant under $U(1)^{\times \s}$ fermionic transformations, see Lemma~\ref{lemma: decomposition_Q_operator}.

We will also need to measure the decay of the lattice operators we consider. To this end, we introduce the following norm: for any $A \in \mathbb{C}^{(\mathbb{Z}^{\mathrm{D}} \times \s) \times (\mathbb{Z}^{\mathrm{D}} \times \s)}$
\begin{equation}
\label{lattice_operator_norm}
\| A\|_{\infty,1} := \sup _{x \in \mathbb{Z}^{\mathrm{D}}} \sum_{y \in \mathbb{Z}^{\mathrm{D}}} \sum_{\sigma, \sigma' \in \s} \big |(A_{x,y})_{\sigma,\sigma'} \big | \;.
\end{equation}
%

%We will henceforth denote the disorder-averaged Green's function as:
%%
%\begin{equation}
%\G_{\Lambda}(x,y;z) := \mathbb{E}_{\omega} G_{\omega,\Lambda}(x,y;z) \;,
%\qquad z \in \mathbb{C} \setminus \mathbb{R}\;.
%\end{equation}
%%
We can finally state the proposition.
\begin{proposition}[SUSY representation]
\label{prop: SUSY_representation}
Let $\Lambda \subset \mathbb{Z}^{\mathrm{D}}$ be finite and $\eps >0$.
Assume that $F_{0} \in L^{1}(\mathcal{S},\mathscr{G})$. Then for any $\gamma > 0$, and any $x,y \in \Lambda$ the functions $ \G_{ \Lambda}(x,y; E + \ii \epsilon)$ and $ \G_{ \Lambda}(x,y; E - \ii \epsilon)$ are analytic respectively in $ \im E > 0$ and in $ \im E < 0$, and the following representations hold true:
\begin{itemize}
\item[(i)] Direct SUSY integral:
\begin{equation}\label{eq: SUSY_greens_function_direct}
   \G_{ \Lambda}(x,y; E \pm \ii \epsilon) =\pm \ii  \, \gamma^{-1} \,\int \dd \mu ^{\pm } _{\Lambda}(\Phi) F^{\Lambda}_{\pm \ii E -  \eps}(\Phi) \psi^{-}_{x} \psi^{+}_{y} \; ,
\end{equation}
where for any $\XX \subset \Lambda$ we define:
\begin{align}\label{eq: functional_and_measure_direct}
& \dd \mu _{\xx}^{\pm}(\Phi) = \dd \Phi_{\xx} \, \mu _{\xx}^{\pm}(\Phi)\,,
&& 
\mu _{\xx}^{\pm}(\Phi) := \ee^{\mp \ii \gamma^{-1}\sum_{x,y \in \xx}  \Phi_{x}^{+} H_{x,y}\Phi_{y}^{-}} \,, \nonumber
\\
& F^{\xx}_{\pm \ii E -  \eps}(\Phi) =  \prod_{x \in \XX} F_{\pm \ii E -  \eps}(\Phi_{x}) \;. &&
\end{align}
\item[(ii)] Dual SUSY integral:
\begin{equation}\label{eq: SUSY_greens_function_fourier}
 \G_{\Lambda}(x,y; E \pm \ii  \epsilon) =\pm \ii  \,  \gamma^{-1} \,\int \dd \widehat{\mu}^{\pm} _{\Lambda}(\xi)\frac{\partial}{\partial\eta^{+}_{x}} \frac{\partial}{\partial\eta^{-}_{y}} \widehat{F_{0}}^{\Lambda}(\xi)  \; ,
\end{equation}
where for any $\XX \subset \Lambda$ we define:
\begin{align}\label{eq: functional_and_measure_fourier}
& \dd \widehat{\mu} _{\xx}^{\pm}(\xi) = \dd \xi_{\xx} \, \widehat{\mu} _{\xx}^{\pm}(\xi)\;,
& 
&\widehat{\mu} _{\xx}^{\pm}(\xi) := \ee^{\pm\ii \gamma \sum_{x,y \in \XX}  \xi_{x}^{+} C_{x,y}^{E \pm\ii \eps}\xi_{y}^{-}} \;, 
\nonumber
\\
&  \widehat{F_{0}}^{\xx}(\xi) =  \prod_{x \in \XX} \widehat{F_{0}}(\xi_{x})\;,&& 
\end{align}
with $C^{z} = (H -z)^{-1}$ and $\xi = (\kappa,\eta) \in \mathcal{S}^{\Lambda}$, see below Eq.~\eqref{eq: super_Fourier_Transform}.
%%
%\begin{equation}
%\begin{split}
%& \xi_{x}^{+} C_{xy}^{z} \xi_{y}^{-} = \kappa_{x}^{+} C_{xy}^{z} \kappa_{y}^{-} + \eta_{x}^{+} C_{xy}^{z} \eta_{y}^{-} \;,
%\\
%& \dd \xi_{\xx} := \prod_{x \in \sX} 
%\dd \xi_{x}\;,
%\qquad \dd \xi_{x} := \dd \eta_{x}^{+} \dd \eta_{x}^{-} \dd \kappa_{x}^{+} \dd \kappa_{x}^{-} \;.
%\end{split}
%\end{equation}
%
\end{itemize}
\end{proposition}
\begin{proof}
We apply Proposition~\ref{prop: SUSY_replica_trick} to 
$\pm \ii (H_{\omega} - E \mp \ii \eps) \in \mathbb{C}^{\mathbf{\Lambda}\times\mathbf{\Lambda}}$, which has positive definite Hermitian part for any $\omega$ ($H_{\omega} $ is Hermitian) provided that $\eps >0$. After rescaling $\Phi \to \gamma^{1/2} \Phi$ (which preserves $\dd \Phi_{\Lambda}$, see, e.g., \cite{Berezin66}) we can write:
\begin{equation}
\pm \ii \, \Big(\!\big(H_{\omega,\Lambda} -E \mp \ii \eps \big)^{-1} \Big)_{x,y}   = \gamma^{-1} \int \dd \mu ^{\pm}_{\Lambda} (\Phi) \,\psi^{-}_{x} \psi^{+}_{y} \, \ee^{\mp \ii  \sum_{x'} [\omega_{x'} -\gamma^{-1}( E \pm \ii \eps)] \Phi^{+}_{x'}\Phi^{-}_{x'}}
\end{equation}
where $\psi^{-}_{x} \psi^{+}_{y}$ is a $\s \times \s$ matrix of Grassmann variables and where $\dd \mu ^{\pm}_{\Lambda} (\Phi)$ is as in the statement. We shall swap disorder-expectation with superintegration by Fubini-Tonelli theorem. To this end, we need to prove that
\begin{equation}
\label{fub_ton_chap2}
\sup _{\sigma,\sigma '}\, \mathbb{E}_{\omega} \,\int \dd \phi_{\Lambda} \, \Big |\int \dd \psi_{\Lambda} \,\, \mu ^{\pm}_{\Lambda} (\Phi) \,\psi^{-}_{x,\sigma} \psi^{+}_{y,\sigma'}\, \ee^{\mp \ii  \sum_{x'} [\omega_{x'} -\gamma^{-1}( E \pm \ii \eps)] \Phi^{+}_{x'}\Phi^{-}_{x'}} \Big | < \infty \;.
\end{equation}
By means of \eqref{eq:fermionic_norm_integral}, see also \eqref{L1_bound_norms}, we have that
\begin{multline}
\label{eq: swapping_condition}
 \text{l.h.s.~\eqref{fub_ton_chap2}} \\
 \leq  \sup _{\sigma,\sigma '}\, \mathbb{E}_{\omega} \,\int \dd \phi_{\Lambda} \, \Big \|\mu ^{\pm}_{\Lambda} (\Phi) \,\psi^{-}_{x,\sigma} \psi^{+}_{y,\sigma'}\, \ee^{\mp \ii  \sum_{x'} [\omega_{x'} -\gamma^{-1}( E \pm \ii \eps)] \Phi^{+}_{x'}\Phi^{-}_{x'}} \Big\|  \;.
\end{multline}
Since $\Lambda$ is finite, for any hopping Hamiltonian the following quantity is bounded 
\begin{equation}
\|  H_{\Lambda} \|_{\infty,1} = \sup_{x \in \Lambda}\sum_{x' \in \Lambda} \sum_{\sigma,\sigma' \in \s } \big|(H_{x,x'})_{\sigma,\sigma' } \big| < \infty \;,
\end{equation}
see \eqref{lattice_operator_norm}, where $H_{\Lambda} = \mathbf{1}_{\Lambda}H\mathbf{1}_{\Lambda}$, $\mathbf{1}_{\Lambda}$ being the characteristic function of $\Lambda$. It appears that we can bound the Grassmann norm of $\mu ^{\pm}_{\Lambda} (\Phi)$ as follows:
\begin{equation}
\label{eq: bound_gibbs_free}
\begin{split}
\big \|\mu ^{\pm}_{\Lambda} (\Phi) \big \|  & = \Big \| \ee^{\mp \ii \gamma^{-1}\sum_{x,y \in \Lambda}  \phi_{x}^{+} H_{x,y}\phi_{y}^{-}} \ee^{\mp \ii \gamma^{-1}\sum_{x,y \in \Lambda}  \psi_{x}^{+} H_{x,y}\psi_{y}^{-}}   \Big \|
\\
& \leq \ee^{\gamma^{-1} \| \sum_{x,y \in \Lambda}  \psi_{x}^{+} H_{x,y}\psi_{y}^{-} \|} 
\leq \ee^{\gamma^{-1} \| H_{\Lambda}\|_{\infty,1} \,| \Lambda|} \;,
\end{split}
\end{equation}
where we used the properties of the norm and that $\sum_{x,y \in \Lambda} \phi^{+}_{x}H_{x,y}\phi^{-}_{y} \in \mathbb{R}$.
 
If we denote by $\mathrm{Tay}_{n}$ the $n$-th order Taylor expansion in zero, then the following bound
\begin{multline}
\label{eq: a_priori_omega_integration_bound}
\sup _{\sigma,\sigma '}\,\Big \|\mu ^{\pm}_{\Lambda} (\Phi) \,\psi^{-}_{x,\sigma} \psi^{+}_{y,\sigma'} \, \ee^{\mp \ii  \sum_{x'} [\omega_{x'} -\gamma^{-1}( E \pm \ii \eps)] \Phi^{+}_{x'}\Phi^{-}_{x'}} \Big\| \\ 
\leq \big \|\mu ^{\pm}_{\Lambda} (\Phi) \big \| \, \sup _{\sigma,\sigma '}\,\Big \| 
\psi^{-}_{x,\sigma} \psi^{+}_{y,\sigma'} \Big \| \,\prod_{x' \in \Lambda}  \Big \|\ee^{\mp \ii   [\omega_{x'} -\gamma^{-1}( E \pm \ii \eps)] \Phi^{+}_{x'}\Phi^{-}_{x'}} \Big\|
\\
\leq 
 \ee^{\gamma^{-1} \| H_{\Lambda}\|_{\infty,1} \, |\Lambda|}
 \left [ \,\prod_{x' \in \Lambda} 
 \ee^{-\eps \gamma^{-1}\phi^{+}_{x'}\phi^{-}_{x'} } \Big \|\ee^{\mp \ii   [\omega_{x'} -\gamma^{-1}( E \pm \ii \eps)] \psi^{+}_{x'}\psi^{-}_{x'}} \Big\| \right ]
\\
  \leq  \ee^{\gamma^{-1} \| H_{\Lambda}\|_{\infty,1} \, |\Lambda|}
\left [ \,  
  \prod_{x' \in \Lambda} 
  \ee^{-\eps \gamma^{-1}\phi^{+}_{x'}\phi^{-}_{x'} } \,
\,\mathrm{Tay}_{|\s|}\ee^{|\omega_{x'} -\gamma^{-1}(E \pm \ii \eps)| |\s|} \right ]
\end{multline}
proves the finiteness of the r.h.s.~in \eqref{eq: swapping_condition}, since
\begin{equation}
\begin{split}
& \mathbb{E}_{\omega} \,\prod_{x' \in \Lambda} 
\mathrm{Tay}_{|\s|}\ee^{|\omega_{x'} -\gamma^{-1}(E \pm \ii \eps)| |\s|} < \infty \;,
\\
& \int \dd \phi _{\Lambda} \prod_{x' \in \Lambda} 
\ee^{-\eps \gamma^{-1}\phi^{+}_{x'}\phi^{-}_{x'} } = (\gamma \eps ^{-1})^{|\Lambda| \,|\s|} \;,
\end{split}
\end{equation}
respectively by assumption (H1) and by explicit computation. The claim in \textit{(i)} then follows because 
\begin{equation}
\begin{split}
\mathbb{E}_{\omega} \prod_{x' \in \Lambda}\ee^{\mp \ii \omega_{x'} \Phi^{+}_{x'}\Phi_{x'}^{-}} &=\prod_{x'} \hat{\nu}(\Phi^{+}_{x'}\Phi_{x'}^{-}) 
 \;, 
\end{split}
\end{equation}
where we used that $\{ \omega_{x}\}_{x \in \Lambda}$ are i.i.d.~and  $\nu$ is even. 

We notice that the function $\Phi \mapsto \mu ^{\pm}_{\Lambda}(\Phi)\ee^{\gamma^{-1}(\pm \ii E - \eps) \sum_{x \in \Lambda} \Phi_{x}^{+}\Phi_{x}^{-}} $ belongs to $ \mathscr{S}(\mathcal{S}^{\Lambda},\mathscr{G}^{\Lambda})$ if $\eps>0$ and if respectively $\pm\im E \geq 0$. Furthermore, the Grassmann coefficients of the function $E \mapsto \ee^{\gamma^{-1}(\pm \ii E - \eps) \sum_{x \in \Lambda} \Phi_{x}^{+}\Phi_{x}^{-}} $ are holomorphic. Since $\psi^{-}_{x}\psi^{+}_{y}\,F_{0}^{\Lambda} \in L^{1}(\mathcal{S}^{\XX},\mathscr{G}^{\XX})$, by standard application of dominated convergence and Morera's theorem the functions $\G_{\Lambda}(x,y;E + \ii \eps)$ and $\G_{\Lambda}(x,y;E - \ii \eps)$ are analytic respectively in $ \im E > 0$ and $\im E < 0$. 

Finally, to prove claim \textit{(ii)}, we apply super Plancherel identity (see Proposition~\ref{prop: super_Planchere_identity}) to the r.h.s.~of \eqref{eq: SUSY_greens_function_direct}, and use the fact that
\begin{equation}
\frac{\partial}{\partial\eta^{+}_{x}} \frac{\partial}{\partial\eta^{-}_{y}} \widehat{F_{0}}^{\Lambda}(\xi) = 
\int \dd \Phi_{\Lambda} \ee^{-\ii \sum_{x' \in \Lambda} \big(\xi^{+}_{x'}\Phi^{-}_{x'} + \Phi^{+}_{x'}\xi^{-}_{x'} \big) } \, \psi^{-}_{x} \, \psi^{+}_{y} F_{0}^{\Lambda}(\Phi) \;.
\end{equation}
\end{proof}
\begin{remark}
\label{remark_Disertori}
Our proof of the SUSY representation relies on $\nu$ having finite moments (H1). This hypothesis can be weakened by use of, e.g., supersymmetric polar coordinates \cite{DisertoriLager}. In such a case, the SUSY representation could possibly involve a more complicated expression than $\hat{\nu}(\Phi^{+}\Phi^{-})$.
\end{remark}

With stronger assumptions on $\nu$, we can analytically extend the SUSY integral and hence $ \G_{\Lambda}(x,y; E \pm \ii \epsilon)$ in the variable $E$.
\begin{lemma}[Analytic continuation]
\label{lemma: analytic_continuation_apriori}
Let $\Lambda \subset \mathbb{Z}^{\mathrm{D}}$ be finite and $\eps \geq 0$.
If for some $\beta \geq 0$ $F_{\beta} \in L^{1}(\mathcal{S},\mathscr{G})$, then $ \G_{\Lambda}(x,y; E  +\ii \epsilon )$ and $ \G_{\Lambda}(x,y; E  - \ii \epsilon )$ can be continued to functions that are analytic respectively in $ \im E > -\beta$ and $\im E  < \beta$ and continuous respectively in $ \im E \geq -\beta$ and $\im E  \leq \beta$.
\end{lemma}
\begin{proof}
To begin, we notice that for any $E \in \mathbb{C}$ such that $0\leq \mp \im E \leq \beta$, we respectively have
\begin{equation}
\label{eq: bound_for_analytic_extension}
\begin{split}
\big\| F_{\pm \ii E - \eps}(\Phi)\big\| 
& = \big \| \ee^{\gamma^{-1}(\pm \ii \mathrm{Re} E -(\eps+\beta\pm \im E)) \Phi^{+}\Phi^{-}} F_{\beta}(\Phi)\big \|
\\
& \leq  \ee^{-\gamma^{-1}(\eps+\beta\pm \im E) \phi^{+}\phi^{-}} \big \| \ee^{\gamma^{-1}(\pm \ii  E -(\eps+\beta)) \psi^{+}\psi^{-}}\big \| 
\big \| F_{\beta}(\Phi)\big \|
\\
& \leq \ee^{\gamma^{-1}|\pm \ii E - (\eps+\beta)|\,|\s|}   \big \| F_{\beta}(\Phi)\big \| \;,
\end{split}
\end{equation}
where we used that $\eps + \beta \pm \im E \geq 0$ respectively.
Following \eqref{fub_ton_chap2}, we need to control in absolute value the bosonic integral in the direct SUSY integral, that is, we want to prove the finiteness of
\begin{equation}
\label{intermediate_bosonic_bound}
\sup _{\sigma,\sigma'}\int \dd \phi_{\Lambda} \, \Big | \int \dd \psi_{\Lambda} \, \mu _{\Lambda}^{\pm}(\Phi) \, \psi^{-}_{x,\sigma} \psi^{+}_{y,\sigma'} F^{\Lambda}_{\pm \ii E- \eps}(\Phi) \Big | \;.
\end{equation}
We do it by exploiting \eqref{L1_bound_norms}, thus:
\begin{multline}
\label{eq: chain_analyticity_lemma2}
\eqref{intermediate_bosonic_bound} \leq 
\int \dd \phi_{\Lambda} \big \|  \mu _{\Lambda}^{\pm}(\Phi)\big\|\, \big\|   \psi^{-}_{x,\sigma} \psi^{+}_{y,\sigma'} \big \| \prod_{x' \in \Lambda} \big\| F_{\pm \ii E - \eps}(\Phi_{x'}) \big \|
\\
\leq \big(
\ee^{\gamma^{-1} \|H_{\Lambda}\|_{\infty,1} } \ee^{ \gamma^{-1}|\pm \ii  E - (\eps+\beta)|\,|\s|}\big)^{ |\Lambda|}  \big(\| F_{\beta}\|_{L^{1}(\mathcal{S},\mathscr{G})} \big)^{|\Lambda|} \;,
\end{multline}
where we also used the properties of the Grassmann norms together with the bounds \eqref{eq: bound_gibbs_free} and \eqref{eq: bound_for_analytic_extension}. Thus, at finite $\Lambda$ and finite $\eps \geq 0$ the first line in \eqref{intermediate_bosonic_bound} is finite, uniformly on compacts respectively in the strips $0 < \mp \im \,E <\beta$. The claims follow again by application of 
dominated convergence theorem and Morera's theorem.
\end{proof}
We will henceforth assume that $\G_{\Lambda}(x,y; E \pm \ii \epsilon)$ is the analytic extension if possible.
\medskip

The main technical difficulty in the analysis of the SUSY integrals in Eq.~\eqref{eq: SUSY_greens_function_direct} and Eq.~\eqref{eq: SUSY_greens_function_fourier} is to obtain estimates that are uniform in the volume.
We will achieve this goal by means of SUSY cluster expansions, which we discuss in detail in the next two sections.

%\textcolor{red}{Maybe put comments about previous work in the introduction}.
%We remark that SUSY cluster expansions have already been considered in [Klein] to study the direct SUSY integral. However, although not in their full generality (that is, with Brydges-Federbush tree expansion).

\section{SUSY cluster expansion at strong disorder}
\label{sec: SUSY_cluster_expansion_strong}

In this section, we prove the exponential decay of the disorder-averaged Green's function (Theorem~\ref{thm: exponential_decay_strong_disorder}) and we establish analiticity of the LDOS (Corollary~\ref{thm: DOS_strong_disorder}) at strong disorder and arbitrary energies.
The analysis of the direct SUSY integral is based on the SUSY cluster expansion presented in Proposition~\ref{thm: cluster_expansion_green_function_strong}. The proof of the theorem is then completed by means of tree estimates together with some suitable bounds on the norm of the superfunctions to be integrated. The latter bounds can be achieved under some reasonable assumptions on the disorder distribution that we anticipated in (H2), see the Introduction. In order to make the assumption precise, we need the following definition.
\begin{definition}[Integrable Multiplicative Bounds]
\label{def: bounds_strong_disorder}
Let $f= f(\Phi)$ be a function of a supervector $\Phi = (\phi,\psi) \in \mathcal{S}$. 
We say that $f$ satisfies integrable multiplicative bounds (IMB) if for some $K \geq 0$, $  M \geq 1$ and $p \geq 0$ the following holds true:
\begin{equation}
\int \dd \phi \, \Big\| \Big(\prod_{\varepsilon = \pm} \prod_{ \sigma \in \s} \big(\phi^{\varepsilon}_{\sigma} \big)^{n^{\varepsilon}_{\sigma}}  \Big)f(\Phi) \Big\|\leq K \, M^{n} \, (n!)^{p}  \;,
\end{equation}
for all $ \{n^{\varepsilon}_{\sigma} \} \in \mathbb{N}^{\s \times \{\pm\}}$, having set $n := \sum_{\varepsilon, \sigma}n^{\varepsilon} _{\sigma}$.
\end{definition}

Throughout this section we will assume the following:
\medskip \medskip 

\noindent
$\,$ (H2-$\mathrm{I}$) $ \,$Let $\beta \geq 0$. The superfunction 
$ F_{\beta}(\Phi) = \ee^{\gamma^{-1}\beta \Phi^{+}\Phi^{-}} \hat{\nu}(\Phi^{+}\Phi^{-})$ satisfies \textcolor{white}{wwwwl} IMB for some $K,M$ and $p$.
\medskip \medskip

We believe that our analysis could be extended to the case of weakly positively correlated disorder if we make assumptions on the decay of $\nu$ that are stronger than (H2-$\mathrm{I}$), but still applicable to Gaussian disorder.

We can now state the main results of this section.

%\textcolor{red}{If $H$ is the laplacian there is another version of the theorem with weaker assumptions. Maybe put as a remark.}

\begin{theorem}\label{thm: exponential_decay_strong_disorder}
Let $E \in \mathbb{R}$, $\theta \in [0,1)$ and set $z_{\pm} = E \mp \ii \beta$.
Assume that the matrix elements of the Hamiltonian decay as
\begin{equation}
\label{eq: decay_H}
\sum _{\sigma,\sigma' \in \s}\big | (H_{x,y})_{\sigma,\sigma'}\big| \leq \mathcal{C} \ee ^{-\alpha | x - y |}
\end{equation}
for some $\mathcal{C},\alpha >0$.
There exists a constant $C_{K,M,p,\theta} = C_{K,M,p,\theta}(E) >0$, depending also on $\mathcal{C}$ and $\alpha$, such that if $\gamma \geq C_{K,M,p,\theta} $ then:
\begin{equation}
\label{eq: decay_Greens_function_strong}
\sup _{\sigma,\sigma'}\Big| \Big(\G_{\Lambda}(x,y;z_{\pm}  \pm \ii \eps) \Big)_{\sigma,\sigma'} \Big|\leq \gamma^{-2 + \delta_{x,y}} \, (C_{K,M,p,\theta})^{2-\delta_{x,y}} \,\ee^{-\theta \alpha |x-y|} \;,
\end{equation}
uniformly in $\Lambda \subset \mathbb{Z}^{\mathrm{D}}$ and $0 \leq \eps \leq 1$.
The constant $C_{K,M,p,\theta}$ grows as $|E|^{|\s|/1+|\s|}$ at large $|E|$.
\end{theorem}
\begin{remark}
In particular, the theorem holds for $E \in \sigma(H)$ and $\eps = 0$, which is the most interesting range of parameters.
As soon as $E \pm \ii \eps$ is sufficiently away from the unperturbed spectrum, the result in Theorem \ref{thm: exp_decay_weak} is better suited and gives uniform estimates as $|E|$ or $\eps$ become large.
\end{remark}
\begin{corollary}
\label{thm: DOS_strong_disorder}
Let $\mathcal{B} \subset \mathbb{R}$ bounded. Assume that the matrix elements of the Hamiltonian decay as \eqref{eq: decay_H} for some $\mathcal{C},\alpha >0$. There exists a constant $C_{K,M,p} = C_{K,M,p}(\mathcal{B})>0$, depending also on $\mathcal{C}$ and $\alpha$, such that if $\gamma \geq C_{K,M,p} $ then $\rho(E)$ can be extended to an analytic function in $(\re E, \im E) \in \mathcal{B} \times (-\beta,\beta)$.
\end{corollary}
\begin{remark}
Of course, this statement is meaningful if $\beta > 0$.
\end{remark}
\begin{proof}
Our proof is a rephrasing of the one provided in \cite{BovierKleinPerez}. For the sake of brevity, set $\G_{\Lambda}(z) := \G_{\Lambda}(0,0;z) \in \mathbb{C}^{\s \times \s}$.
The objective is to prove the analyticity of
\begin{equation}
\lim _{\eps \to 0^{+}} \lim _{\Lambda \nearrow \mathbb{Z}^{\mathrm{D}}} \Big(\mathrm{Tr}_{\s}\G_{\Lambda}(z+\ii \eps) -\mathrm{Tr}_{\s}\G_{\Lambda}(z -\ii \eps)
 \Big)
\end{equation}
for $z \in \mathcal{B} \times (-\beta,\beta) \subset \mathbb{C} $.
Since $F_{\beta}$ satisfies IMB, the hypothesis of Lemma~\ref{lemma: analytic_continuation_apriori} is satisfied and hence for any fixed $\Lambda \subset \mathbb{Z}^{\mathrm{D}}$ the function $\mathrm{Tr}_{\s}\G_{\Lambda}(z) $ is analytic on $\in\mathcal{B} \times (-\beta,\beta)$ and continuous on $\in\mathcal{B} \times [-\beta,\beta]$.
The hypotheses of Theorem~\ref{thm: exponential_decay_strong_disorder} are also satisfied: we use the theorem with $\theta = 0$, and with $\re z \in \mathcal{B}$. Define $C_{K,M,p} := \sup_{E \in \mathcal{B}} C_{K,M,p, \theta = 0}(E)$, $C_{K,M,p, \theta = 0}(E)$ being the constant in Theorem~\ref{thm: exponential_decay_strong_disorder}. As a consequence of the theorem, if $\gamma \geq C_{K,M,p} $ then uniformly in $z \in \mathcal{B} \times (-\beta,\beta)$, $\eps \geq 0$
and $\Lambda \subset \mathbb{Z}^{\mathrm{D}}$, the quantities $\mathrm{Tr}_{\s}\G_{\Lambda}(z + \ii \eps) $ and $\mathrm{Tr}_{\s}\G_{\Lambda}(z- \ii \eps) $ are bounded. Since by the second resolvent identity the limit
\begin{equation}
\lim _{\Lambda \nearrow \mathbb{Z}^{\mathrm{D}}} \mathrm{Tr}_{\s}\G_{\Lambda}(z\pm \ii \eps) = \mathrm{Tr}_{\s}\G_{\mathbb{Z}^{\mathrm{D}}}(z\pm \ii \eps)
\end{equation}
exists for any $\eps >0$, we can apply Vitali-Porter theorem and obtain that the convergence is uniform with $\mathrm{Tr}_{\s}\G_{\mathbb{Z}^{\mathrm{D}}}(z\pm \ii \eps)$ analytic in $z \in \mathcal{B} \times (-\beta,\beta)$, $\eps >0$ and continuous in $\eps \to 0^{+}$. Therefore, both
\begin{equation}
\lim _{\eps \to 0^{+}} \mathrm{Tr}_{\s}\G_{\mathbb{Z}^{\mathrm{D}}}(z+ \ii \eps) \;
\qquad
\lim _{\eps \to 0^{+}} \mathrm{Tr}_{\s}\G_{\mathbb{Z}^{\mathrm{D}}}(z- \ii \eps) 
\end{equation}
exist and are analytic in $z \in \mathcal{B} \times (-\beta,\beta)$, hence the claim.
\end{proof}
\begin{remark}
In \cite{SjostrandWang1} they also prove directly that the limit
\begin{equation*}
\lim _{\eps \to 0^{+}} \lim _{\Lambda\nearrow\mathbb{Z}^{\mathrm{D}}} \G_{\Lambda}(x,y;E \pm \ii \eps)
\end{equation*}
exists provided that either $\gamma$ or $E$ is large. We point out that the expansion presented in Proposition~\ref{thm: cluster_expansion_green_function_strong} and the estimates in the proof of Theorem~\ref{thm: exponential_decay_strong_disorder}
allow us to prove that uniformly in $0 \leq \eps \leq 1$, $\big\{\big(\G_{\Lambda}(x,y;E \pm \ii \eps)\big)_{\sigma,\sigma'}\big\}_{\Lambda}$ is a Cauchy sequence in $\Lambda$, at any fixed $x,y \in \mathbb{Z}^{\mathrm{D}}$, $\sigma,\sigma' \in \s$ and $E \in \mathbb{R}$. 
\end{remark}

The SUSY cluster expansion presented in Proposition~\ref{thm: cluster_expansion_green_function_strong} is crucially based on the following important result in statistical mechanics.
\begin{theorem}[Battle-Brydges-Federbush representation]
\label{thm: FBB_representation}
For any $\XX \subset \Lambda$ let $V_{\XX}:= 1/2\sum_{x,y \in \XX} v_{x,y}$, where $v_{x,y} = v_{y,x}$ is an even element of a Grassmann algebra.
The following representation holds true:
\begin{equation}
\begin{split}
& \ee^{V_{\XX}} = \sum_{\Pi \in \mathrm{Part}(\XX)} \,  \,\prod_{\YY \in \Pi} \, K(\YY) \;,
\end{split}
\end{equation}
where the sum is over the partitions of $\XX$ and where $K(\YY) = \ee^{V_{\YY}}$ if $|\YY| = 1$, otherwise
\begin{equation}
\label{eq: BBF_polymer_formula}
K(\YY) = \sum_{\textsc{T} \mathrm{\,on\,}
 \YY} \, \bigg(\prod_{\{x,y\} \in \textsc{T}} \,v_{x, y} \bigg)\,\int \mathrm{d}p_{\textsc{t}}(s) \, \ee^{V_{\YY}(s)}\;.
\end{equation}
For $s = (s_{x,y}) \in [0,1]^{\mathcal{P}(\YY)}$, $ \mathcal{P}(\YY)$ being the set of unordered pairs in $\YY$, we have defined
\begin{equation}
V_{\YY}(s):= \frac{1}{2} \sum_{x,y \in \YY} s_{x,y} \, v_{x,y} \;,
\end{equation}
and denoted by $\dd p_{\textsc{t}}$ a probability measure supported on $s \in [0,1]^{\mathcal{P}(\YY)}$ such that $V_{\YY}(s)$ is a convex decoupling of $V_{\YY}$, that is, it is a convex linear combination of quantities of the form $\sum_{\yy' \in \Pi} V_{\yy'}$, with $\Pi$  being a partition of $\YY$.
\end{theorem}
The proof of this statement can be found in Appendix B of \cite{Brydges84}. See also \cite{BattleFederbush,BrydgesKennedy,AbdesselamRivasseau}
for other proofs of the Battle-Brydges-Federbush formula or generalizations of it.
Technically speaking, the proof in \cite{Brydges84} is stated for the case of $v_{x,y} \in \mathbb{C}$, however, as we briefly show below, the strategy can be applied to the case in which $v_{x,y}$ is an even element of a Grassmann algebra. 

In the proof, one iteratively decouples clusters of increasing size from the whole $\XX$ by means of an interpolation formula, see \cite{Brydges84}. As it turns out, we only need to check the validity of such interpolation formula in our setting, the rest of the proof being of combinatoric nature. For example, we discuss the first decoupling step. Let $x_{1} \in \XX$ and for $s_{1} \in [0,1]$ define 
\begin{equation}
W(\{x_{1}\};s_{1}):= \frac{1}{2}\sum_{x,y \in \XX} s_{1}(\{x,y\}) \, v_{x,y}  \;,
\end{equation}
together with
\begin{align}
s_{1}(\{x,y\})  :=  \begin{cases} s_{1}  &\qquad \text{if $\{ x,y\}$ couples across $\partial \{x_{1}\}$} \;,
\\
1  & \qquad\text{otherwise} \;,
\end{cases}
\end{align}
where $\{ x,y\}$ is said to couple across $\partial \YY$, $\YY \subset \Lambda$, if $\{x,y\} \cap \YY$ and $\{x,y \} \cap \YY^{c}$ are both non-empty.
Clearly, $W(\{x_{1}\},1) = V_{\XX}$ while $W(\{x_{1}\},0)$ has $x_{1}$ completely decoupled.

The interpolation formula is a simple application of the fundamental theorem of calculus, see formula (B.1) in \cite{Brydges84}:
\begin{equation}
\label{interpolation_formula}
\begin{split}
\ee^{V_{\xx}} & =  \int_{0}^{1} \dd s_{1} \, \frac{\partial}{\partial s_{1}} \ee^{W(\{x_{1}\},s_{1})} +\ee^{W(\{x_{1}\},0)} 
\\ 
& = \sum_{\substack{ \{x,y \} \subset \XX: \\ \{x,y \} \text{ couples across } \partial \{x_{1}\}  }} v_{x,y} \;\int_{0}^{1} \dd s_{1} \,  \ee^{W(\{x_{1}\},s_{1})}  +\ee^{W(\{x_{1}\},0)}  \;.
\end{split}
\end{equation}
The first equality holds true also when the Gibbs' weight $\exp(W(\{x_{1}\},1))$ is Grassmann-valued because its Grassmann coefficients are smooth in $s_{1}$. On the other hand, the second equality is true because $\{v_{x,y}\}_{x,y \in \XX}$ commute, being even element of a Grassmann algebra, and this shows that the proof in \cite{Brydges84} is valid in our setting as well.
In passing, we also notice that such interpolation formula is valid more generally if $\{v_{x,y}\}_{x,y \in \XX}$
belong to a commutative Banach algebra, see also \cite{AbdesselamRivasseau}, where more general formulas are covered.

%See \cite{Brydges84,BattleFederbush} for details and \cite{BrydgesKennedy,AbdesselamRivasseau} for generalizations.
%
%The following theorem is used in order to carry out estimates on tree-graphs.
%%
%\begin{theorem}[Cayley]
%\label{thm: Cayley_number_tree}
%The number of labelled trees with $N$ vertices and incidence numbers $\{d_{i}\}_{i = 1,\dots,N}$ is $(N-2)!/\prod_{i=1}^{N} (d_{i}-1)!$. The number of labelled trees with $N$ vertices is $N^{N-2}$.
%\end{theorem}
%%
%The first formula can be proven by induction over $N$; the second one by noticing that  the number of labelled trees with fixed coordination number is a multinomial coefficient,and summation over $\{d_{i} \}$ gives the claim.
%
\begin{remark}
\label{rmk: BBF_representation}
%We require that $v_{x,y}$ are elements of a Banach algebra in order for the interpolating procedure to be well-defined.
%This is not necessary: for example, in the case of Grassmann algebra, nilpotency is already enough for this purpose.
%
We say that $V_{\XX}$ is stable if the real part of $(V_{\XX})_{\emptyset}$ is non-positive. If $V_{\XX}$ is stable also its convex decouplings are stable: this property is crucial in order to have well-defined integrals.
All the cases we consider below have this stability property provided that $\epsilon \geq 0$. For example, $V_{\XX} = \ii \sum_{x,y \in \XX} \Phi^{+}_{x} \big( H_{x,y} + (E +\ii \eps) \delta_{x,y} \big) \Phi^{-}_{y}$ 
and $\hat{V}_{\XX} = \ii \sum_{x,y \in \XX}  \xi_{x}^{+} C_{x,y}^{E +\ii \eps}\xi_{y}^{-}$
with  $\eps \geq 0$ are stable, in fact, $\re (V_{\XX})_{\emptyset} = \re \ii \sum_{x,y \in \XX} \phi^{+}_{x} \big( H_{x,y} + (E+\ii\eps) \delta_{x,y} \big) \phi^{-}_{y} \leq 0$ and $\re (\hat{V}_{\XX})_{\emptyset} = \re \ii \sum_{x,y \in \XX} \kappa^{+}_{x} C^{E+\ii \eps}_{x,y} \kappa^{-}_{y} \leq 0$. Incidentally, notice that this particular $V_{\XX}(s)$  is stable for any $s \in [0,1]^{\mathcal{P}(\YY)}$.
%
%It is crucial that $\dd p_{\textsc{t}}$ is supported on $s$ such that $V_{\YY}(s)$ is a convex decoupling of $V_{\YY}$. In fact, this feature allows to preserve ``stability properties'': if $V_{\YY} \leq 0$ then also its convex decouplings are non-positive, fact that will be used for granting the convergence of the integrals.
\end{remark}
\begin{remark}
With abuse of notation we shall introduce the ``empty tree'' and denote it by $\emptyset$, as the only tree-graph on a set $\YY$ with $|\YY| = 1$. Furthermore, $\int \dd p_{\textsc{t} = \emptyset}(s)$ is a complicated way to denote the multiplication identity.
Thus, if $|\YY|=1$ we define
\begin{equation}
\label{eq: empty_tree}
\ee^{V_{\YY}} =: \sum_{\textsc{T} \mathrm{\,on\,}
 \YY} \, \bigg(\prod_{\{x,y\} \in \textsc{T}} \,v_{x, y} \bigg)\,\int \mathrm{d}p_{\textsc{t}}(s) \, \ee^{V_{\YY}(s)}\;.
\end{equation}
\end{remark}

We apply the BBF representation to the super Gibbs' weights $\mu^{\pm} _{\Lambda}(\Phi)$.  
The expansion is well-suited for our purposes because it exploits the smallness of $\gamma^{-1}$: it formalises the fact that $\mu^{\pm} _{\Lambda}(\Phi)$ is somewhat close to one.
The expansion can also be considered as an improvement of the simple Taylor expansion of $\mu^{\pm} _{\Lambda}(\Phi)$.
\begin{proposition}
\label{thm: cluster_expansion_green_function_strong}
Let $E\in \mathbb{R}$ and set $z_{\pm} : = E\mp \ii \beta$. The following representation holds true for any $\eps \geq 0$:
\begin{eqnarray}
\label{eq: cluster_expansion_green_function_strong}
&& \!\!\!\!\!\!\! \G_{\Lambda}(x,y; z_{\pm} \pm \ii \epsilon) = \big(\pm\ii \gamma^{-1}\big)^{2-\delta_{x,y}}\sum_{N = 0}^{|\Lambda'|}
\big(\pm\ii \gamma^{-1}\big)^{N} g_{N}(x,y; z_{\pm} \pm \ii \epsilon)
\\
&& \!\!\!\!\!\!\! g_{N}(x,y; z_{\pm} \pm \ii \epsilon):= \frac{1}{N!}
 \sum_{\substack{\textsc{t} \,\,\, \mathrm{on} \\
 \,\{1,...,N+2 - \delta_{x,y}\}}} \,\,
\sum_{\substack{x_{1},...,x_{N} \in \Lambda' \\ \mathrm{distinct}}}  
\int \dd p_{\textsc{t}}(s) \;
\\ \nonumber
&&
\qquad \qquad \qquad \qquad  \;
\int \dd \mu_{\yy}^{\pm}(\Phi,s) \,
 \Big(\prod_{\{i,j\} \in \textsc{t}} v_{x_{i},x_{j}}(\Phi)\Big) \,\, F^{\yy}_{\pm \ii E- \eps}(\Phi)\, \psi^{-}_{x} \, \psi^{+}_{y}\,, \;
 \label{eq: g_N_definition}
\end{eqnarray}
where we have set $\dd \mu^{\pm}_{\yy}(\Phi,s) : = \dd \Phi_{\yy} \, \mu^{\pm}_{\yy}(\Phi,s)$ and 
\begin{equation}
\label{eq: definition_mu_v_strong_disorder}
\begin{split}
&  \mu^{\pm}_{\yy}(\Phi,s) : = \ee ^{ \pm \frac{\ii}{2 }\gamma^{-1}\sum_{x,y \in \yy} s_{x,y} v_{x,y}(\Phi)} \;,
\\
& v_{x,y}(\Phi): =  - \Phi_{x}^{+}H_{x,y}\,\Phi_{y}^{-}  - \Phi_{y}^{+}H_{y,x}\,\Phi_{x}^{-}\;.
\end{split}
\end{equation}
with $s = (s_{x,y}) \in [0,1]^{\mathcal{P}(\YY)}$ and $\dd p_{\textsc{t}}(s)$ being a probability measure.
%with support on $s$ such that the exponent in the super Gibbs' weights $\mu _{\yy}^{\pm}(\Phi,s)$ satisfies respectively $ \re \pm \ii\sum_{x,y \in \YY}  s_{x,y} v_{x,y}((\phi,0))  \leq 0$. 
Furthermore, we have set $x_{N+1} := x$ and $x_{N+2} := y$, $\sY := \{x_{1},...,x_{N+2} \}$ and $\Lambda' := \Lambda \setminus \{x,y\}$.
\end{proposition}
\begin{remark}
It is important to notice that the expansion of $\G_{\Lambda}$ is in terms of connected graphs, specifically tree-graphs, while usually one would expect an expansion in forests \cite{AbdesselamRivasseau}.
This is ultimately a consequence of supersymmetry which bypasses the need for logarithms.
\end{remark}
\begin{remark}
Notice that only the empty-tree term contributes to the quantity $g_{0}(x,x;z_{\pm}\pm \ii \eps)$, for which definition \eqref{eq: empty_tree} applies.
\end{remark}
\begin{proof}
Since $F_{\beta}$ satisfies IMB, it decays in norm faster than any power of $\phi$. In particular, $F_{\beta} \in L^{1}(\mathcal{S},\mathscr{G})$ and we can make sense of the SUSY integral in Eq.~\eqref{eq: SUSY_greens_function_direct} when $E \to z_{\pm}$, $\eps \geq 0$.
Since $v_{x,y}(\Phi)$ are elements of a Grassmann algebra, we apply Theorem~\ref{thm: FBB_representation} to $\mu^{\pm} _{\Lambda} =: \ee^{V_{\Lambda}^{\pm}}$ and after simple manipulations we obtain the following polymer expansion:
\begin{eqnarray}
&&\int \dd \mu^{\pm} _{\Lambda}(\Phi) F^{\Lambda}_{\pm \ii z_{\pm} - \eps}(\Phi) \psi^{-}_{x} \psi^{+}_{y}
= \sum_{\Pi \in \mathrm{Part}(\Lambda)} \prod_{\YY \in \Pi} 
K(\YY)\;, \label{eq: polymer_expansion_strong_disorder}
\end{eqnarray}
where 
\begin{equation}
\label{eq: K_fomula_proof}
\begin{split}
 K(\YY)  = &\sum_{\textsc{T} \mathrm{\,on\,} \YY} \,
 \, \, \int \dd p_{\textsc{t}}(s) \, \int \dd \mu ^{\pm}_{\yy} (\Phi,s) 
 \\
 & \;\Big(\prod_{\{x',y'\} \in \textsc{T}} 
\pm\ii \gamma^{-1} \,v_{x',y'}(\Phi)\Big) F^{\yy}_{\pm \ii z_{\pm}-\eps}(\Phi) \, \big(\psi^{-}_{x}\big)^{\I_{\yy}(x)} \big(\psi^{+}_{y}\big)^{\I_{\yy}(y)}\;,
\end{split}
\end{equation}
$\I_{\yy}(\cdot)$ denoting the indicator function of $\YY$ and $(\psi^{\varepsilon}_{x,\sigma})^{0} = 1$ by convention. 
The final expression in \eqref{eq: K_fomula_proof} has been obtained by swapping the integration over the superfields $\Phi$ with the integration over $s$, which we can do since uniformly in $s$ the superintegral is bounded. This swap is mainly a matter of taste. 
%In fact, we bound the super Gibbs' weight as follows,
%%
%\begin{equation}
%\label{eq: bound_Gibbs'}
%\sup _{ s} \,\big \| \mu ^{\pm}_{\yy}(\Phi,s)\big \| \leq \ee^{ \gamma ^{-1}\, \| H\|_{\infty,1} \, |\YY |} \;,
%\end{equation}
%%
%therefore, using 
%$F_{\pm \ii z - \eps}(\Phi) = \ee^{\gamma^{-1}(\pm \ii E - \eps) \Phi^{+}\Phi^{-}} F_{\beta}(\Phi)$, the iterated integral
%%
%\begin{equation}
%\label{eq: cluster_expansion_bound_proof}
%\begin{split}
%&\int \dd p_{\textsc{t}}(s) \, \int \dd \phi_{\yy}  \,\Big \| \mu ^{\pm}_{\yy} (\Phi,s) \Big(\prod_{\{x',y'\} \in \textsc{T}} 
%\pm\ii \gamma^{-1} \,v_{x',y'}(\Phi)\Big)
%\qquad
% \\
% & \qquad \qquad \qquad \qquad \qquad
% \qquad \qquad
%  F^{\yy}_{\pm \ii z_{\pm}-\eps}(\Phi) \, \big(\psi^{-}_{x}\big)^{\I_{\yy}(x)} \big(\psi^{+}_{y}\big)^{\I_{\yy}(y)}  \Big \|
%\\
%& \leq \ee^{\gamma^{-1}(\|H \|_{\infty,1}+ |\ii E - \eps||\s|)\,|\YY|} \int \dd \phi_{\yy} \Big \|\Big(\prod_{\{x',y'\} \in \textsc{T}} 
%\pm\ii \gamma^{-1} \,v_{x',y'}(\Phi)\Big) F^{\yy}_{\beta}(\Phi) \Big\|
%\end{split}
%\end{equation}
%%
%is finite because $F_{\beta}$ satisfies IMB.
Notice that by the BBF formula $\dd p_{\textsc{t}}(s)$ is supported on $s$ such that $V_{\XX}^{\pm}(s)$ are convex decouplings of $V_{\XX}^{\pm}$. However, since $V_{\XX}^{\pm}(s)$ are stable for any $s \in [0,1]^{\mathcal{P}(\XX)}$, this information is superfluous in this special case.

We shall now make the expansion in \eqref{eq: polymer_expansion_strong_disorder} simpler by symmetry considerations. To begin, we notice that $\mu ^{\pm}_{\yy} (\cdot,s) $, $v_{x,y}$ and $F^{\yy}_{\pm \ii z_{\pm}-\eps}$ are even and supersymmetric, see also discussion before Proposition~\ref{prop: SUSY_representation}.
Thus, by parity $K(\YY) = 0$ unless $\YY \cap \{x,y\}$ is either $\{x,y \}$ or $\emptyset$. 
On the other hand, if $\YY \cap \{x,y\} = \emptyset$, the superfunction to be integrated with respect to $\int \dd \Phi_{\yy}$ is supersymmetric. If we show that such integrand satisfies the hypotheses of the SUSY localization formula, then $K(\YY) = \delta_{1,|\YY|}$ whenever $\YY \cap \{x,y\} = \emptyset$, since in fact $\mu ^{\pm}_{\yy} (0,s) = F^{\yy}_{\pm \ii z_{\pm}-\eps}(0) =1$ and $v_{x,y}(0) = 0$. Thus, we are left with proving that for any $(x,\sigma,\varepsilon) \in \YYY \times \{ \pm\}$ the superfunction
\begin{equation}
\label{eq: integrand}
\Phi \mapsto \frac{\psi^{\varepsilon}_{x,\sigma}}{(1+|\phi^{\varepsilon}_{x,\sigma}|)}
\frac{\partial}{\partial \phi^{\varepsilon}_{x,\sigma}} \bigg(
\mu ^{\pm}_{\yy} (\Phi,s) \Big(\prod_{\{x',y'\} \in \textsc{T}} 
v_{x',y'}(\Phi)\Big)
  F^{\yy}_{\pm \ii z_{\pm}-\eps}(\Phi)\bigg)
\end{equation}
is in $L^{1}(\mathcal{S}^{\YY},\mathscr{G}^{\YY})$, see Proposition \ref{prop: SUSY_localization_theorem}. When the derivative acts on $\mu ^{\pm}_{\yy} (\cdot,s)$ or $v_{x',y'}$ no dangerous terms are generated and the integrability follows once more because $F_{\beta}$ satisfies IMB. If, instead, the derivative acts on $F_{\pm \ii z_{\pm} - \eps}^{\yy}$, since $F_{z}$ is invariant under $U(1)^{\times \s}$ fermionic transformations, Lemma \ref{lemma: decomposition_Q_operator} implies the identity
\begin{equation}
\frac{\psi^{\varepsilon}_{\sigma}}{(1+|\phi^{\varepsilon}_{\sigma}|)}
\frac{\partial}{\partial \phi^{\varepsilon}_{\sigma}}F_{z}(\Phi) =-\varepsilon 
\frac{\phi^{-\varepsilon}_{\sigma}}{(1+|\phi^{\varepsilon}_{\sigma}|)}
\frac{\partial}{\partial \psi^{-\varepsilon}_{\sigma}} F_{z}(\Phi) \;,
\end{equation}
and hence integrability of \eqref{eq: integrand} because $F_{\beta}$ satisfies IMB.

The previous considerations imply that the sum in Eq.~\eqref{eq: polymer_expansion_strong_disorder} is really just
\begin{equation}\label{eq: trimmed_polymer_expansion_strong_disorder}
\begin{split}
\sum_{\yy \supset \{x,y\}} K(\YY)
&= \sum_{N = 0}^{|\Lambda'|} \; \sum_{\substack{\yy \supset \{x,y\} \\ |\yy| = N}} K(\YY)
\\
 &= \sum_{N = 0}^{|\Lambda'|} \frac{1}{N!}\sum_{\substack{ x_{1},\dots,x_{N} \in \Lambda'\\ \mathrm{distinct}}} K(\{x_{1},\dots,x_{N},x,y \}) \;.
\end{split}
\end{equation}
The $1/N!$ factor is needed because each set $Y$ is counted that many times in the sum over the \textit{distinct} points $x_{1},\dots,x_{N}$.
We then set $x_{N+1} := x$, $x_{N+2} := y$ and swap the sum over distinct points with the sum over trees on $\YY$, the latter becoming a sum over trees on $\{1,\dots,N+2 - \delta_{x,y} \}$. Notice that if $x = y$, then the set $\YY$ contains only $N+1$ points.
\end{proof}
Before delving into the proof of Theorem \ref{thm: exponential_decay_strong_disorder}, we want to make a general remark about the strategy. In interacting fermionic systems determinant bounds provide a useful tool to control the convergence of the perturbative expansion \cite{GawedzkiKupiainen,Lesniewski,
GentileMastropietro,Salmhofer,Mastropietro}.
We avoid this type of bounds and instead follow a different approach, based on a combination of the BBF formula (through Proposition~\ref{thm: cluster_expansion_green_function_strong}) and Grassmann norms. 
%In a sense, our approach is taylored to bosons and 
%

\subsection*{Proof of Theorem \ref{thm: exponential_decay_strong_disorder}}
\label{sec:proof_theorem}
Since $F_{\beta}$ satisfies IMB we can apply Proposition~\ref{thm: cluster_expansion_green_function_strong}.
We will prove that for some constant $\overline{C}_{K,M,p,\theta} >1$ not depending on $E$ the following bound holds true:
\begin{multline}
\label{ineq: g_N_bound}
\sup _{\sigma,\sigma'}\big| \big(g_{N}(x,y; z_{\pm} \pm \ii \eps)\big)_{\sigma,\sigma'} \big| 
\\
\leq \big( 1+ \gamma^{-|\s|}(1+|E|)^{|\s|} \big)^{|\YY|}  \, (\overline{C}_{K,M,p,\theta})^{| \YY|} \,\alpha^{-\mathrm{D} \,N } \,\ee^{-\theta \alpha |x - y|} \;.
\end{multline}
By plugging this bound into the expansion of Proposition~\ref{thm: cluster_expansion_green_function_strong}
we see that if
\begin{equation}
\gamma \geq 4 \,\overline{C}_{K,M,p,\theta}\,\big(\alpha^{-\mathrm{D}} + \alpha^{-\mathrm{D}/1+|\s|} \big)  (1+|E|)^{|\s|/1+|\s|}
\end{equation}
the expansion is convergent and the claim of the theorem ensues. It appears that the constant $C_{K,M,p,\theta}$ grows like $|E|^{|\s|/1+|\s|}$ at large $|E|$.

To prove \eqref{ineq: g_N_bound} we proceed as follows. Let us consider a polymer $\YY \ni x,y$ and an oriented tree $\textsc{T}_{\wp}$ on $\YY$, that is, a tree-graph whose links are oriented, $\wp$ denoting the choice of the orientations. We denote by $\ell^{+}$ ($\ell^{-}$) the starting (ending) vertex of the oriented link $\ell \in \textsc{T}_{\wp}$. Links have to be oriented in order to select one of the two elements in $v_{\ell}(\Phi)$, see in Eq.~\eqref{eq: definition_mu_v_strong_disorder}. In other words, in Eq.~\eqref{eq: g_N_definition} we want to sum over oriented trees instead of trees and for each link keep one of the two terms in $v_{\ell}(\Phi)$ depending on the orientation $\wp$.
We thus introduce the sequences $\underline{\sigma} = \big\{ \sigma^{\varepsilon}_{\ell} \in \s \big\}_{\ell \in \textsc{T},\,\varepsilon = \pm}$ and $\underline{\sharp} = \big\{ \sharp_{\ell} \in \{B,F\} \big\}_{\ell \in \textsc{T}}$. Also, we set $\Phi^{\varepsilon}_{B,x,\sigma} = \phi^{\varepsilon}_{x,\sigma}$ and $\Phi^{\varepsilon}_{F,x,\sigma} = \psi^{\varepsilon}_{x,\sigma}$, and we shall henceforth write $|\YY|$ instead of $N+2 - \delta_{x,y} = |\YY|$. We define:
\begin{equation}
\begin{split}
\mathcal{F}_{\textsc{t}_{\wp},\,\underline{\sharp}, \,\underline{\sigma}}^{\sy}(\Phi) &:=
\Big ( \prod_{\ell \in \mathrm{T}_{\wp}} \,\Phi^{+}_{\sharp_{\ell},x_{\ell^{+}},\sigma_{\ell}^{+}} \, \Phi^{-}_{\sharp_{\ell}, \,x_{\ell^{-}},\sigma_{\ell}^{-}} \Big )
\, \psi^{-}_{x} \psi^{+}_{y} \,F_{\pm \ii z_{\pm} - \eps}^{\yy}(\Phi) \;.
\end{split}
\end{equation}
All in all, we can rewrite Eq.~\eqref{eq: g_N_definition} as
\begin{multline}
\label{eq: polarised_tree_expansion_green_function_strong}
g_{N}(x,y;z_{\pm} \pm \ii \eps) =  \frac{1}{N!}\sum_{\substack{\textsc{t} \,\,\, \mathrm{on} \\
 \,\{1,...,|\yy|\}}} \,\sum_{\wp, \, \underline{\sharp}, \,\underline{\sigma}}\,\,
\sum_{\substack{x_{1},...,x_{N} \in \Lambda' \\ \mathrm{distinct}}}  
\,
\int \dd p_{\textsc{t}}(s) \; \cdot
\\
\cdot \; \int \dd \mu^{\pm}_{\yy}(\Phi,s) \,
\mathcal{F}_{\textsc{t}_{\wp},\,\underline{\sharp}, \,\underline{\sigma}}^{\yy}(\Phi)  \,\Big (\prod_{\ell \in \textsc{T}_{\wp}}
\big(H_{x_{\ell^{+}},x_{\ell^{-}}}\big)_{\sigma_{\ell}^{+},\sigma_{\ell}^{-}}\Big ) \;,
\end{multline}
which can be bounded as follows:
\begin{multline}
\sup _{\sigma,\sigma'}\big| \big(g_{N}(x,y;z_{\pm} \pm \ii \eps) \big)_{\sigma,\sigma'}\big| \leq  
  \\ 
\frac{4^{| \YY |}}{N!}\sum_{\substack{\textsc{t} \,\,\, \mathrm{on} \\  \,\{1,...,|\yy | \}}} \,
  \sup _{x_{1}, \dots, x_{N}} 
 \sup _{\substack{\wp, \, \underline{\sharp}, \,\underline{\sigma},\, s \\ \sigma,\sigma'}} \Big|
\int \dd \mu^{\pm}_{\yy}(\Phi,s) \,
\big(\mathcal{F}_{\textsc{t}_{\wp},\,\underline{\sharp}, \,\underline{\sigma}}^{\yy}(\Phi) \big)_{\sigma,\sigma'} \Big|   \; 
\\
\sup _{\wp} \,\sum_{\underline{\sigma}} \sum_{\substack{x_{1},...,x_{N} \in \Lambda' \\ \mathrm{distinct}}} \,  
\prod_{\ell \in \textsc{T}_{\wp}}
\big|\big( H_{x_{\ell^{+}},x_{\ell^{-}}}\big)_{\hat{\sigma}_{\ell}^{+},\hat{\sigma}_{\ell}^{-}}\big|  \,.
\label{eq: partial_bound_g_N}
\end{multline}
%
%We shall first obtain a useful bound for the superintegral by using the properties of the Grassmann norms:
To bound the superintegral we exploit the properties of the Grassmann norms and obtain
\begin{equation}
\begin{split}
\Big|
\int \dd \mu^{\pm}_{\yy}(\Phi,s) \,
\big(\mathcal{F}_{\textsc{t}_{\wp},\,\underline{\sharp}, \,\underline{\sigma}}^{\yy}(\Phi) \big)_{\sigma,\sigma'} \Big| & \leq
\int \dd \phi_{\yy} \big \| \mu ^{\pm}_{\yy}(\Phi,s)\big \| \, \big \| \big(\mathcal{F}_{\textsc{t}_{\wp},\,\underline{\sharp}, \,\underline{\sigma}}^{\yy}(\Phi) \big)_{\sigma,\sigma'}\big \|
\\
& \leq \ee^{\gamma^{-1} \| H\|_{\infty,1} |\YY|} \int \dd \phi_{\yy} \,\big \| \big(\mathcal{F}_{\textsc{t}_{\wp},\,\underline{\sharp}, \,\underline{\sigma}}^{\yy}(\Phi) \big)_{\sigma,\sigma'}\big \|
\;,
\end{split}
\end{equation}
where in the last step we used the same bound as in Eq.~\eqref{eq: bound_gibbs_free}, noticing that $s \in [0,1]^{\mathcal{P}(\YY)}$ and that $\| H_{\Lambda} \|_{\infty,1} \leq \| H \|_{\infty,1} < \infty$, the latter being bounded because of \eqref{eq: decay_H}.

It is clear that 
$\mathcal{F}_{\textsc{t}_{\wp},\,\underline{\sharp}, \,\underline{\sigma}}^{\yy}(\Phi) $ is a local function, i.e., it factorizes 
\begin{equation}
\label{eq: factorization_F_expansion_strong}
\mathcal{F}_{\textsc{t}_{\wp},\,\underline{\sharp}, \,\underline{\sigma}}^{\YY}(\Phi) = \pm \prod_{i = 1}^{|\YY|} \, \mathcal{F}_{x,y}^{\, \underline{d_{i}}}(\Phi_{x_{i}}) \;,
\end{equation}
having set:
\begin{equation}
\label{eq: mathcal_F_local_function}
\mathcal{F}_{x,y}^{\,\underline{d}}(\Phi_{x'}):=
\Big (
\prod_{\varepsilon, \sharp, \sigma}  \big (\Phi_{\sharp,x',\sigma}^{\varepsilon}\big )^{d_{\sharp,\sigma}^{\varepsilon}}
\Big )
\,\left (\psi_{x}^{-}\right )^{\delta_{x',x}} \left (\psi_{y}^{+}\right )^{\delta_{x',y}} F_{\pm \ii z_{\pm} - \eps}(\Phi_{x'}) \;,
\end{equation}
together with the sequence $\underline{d}_{i}= \{(d_{i})_{\sharp,\sigma}^{\varepsilon}\}_{\varepsilon = \pm, \,\sigma \in \s,\,\sharp = B,F} $, where $(d_{i})_{\sharp,\sigma}^{\varepsilon}$ is defined by 
\begin{equation}
(d_{i})_{\sharp,\sigma}^{\varepsilon} :=
\sum_{\ell \in \textsc{t}_{\wp}} \delta_{i,\ell^{\eps} } \delta_{\sharp,\sharp_{\ell}} \delta_{\sigma, \sigma_{\ell}^{\eps}} \;,
\end{equation}
that is, given $\textsc{T}_{\wp}$, $\underline{\sigma}$ and $\underline{\sharp}$, it counts how many times $ (\ell^{\varepsilon}, \sharp_{\ell},\sigma _{\ell}^{\varepsilon}) = (i,\sharp,\sigma)$.
Notice that the sign in front of the r.h.s.~in Eq.~\eqref{eq: factorization_F_expansion_strong} is unimportant and thus left unspecified.
For the sake of notation, we also define 
$d_{i,\sharp} := \sum_{\varepsilon,\sigma} (d_{i})_{\sharp,\sigma}^{\varepsilon}$ together with $d_{i}: = \sum_{\sharp} d_{i,\sharp}$.
Since $F_{\pm \ii z_{\pm} - \eps}(\Phi) = \ee^{\gamma^{-1}(\pm \ii E - \eps) \Phi^{+}\Phi^{-}} F_{\beta}(\Phi)$ we have, compare with \eqref{eq: a_priori_omega_integration_bound}:
\begin{equation}
\big \| \big(\mathcal{F}_{x,y}^{\,\underline{d}}(\Phi_{x'})\big)_{\sigma,\sigma'}\big \|
\leq  
\Big(\mathrm{Tay}_{|\s|}
\ee^{\gamma^{-1}|\ii E - \eps|\, |\s|} \Big)\Big \| \Big (
\prod_{\varepsilon,\sigma}  \big (\phi_{x',\sigma}^{\varepsilon}\big )^{d_{B,\sigma}^{\varepsilon}}
\Big ) F_{\beta}(\Phi)\Big \| \;.
\end{equation}
%
%See \eqref{eq: a_priori_omega_integration_bound} \textcolor{green}{Sostituire con $\mathrm{Tay}_{|\s|}$ per rendere la stima migliore in E.}
Using that $F_{\beta}$ satisfies IMB and that $\eps \leq 1$ we finally obtain
\begin{multline}
\!\!\!\!\!\sup _{\substack{\wp, \,\underline{\sharp}, \, \underline{\sigma}, \,s \\  \sigma,\sigma' }}\Big|\int \dd \mu^{\pm}_{\yy}(\Phi,s) \,
\big(\mathcal{F}_{\textsc{t}_{\wp},\,\underline{\hat{\sigma}}}^{\yy}(\Phi)\big)_{\sigma,\sigma'} \Big| \leq
\ee^{ \gamma^{-1}\, \| H\|_{\infty,1} \, |\YY |} \,\cdot 
\\
\cdot \,
\left [
\prod_{i = 1}^{|\YY|} \Big(\mathrm{Tay}_{|\s|}
\ee^{\gamma^{-1}(1+|E|)\, |\s|} \Big) \,
 \sup _{\wp, \, \underline{\sharp}, \, \underline{\sigma}}
\int \dd \phi_{x_{i}}\Big \| \Big (
\prod_{\varepsilon,}  \big (\phi_{ x'}^{\varepsilon}\big )^{d_{i,B}^{\varepsilon}}
\Big ) F_{\beta}(\Phi)\Big \| \right ]
\\
\leq \ee^{ \gamma ^{-1}\, \| H\|_{\infty,1} \, |\YY |} 
K^{|\YY|} \Big(\mathrm{Tay}_{|\s|}
\ee^{\gamma^{-1}(1+ |E|)\, |\s|} \Big)^{|\YY|}
\Big [
\prod_{i = 1}^{|\YY|}  
 (d_{i}!)^{p} \, M^{d_{i}} \Big]\,.
\end{multline}
To estimate the second line in \eqref{eq: partial_bound_g_N}, we shall exploit the exponential decay of the hopping Hamiltonian. If we define $ (H_{\theta})_{x,y} := \ee^{(1+\theta) \alpha|x-y|/2} H_{x,y} $, it follows that, see \eqref{lattice_operator_norm} and \eqref{eq: decay_H}, 
\begin{equation}
\begin{split}
& \big\| H_{\theta}\big \|_{\infty,1} 
= \sup _{x \in \mathbb{Z}^{\mathrm{D}}}\sum_{x' \in \mathbb{Z}^{\mathrm{D}}} \ee^{(1+\theta)\alpha|x-x'|/2} \sum_{\sigma,\sigma' \in \s}\big| (H_{x,x'})_{\sigma,\sigma'} \big| \leq  \mathcal{C}_{\theta} \, \alpha^{-\mathrm{D}} \;,
\end{split}
\end{equation}
where clearly $\mathcal{C}_{\theta} > \mathcal{C}$. The following quantity will be needed as well
\begin{equation}
 \big\| H_{\theta}\big \|_{\infty,\infty} := \sup _{x, x' \in \mathbb{Z}^{\mathrm{D}}}  \,\ee^{(1+\theta)\alpha|x-x'|/2} \, \sum_{\sigma,\sigma' \in \s}\big| (H_{x,x'})_{\sigma,\sigma'} \big| \leq \mathcal{C} \;.
\end{equation}
Standard tree-stripping estimates (see Fig.~\ref{fig: trees}) based on the exponential decay of $H$ give \cite{Brydges84}:

\begin{figure}

\centering
\begin{tikzpicture}
%%entering the nodes
\node[shape = circle,draw,label={ $x$}, fill, scale=.4] at (0, 0)   (x) {};
\node[shape = circle,draw, fill,scale=.2] at (.8, .8)   (x4) {};
\node[shape = circle,draw, fill,scale=.2] at (.8,-.8)   (x5) {};
\node[shape = circle,draw, fill, scale=.2] at (1.414*0.8,0)   (x3) {};
\node[shape = circle,draw,fill,scale=.2] at (2.828*0.8,0)   (y) {};
\node[shape = circle,draw, fill,scale=.2] at (2.414*0.8,0.8)   (x1) {};
\node[shape = circle,draw, fill,scale=.2] at (2.414*0.8,-.8)   (x2) {};
\node[xshift=3mm,yshift=1mm] at (x1) {$x_{1}$};
%%labels
\node[yshift=-2.5mm] at (y) {$y$};
\node[xshift=3mm,yshift=-1mm] at (x2) {$x_{2}$};
\node[xshift=-2mm,yshift=-2mm] at (x3) {$x_{3}$};
\node[xshift=1mm,yshift=-2mm] at (x5) {$x_{5}$};
\node[xshift=2mm,yshift=2mm] at (x4) {$x_{4}$};
\node[xshift=-8mm,yshift=-2mm] at (x) {\Large $\sum\limits_{\substack{x_{1},\dots,x_{5} \\ \mathrm{distinct}}}$};

%%branches of the tree
\draw (x5) -- (x); \draw (x) -- (x4);
\draw (x3) -- (x); \draw (x3) -- (x1);
\draw (x3) -- (y); \draw (x3) -- (x2); 

\node[xshift=1cm,yshift=-0mm] at (y) (sup){ $\leq \Bigg( \sup\limits_{{\displaystyle{ x_{3}}}}$}; 
\node[xshift=2.25cm,yshift=-2mm] at (y) {\Large  $ \, \sum\limits_{\substack{x_{1}, x_{2} \\ \mathrm{distinct}}}$};
%\node at (6.5,-.2)   { \Large $ \sum\limits_{\substack{x_{3},\dots,x_{5} \\ \mathrm{distinct}}}$};

%% first nodes
\node[xshift=3.2cm,shape = circle,draw, fill,scale=.2] at (y)  (x3') {};

\node[xshift=0.8*1.414cm,shape = circle,draw, fill,scale=.2] at (x3') (y'){};
\node[xshift=0.8cm,yshift=0.8cm,shape = circle,draw, fill,scale=.2] at (x3') (x1'){};
\node[xshift=0.8cm,yshift=-0.8cm,shape = circle,draw, fill,scale=.2] at (x3') (x2'){};

\node[xshift=-1.5mm,yshift=2mm] at (x3') {$x_{3}$};
\node[xshift=2mm,yshift=2mm] at (x1') {$x_{1}$};
\node[xshift=2mm,yshift=-2mm] at (x2') {$x_{2}$};
\node[xshift=0mm,yshift=-2.5mm] at (y') {$y$};
%% first branches
\draw (x3') -- (y'); \draw (x3') -- (x1');\draw (x3') -- (x2');  

\node[xshift = .45cm,yshift = -0mm] at (y') (bracket){$\Bigg)$};
\node[xshift=9mm,yshift = -1mm] at (bracket) { \Large $ \sum\limits_{\substack{x_{3},x_{4},x_{5} \\ \mathrm{distinct}}}$};

%%% other nodes
\node[xshift = 1.8cm,yshift = 1mm,shape = circle,draw,label={ $x$}, fill, scale=.4] at (bracket)  (x') {};
\node[xshift =.8*1.414cm ,shape = circle, draw, fill,scale=.2] at (x')   (x3'') {};
\node[xshift = .8cm,yshift = .8cm,shape = circle,draw, fill,scale=.2] at (x')   (x4') {};
\node[xshift = .8cm, yshift = -.8cm,shape = circle,draw, fill, scale=.2] at (x')  (x5') {};

\node[xshift=2mm,yshift=2mm] at (x4') {$x_{4}$};
\node[xshift=2mm,yshift=-2mm] at (x5') {$x_{5}$};
\node[xshift=0mm,yshift=-2.5mm] at (x3'') {$x_{3}$};

%% underbraces
\node[xshift= 0mm,yshift = -1.4cm] at (sup) (sup') {};
\node[yshift = -1.4cm] at (bracket) (bracket') {};
\node[xshift=4mm,yshift=1.4cm] at (bracket) (bracket'') {};
\node[xshift=1mm,yshift = 1.4cm] at (x3'') (end) {};
\draw [thick,decoration={brace,amplitude = 8pt,mirror,raise = .2mm},decorate] (sup') -- (bracket');
\draw [thick,decoration={brace,amplitude = 8pt,raise = .2mm},decorate] (bracket'') -- (end);
%% formulas underbrace
\path (sup') -- (bracket') node[midway,yshift = -9mm] (middle) {$\leq (4!)^{-q} \,(C_{q,\theta})^{3}\, \| H_{\theta}\|_{\infty,1}^{2}\,\| H_{\theta}\|_{\infty,\infty} $};
\path (bracket'') -- (end) node[midway,yshift = 9mm] (middle') {$\leq (3!)^{-q}(C_{q,\theta})^{2} \, \| H_{\theta}\|_{\infty,1}^{3} $};

%%branches
\draw (x') -- (x3'');  \draw (x') -- (x4');
\draw (x') -- (x5');

\end{tikzpicture}
\caption{Representation of the stripping procedure of Eq.~\eqref{eq: tree_like_estimate_H}. The vertices represent the spacial points and the lines the modulus of the entries of $H_{2\theta -1}$. For simplicity, we did not represent the internal degrees of freedom. We fix a root, say $x$, and start stripping the tree from its outermost branches, those with incidence number equal to one. The estimate is then carried out iteratively. 
\label{fig: trees}}
\end{figure}

\begin{multline}
\label{eq: tree_like_estimate_H}
\sum_{\underline{\sigma}}\sum_{\substack{x_{1},...,x_{N} \in \Lambda' \\ \mathrm{distinct}}}\prod_{\ell \in \textsc{T}_{\wp}}
\big| \big(H_{x_{\ell^{+}},x_{\ell^{-}}}\big)_{\sigma^{+}_{\ell},\sigma^{-}_{\ell}} \big| 
\\
\leq \,
\ee^{-\theta \alpha\, |x - y|} \, \sum_{\underline{\sigma}}\sum_{\substack{x_{1},...,x_{N} \in \Lambda' \\ \mathrm{distinct}}}\prod_{\ell \in \textsc{T}_{\wp}}
\big| \big((H_{2\theta-1})_{x_{\ell^{+}},x_{\ell^{-}}}\big)_{\sigma^{+}_{\ell},\sigma^{-}_{\ell}} \big| 
\\
\leq  \,
\ee^{-\theta \alpha\, |x - y|} \, \big\| H_{\theta}\big \|_{\infty,1}^{N} \, \big\| H_{\theta}\big \|_{\infty,\infty}^{1- \delta_{x,y}} 
\prod_{i = 1}^{| \YY|} 
(C_{q,\theta})^{d_{i}-1} \, (d_{i}!)^{-q}
 \;,
\end{multline}
where $q >0$ is a fixed parameter and where $C_{q,\theta} > 1$ is a constant that depends also on $\mathrm{D}$, the dimension of the lattice. 
The estimate is carried out by progressively stripping the outer branches as shown in Fig~\ref{fig: trees}. The branches that have been removed are then bounded as follows:
\begin{multline}
\label{eq: detail_stripping_procedure}
(d!)^{q} \,
\sup _{\bar{x}, \sigma} \sum_{\sigma_{1},\dots,\sigma_{d-1}}\sum_{\substack{x_{1},\dots,x_{d-1} \\ \mathrm{distinct}}} \prod _{j = 1}^{d-1} \,\big| \big((H_{2\theta-1})_{\bar{x},x_{j}}\big)_{\sigma, \sigma_{j}}\big|
\\
\leq 
\ee^{\,q\, d \ln d} \,
\sup _{\bar{x}, \sigma} \sum_{\sigma_{1},\dots,\sigma_{d-1}}\sum_{\substack{x_{1},\dots,x_{d-1} \\ \mathrm{distinct}}} \prod _{j = 1}^{d-1} \ee^{-(1-\theta)\alpha|\bar{x}-x_{j}|/2}\,\big| \big((H_{\theta})_{\bar{x},x_{j}}\big)_{\sigma, \sigma_{j}}\big|
\\
\leq \ee^{\,q \,d \ln d} \Big(
\sup _{\substack{x_{1},\dots,x_{d-1} \\ \mathrm{distinct}}} \prod_{j=1}^{d-1} \ee^{-(1-\theta)\alpha|\bar{x}-x_{j}|/2} \Big) \|H_{\theta} \|_{\infty,1}^{d-1}
\\
\leq \underbrace{\bigg(\sup _{d \geq 1} \Big(\ee^{\,2q\ln d- (1-\theta)\alpha \,\Omega_{\mathrm{D}}(d-1)^{1/\mathrm{D}}}  \Big)\bigg)^{d-1}}_{=: \,C_{q,\theta}^{d-1}} \,\| H_{\theta} \|_{\infty,1}^{d-1} \;.
\end{multline}
In the second inequality in \eqref{eq: detail_stripping_procedure} we used a simple bound for the factorial and we extracted an exponential weight from $H$ and pulled it out of the summation by taking the superior over all distinct points $x_{1}\neq\dots \neq x_{d-1} \neq \bar{x}$. The latter can be computed by noticing that $ \big| \{ x \in \mathbb{Z}^{\mathrm{D}} \, |\, |x| = r \} \big| \leq \Omega_{\mathrm{D}}' r^{\mathrm{D}-1}$, therefore $\sum_{j = 1}^{d-1}|\bar{x}-x_{j}|/2 \geq \Omega_{\mathrm{D}} (d-1)^{1+1/\mathrm{D}}$, for some constants $\Omega_{\mathrm{D}}<1$ and $\Omega_{\mathrm{D}}'>1$. 
%\textcolor{red}{put short justification on the constant or expand!!}
While stripping the tree, if one of the outer vertices is $y$, there is no summation and this produces simply $\|H_{\theta} \|_{\infty,\infty}$ instead of $\|H_{\theta} \|_{\infty,1}$, see Fig.~\ref{fig: trees}.
We shall remark that the factorial $(d_{i}!)^{-q}$ is gained only because we are summing over distinct points and because the matrix elements of the Hamiltonian decay exponentially.

Plugging the bounds \eqref{eq: superintegral_bound_fourier} and \eqref{eq: tree_like_estimate_C} with $q= p+1$ into \eqref{eq: partial_bound_G_N} and using Cayley's theorem on the number of trees with fixed coordination numbers $\{ d_{i} \}_{i}$, see \cite{Brydges84}, we obtain \eqref{ineq: g_N_bound}:
\begin{equation}
\begin{split}
\big | \big(g_{N}(x,y;z_{\pm} \pm \ii \eps)\big)_{\sigma,\sigma'} \big |
& \leq  \ee^{-\theta \alpha|x-y| } 
K^{|\YY|} \ee^{\mathcal{C} |\YY|}\,\Big(\mathrm{Tay}_{|\s|}
\ee^{\gamma^{-1}(1+|E| )\, |\s|} \Big)^{|\YY|}
\mathcal{C} \;\cdot  
\\
& \cdot (\mathcal{C}_{\theta} \, \alpha^{-\mathrm{D}})^{N} 
\sum_{ \substack{\{ d_{i}\}_{i}, \, d_{i} \geq 1 \\  \sum_{i} d_{i} = 2| \YY| - 2}} \, \prod_{i = 1}^{|\YY |} \frac{(d_{i}!)^{p} (C_{p+1,\theta})^{d_{i} -1} \,M^{d_{i}}}{(d_{i} - 1)! \, (d_{i}!)^{p+1}}\,
\\
& \leq \big( 1+ \gamma^{-|\s|}(1+|E|)^{|\s|} \big)^{|\YY|}  \, (\overline{C}_{K,M,p,\theta})^{| \YY|} \,\alpha^{-\mathrm{D} \,N } \,\ee^{-\theta \alpha |x - y|} 
\end{split}
\end{equation}
where 
$
\overline{C}_{K,M,p,\theta} := 2 K M \ee^{\mathcal{C}} \mathcal{C}_{\theta}\sum_{n = 0}^{\infty} \frac{((n+1)!)^{p} \, (M\,C_{p+1,\theta})^{n}}{(n!)^{p+2}}$ $ < \infty \;.
$
\qed
\medskip

\section{SUSY cluster expansion at weak disorder}
\label{sec: SUSY_cluster_expansion_weak}

%\textcolor{red}{Here we need a priori that $\mathbb{E}_{\omega} |\omega|^{n} $ finite for any $n$, because we perform derivative expansion. }

In this section, we prove the exponential decay of the disorder-averaged Green's function (Theorem~\ref{thm: exp_decay_weak}), the smoothness of the LDOS (Corollary \ref{thm: analyticity_LDOS_weak}) and we establish Lipshitz-tail-type estimates for the latter (Theorem~\ref{thm: Lifshitz_LDOS_weak})  at weak disorder and at energies away from the unperturbed spectrum.
The analysis of the dual SUSY integral is based on the SUSY cluster expansion presented in Proposition~\ref{thm: cluster_expansion_green_function_weak}. The proof of the theorems is again completed by means of tree estimates together with some suitable bounds on the norm of the superfunctions to be integrated. As was the case in Section \ref{sec: SUSY_cluster_expansion_strong}, such bounds can be achieved under some reasonable assumptions on the disorder distribution that we anticipated in (H2), see the Introduction. In order to make the assumption precise, we need the following definition.
\begin{definition}[Integrable Derivative Bounds]
\label{def: bounds_weak_disorder}
Let $f= f(\Phi)$ be a function of a supervector $\Phi = (\phi,\psi)$. We say that $f$ satisfies integrable derivative bounds (IDB) if for some $K \geq 0$, $  M \geq 1$ and $p \geq 0$ the following  holds true
\begin{equation}
\int \dd \phi \, \Big\|\Big( \prod_{\varepsilon = \pm} \prod_{\sigma \in \s} \Big(\frac{\partial}{\partial \phi^{\varepsilon}_{\sigma}} \Big)^{n^{\varepsilon}_{\sigma}}\Big)  f(\Phi) \Big\|\leq K \, M^{n} \, (n!)^{p}  \;,
\end{equation}
for all $ \{n^{\varepsilon}_{\sigma} \} \in \mathbb{N}^{\s \times \{\pm\}}$, having set $n := \sum_{\varepsilon, \sigma}n^{\varepsilon} _{\sigma}$.
\end{definition}

Throughout this section we will assume the following:
\medskip \medskip 

\noindent
$\,$ (H2-$\mathrm{I}$) $ \,$Let $\beta \geq 0$. The superfunction 
$ F_{\beta}(\Phi) = \ee^{\gamma^{-1}\beta \Phi^{+}\Phi^{-}} \hat{\nu}(\Phi^{+}\Phi^{-})$ satisfies \textcolor{white}{wwwwl} IDB for some $K,M$ and $p$.
\medskip \medskip

In passing, we notice that in the case of Gaussian disorder, the superfuction $F_{\beta}$ with $\beta \leq C \gamma$, $C$ being some universal constant, satisfies IDB for some $K$, $M$ and $p$ independent of $\gamma$.
Furthermore, we believe that our analysis could be extended to the case of weakly positively correlated disorder if we make assumptions on the regularity of $\nu$ that are stronger than (H2-$\mathrm{II}$), but still applicable to Gaussian disorder.

We can finally state the main theorems of this section. 
\begin{theorem}[Exponential decay]
\label{thm: exp_decay_weak}
Let $E \in \mathbb{R}$, $\theta \in [0,1)$, and set $z_{\pm} = E \mp \ii \beta$ together with $\delta:=\mathrm{dist}(E,\sigma(H))$. Assume that the covariance decays as
\begin{equation}
\label{eq: decay_covariance_hopping_operator}
\sum_{\sigma,\sigma' \in \s}\big| \big(C^{E \pm \ii \eps}_{x,y}\big)_{\sigma,\sigma'}\big|
\leq \mathcal{C}_{\mathrm{D}} \,\frac{\ee^{- \sqrt{\delta} \,|x - y|}} {1 + |x - y |^{\mathrm{D}-2}} \;, \qquad \forall \eps \geq 0 \;.
\end{equation}
There exists a constant $C_{K,M,p,\theta} > 0$ such that if $\delta \geq \gamma \, C_{K,M,p,\theta} $ then:
\begin{equation}
\label{eq: exp_decay_weak}
\sup _{\sigma,\sigma'}\Big| \big(\G_{\Lambda}(x,y;z_{\pm} \pm \ii \eps)\big)_{\sigma,\sigma'} \Big| \leq \gamma^{-\delta_{x,y}} \, (C_{K,M,p,\theta})^{2-\delta_{x,y}} \,\ee^{-\theta \,\sqrt{\delta}|x-y |} \;,
\end{equation}
uniformly in $\Lambda \subset \mathbb{Z}^{\mathrm{D}}$ and $\epsilon \geq 0$. 
\end{theorem}
\begin{remark}
Notice that \eqref{eq: decay_covariance_hopping_operator} is satisfied when $H = -\Delta_{\mathbb{Z}^{\mathrm{D}}}$ and, more generally, when $H$ exhibits quadratic dispersion at the band edges. 
Notice that at $x=y$ the limit $\gamma \to 0$ of the r.h.s.~of \eqref{eq: exp_decay_weak} is divergent. This is due to the fact that the bound is uniform in $\eps \geq 0$; the meaningful way to compute the $\gamma \to 0$ limit is however at finite $\eps$. A bound which is uniform in $\eps \geq \eps _{0}$ for some $\eps _{0} >0$, and which does not diverge as $\gamma$ is sent to $0$ can also be obtained using our methods but this is beyond our scope.
\end{remark}
The case $\beta = 0$, $\eps \to 0^{+}$ gives the exponential decay of the disorder-averaged Green's function for energies up to the edges of the spectrum. On the other hand, at $\beta >0$ and $x = y$, the result implies analyticity of the LDOS in a suitable region of the complex plane.
\begin{corollary}[Analyticity of LDOS]
\label{thm: analyticity_LDOS_weak}
Assume that the covariance satisfies \eqref{eq: decay_covariance_hopping_operator}.
Then, there exists a constant $C_{K,M,p} > 0$ such that $\rho(E)$ can be extended to an  analytic function on 
\begin{equation}
D_{K,M,p}(\gamma,\beta) := \Big \{ z \in \mathbb{C} \;\Big | \; \mathrm{dist}\big( \re z, \sigma(H)\big) > \gamma\, C_{K,M,p} \,, \; \big| \im z\big| < \beta 
\Big\} \;.
\end{equation}
\end{corollary}
%%
%\begin{remark}
%Of course, this statement is meaningful if $\beta > 0$.
%\end{remark}
%%
The proof of this corollary is identical to that of Corollary~\ref{thm: DOS_strong_disorder}: Lemma~\ref{lemma: analytic_continuation_apriori} gives the analyticity of $\G_{\Lambda}(0,0;E \pm \ii \eps)$ while Theorem~\ref{thm: exp_decay_weak} provides the uniform bounds that allow us to extend the analyticity to $\G_{\mathbb{Z}^{\mathrm{D}}}(0,0;E \pm \ii 0^{+})$.

Our last result is as follows.

\begin{theorem}[Lifshitz-tail-type estimate]
\label{thm: Lifshitz_LDOS_weak}
Let $E \in \mathbb{R}$ and set $\delta:=\mathrm{dist}(E,$ $\sigma(H))$. Assume  that the covariance satisfies \eqref{eq: decay_covariance_hopping_operator} and that $C_{x,y}^{E} \in \mathbb{R}$ for any $x,y \in \mathbb{Z}^{\mathrm{D}}$.
There exists a constant $C_{K,M,p}$ such that if $\delta \geq \gamma C_{K,M,p}$ then a Lifshitz-tail-type estimate is satisfied:
\begin{equation}
\label{lifshitz_bound}
\big| \rho_{\eps,\Lambda}(E) \big|
\leq \gamma^{-1} C'_{K,M,p} \, \ee ^{ - \overline{C}_{K,M,p} \,(\gamma^{-1} \delta)^{1/2p}}\;,
\end{equation}
for some constants $C'_{K,M,p}, \overline{C}_{K,M,p}>0$, uniformly in $\Lambda \subset \mathbb{Z}^{\mathrm{D}}$ and $\eps \geq 0$.
\end{theorem}
\begin{remark}
Since $E \in \mathbb{R}$, $\beta = 0$ is sufficient for proving \eqref{lifshitz_bound}.
\end{remark}
\begin{remark}
Notice that the condition $C_{x,y}^{E} \in \mathbb{R}$ is satisfied in the case $H = -\Delta_{\mathbb{Z}^{\mathrm{D}}}$ and, more generally, in systems for which the Bloch Hamiltonian is such that $\hat{H}(-k) = \overline{\hat{H}(k)}$. This property holds in many condensed-matter systems, an example of which is graphene \cite{GiulianiMastropietro}.
\end{remark}
\begin{remark}
\label{rmk: IDOS_rmk}
Since the LDOS is finite, the integrated density of states IDOS can be obtained via \cite{Simon}
\begin{equation}
\mathcal{N}(E) = \int_{-\infty}^{E} \dd E' \, \rho(E') \;.
\end{equation}
As it turns out, the IDOS controls the probability of having an eigenvalue of $H_{\omega,\Lambda}$ below $E$ \cite{Klopp, KloppWolff}:
\begin{equation}
\label{eq: probability_estimate_IDOS}
\mathbb{P}\Big(\big\{ \sigma(H_{\omega,\Lambda}) \cap (-\infty,E) \neq \emptyset \big\} \Big) \leq C\, |\Lambda| \, \mathcal{N}(E)\;,
\end{equation}
for some constant $C >0$.
A sufficiently small IDOS, e.g., as $\gamma \to 0$, allows us to prove localization via finite-volume criteria \cite{AizenmanSchenker}, see \cite{Klopp} for details.
Integration of the Lifshitz-tail-type estimate of Theorem~\ref{thm: Lifshitz_LDOS_weak}, up to energies below the bottom of the spectrum such that $ \,(\gamma^{-1} \delta)^{1/2p} \gtrsim  |\ln \gamma|$, provides a sufficiently small upper bound on the IDOS.
\end{remark}
The proof of the theorems above is based on a dual SUSY cluster expansion. The approach is the same as the one in Proposition~\ref{thm: cluster_expansion_green_function_strong}. This time we shall expand the super Gibbs' weight $\widehat{\mu}^{\pm}_{\Lambda}(\xi)$ and obtain polymers connected by the covariance. The expansion is well-suited for our purposes because it exploits the smallness of $\gamma \, \delta^{-1}$. 
\begin{proposition}
\label{thm: cluster_expansion_green_function_weak}
Let $E\in \mathbb{R}$ and set $z_{\pm}:= E \mp \ii \beta$. The following representation holds true for any $\eps \geq 0$:
\begin{eqnarray}
&&  \G_{\Lambda}(x,y; z_{\pm} \pm \ii \epsilon ) =  -(\pm \ii \gamma)^{-\delta_{x,y}} \sum_{N = 0}^{|\Lambda'|}
\big( \pm \ii \gamma\big)^{N}
 G_{N}(x,y; z_{\pm}  \pm \ii \eps) \qquad
\label{eq: cluster_expansion_green_function_weak}
\\
&& G_{N}(x,y;z_{\pm} \pm \ii \eps):= \frac{1}{N!}
 \sum_{\substack{\textsc{t} \,\,\, \mathrm{on} \\
 \,\{1,...,N+2 - \delta_{x,y}\}}} \,\,
\sum_{\substack{x_{1},...,x_{N} \in \Lambda' \\ \mathrm{distinct}}}  
\int \dd p_{\textsc{t}}(s)  \;\cdot
\nonumber \\
&&
\qquad \qquad \qquad \qquad \cdot \; \int \dd \widehat{\mu}^{\pm}_{\yy}(\xi,s) \,
 \prod_{\{i,j\}\in \textsc{t}} \hat{v}^{\pm}_{x_{i},x_{j}}(\xi) \,\, \frac{\partial}{\partial \eta^{+}_{x}} \frac{\partial}{\partial \eta^{-}_{y}}\widehat{F_{\beta}}^{\yy}(\xi)\,,\qquad \;
 \label{eq: G_N_definition}
\end{eqnarray}
where we have set $\dd \widehat{\mu}^{\pm}_{\yy}(\xi,s) : = \dd \xi_{\yy}  \widehat{\mu}^{\pm}_{\yy}(\xi,s)$ with
\begin{equation}
\label{eq: definition_mu_v_weak_disorder}
\begin{split}
& \widehat{\mu}^{\pm}_{\yy}(\xi,s) : = \ee^{ \pm \frac{\ii }{2}\gamma \sum_{x,y \in \yy} s_{x,y} \hat{v}_{x,y}^{\pm}(\xi)} \;,
\\
& \hat{v}_{x,y}^{\pm}(\xi): =  \xi_{x}^{+}\, C^{E\pm \ii \eps}_{x,y}\,\xi_{y}^{-}  +\xi_{y}^{+} \, C^{E \pm \ii \eps}_{y,x}\,\xi_{x}^{-}\;.
\end{split}
\end{equation}
Above, $s = (s_{x,y}) \in [0,1]^{\mathcal{P}(\YY)}$ while $\dd p_{\textsc{t}}(s)$ is a probability measure with support on $s$ such that the exponent in the super Gibbs' weights $\widehat{\mu} _{\yy}^{\pm}(\xi,s)$ satisfies respectively $ \re \pm \ii \sum_{x,y \in \YY}  s_{x,y} \hat{v}^{\pm}_{x,y}((\kappa,0))  \leq 0$. Furthermore, we have set $x_{N+1} := x$ and $x_{N+2} := y$, $\YY := \{x_{1},...,x_{N+2} \}$ and $\Lambda' := \Lambda \setminus \{x,y\}$.
\end{proposition}
\begin{remark}
Notice that only the empty-tree term contributes to the quantity $G_{0}(x,x;z_{\pm}\pm \ii \eps)$.
\end{remark}
In order to prove this SUSY cluster expansion, we need an auxiliary result that allows us to apply the SUSY localization formula in this context as well.
\begin{lemma}
\label{lemma: SUSY_fourier_transform}
Let $f \in L^{1}(\mathcal{S},\mathscr{G})$ be even, supersymmetric and invariant under $U(1)^{\times \s}$ fermionic transformations, see Lemma~\ref{lemma: decomposition_Q_operator}. Assume that the function $\Phi \mapsto \psi^{\varepsilon}_{\sigma}(\partial/\partial \phi^{\varepsilon}_{\sigma}) f(\Phi)$ is in $L^{1}(\mathcal{S},\mathscr{G})$ for any $\varepsilon,\sigma$. Then
\begin{equation}
\widehat{f}(\xi) = \int \dd \Phi \, \ee^{-\ii \xi^{+}\Phi^{-} - \ii \Phi^{+}\xi^{-}} \, f(\Phi) \;,
\end{equation}
is even, satisfies
\begin{equation}
\label{eq: SUSY_fourier_transform}
\eta^{\varepsilon}_{\sigma} \frac{\partial}{\partial \kappa^{\varepsilon}_{\sigma}} \widehat{f}(\xi) = - \varepsilon 
\kappa^{-\varepsilon}_{\sigma} \frac{\partial}{\partial \eta^{-\varepsilon}_{\sigma}} \widehat{f}(\xi)
\end{equation}
and is therefore supersymmetric.
\end{lemma}
\begin{proof}
Integration by parts gives
\begin{equation}
\begin{split}
-\varepsilon \kappa^{\varepsilon}_{\sigma} \frac{\partial}{\partial \eta^{\varepsilon}_{\sigma}} \, \widehat{f}(\xi)
 & = 
\int \dd \Phi \, \ee^{-\ii \xi^{+}\Phi^{-} - \ii \Phi^{+}\xi^{-}} \,
\psi^{-\varepsilon}_{\sigma} \frac{\partial}{\partial \phi^{-\varepsilon}_{\sigma}} f(\Phi)\;,
\\ 
\eta^{\varepsilon}_{\sigma} \frac{\partial}{\partial \kappa^{\varepsilon}_{\sigma}} \, \widehat{f}(\xi)
 & = 
\int \dd \Phi \, \ee^{-\ii \xi^{+}\Phi^{-} - \ii \Phi^{+}\xi^{-}}\,\varepsilon 
\phi^{-\varepsilon}_{\sigma} \frac{\partial}{\partial \psi^{-\varepsilon}_{\sigma}} f(\Phi) \;,
\end{split}
\end{equation}
provided that the integrands are in $L^{1}(\mathcal{S},\mathscr{G})$. In the first line this is the case by assumption. In the second line, by Lemma~\ref{lemma: decomposition_Q_operator} we have that
\begin{equation}
\label{eq: consequence_lemma}
- \varepsilon 
\phi^{-\varepsilon}_{\sigma} \frac{\partial}{\partial \psi^{-\varepsilon}_{\sigma}} f(\Phi) = \psi^{\varepsilon}_{\sigma} \frac{\partial}{\partial \phi^{\varepsilon}_{\sigma}} f(\Phi) 
\end{equation}
and thus the integrand is in $L^{1}(\mathcal{S},\mathscr{G})$ as well. Taking the difference of the two equations at fixed $\varepsilon$ and $\sigma$, and using Eq.~\eqref{eq: consequence_lemma} gives identity \eqref{eq: SUSY_fourier_transform}. Parity follows by using that $\dd \Phi$ is invariant under $\Phi \mapsto - \Phi$.
\end{proof}
\begin{proof}[Proof of Proposition~\ref{thm: cluster_expansion_green_function_weak}] We shall skip many details since the proof is analogous to the one of Proposition~\ref{thm: cluster_expansion_green_function_strong}.
Since $F_{\beta}$ satisfies IDB, $\widehat{F_{\beta}}(\xi)$ decays in norm faster than any power of $\kappa$. In particular,
 $\widehat{F_{\beta}} \in L^{1}(\mathcal{S},\mathscr{G})$ and we can make sense of the SUSY integral in Eq.~\eqref{eq: SUSY_greens_function_direct} when $E \to z_{\pm}$, $\eps \geq 0$.
We set $\widehat{\mu}_{\xx}^{\pm} =: \ee^{\widehat{V}_{\XX}^{\pm}}$, apply the BBF formula and, by exploiting the decay properties of $\widehat{F_{\beta}}$, obtain a polymer expansion like the one in the proof of Proposition~\ref{thm: cluster_expansion_green_function_strong}. The stability condition $\re (\hat{V}_{\XX}^{\pm})_{\emptyset} = \re \pm \ii \gamma\sum_{x,y\in \XX} \kappa^{+}_{x} C^{E\pm \ii \eps}_{x,y} \kappa^{-}_{y} \leq 0$ is satisfied, since  $C^{E\pm \ii \eps} = (H-E \pm \ii \eps)/((H-E)^{2} + \eps ^{2})$ with $H$ Hermitian and $\eps \geq 0$. Accordingly, $\dd p_{\textsc{t}}(s)$ is supported on $s$ such that $ \re \pm \ii \sum_{x,y \in \YY}  s_{x,y} \hat{v}^{\pm}_{x,y}((\kappa,0))  \leq 0$ respectively. In the final formula for $G_{N}$ we have swapped integration with respect to $\dd p_{\textsc{t}}(s)$ with the superintegral with respect to $\dd \xi_{\YY}$ because $F_{\beta}$ satisfies IDB.

Again, we notice that $\widehat{\mu}^{\pm}_{\yy}(\cdot,s)$ and $\hat{v}^{\pm}_{x,y}$ are even and supersymmetric by inspection.
 We have already pointed out that $F_{\beta}$ is even, supersymmetric and invariant under $U(1)^{\times \s}$ fermionic transformation, see discussion before Proposition~\ref{prop: SUSY_representation}. Since by assumption it satisfies IDB, by Lemma~\ref{lemma: SUSY_fourier_transform} we have that $\widehat{F_{\beta}}$ is even, supersymmetric and satisfies identity \eqref{eq: SUSY_fourier_transform}, or equivalently:
\begin{equation}
\frac{\eta^{\varepsilon}_{\sigma}}{(1+|\kappa^{\varepsilon}_{\sigma}|)}
\frac{\partial}{\partial \kappa^{\varepsilon}_{\sigma}}\widehat{F_{\beta}}(\xi) =-\varepsilon 
\frac{\kappa^{-\varepsilon}_{\sigma}}{(1+|\kappa^{\varepsilon}_{\sigma}|)}
\frac{\partial}{\partial \eta^{-\varepsilon}_{\sigma}} \widehat{F_{\beta}}(\xi) \;.
\end{equation}
Following the argument below \eqref{eq: integrand}, we finally have that the superfunction
$\xi \mapsto \widehat{\mu}^{\pm}(\xi,s)$ $\big( \prod_{ \{x',y'\} \in \textsc{T}} \hat{v}_{x',y'}(\xi) \big)\; \widehat{F_{\beta}^{\yy}}(\xi)$ satisfies the hypotheses of Proposition~\ref{prop: SUSY_localization_theorem}.
The rest of the proof is identical to the one of Proposition~\ref{thm: cluster_expansion_green_function_strong}.
\end{proof}

\subsection*{Proof of Theorem \ref{thm: exp_decay_weak}}
Since $F_{\beta}$ satisfies IDB we can apply Proposition~\ref{thm: cluster_expansion_green_function_weak}. We shall prove that if $\gamma \, \delta^{-1}$ is small enough then for some constant $C_{K,M,p,\theta}$ the following bounds holds true:
\begin{equation}
\label{eq: main_bound_weak_disorder}
\sup _{\sigma,\sigma'} \big| \big(G_{N}(x,y;z_{\pm} \pm \ii \eps)\big)_{\sigma,\sigma'} \big| \leq 
(C_{K,M,p,\theta}/2)^{N + 2 - \delta_{x,y}}\, \delta^{-N}  \, \ee^{-\theta \sqrt{\delta} |x - y|} \;.
\end{equation}
Plugging these bounds into the expansion of Proposition~\ref{thm: cluster_expansion_green_function_weak} proves the convergence of the expansion and hence the claim.
The proof of the bound~\eqref{eq: main_bound_weak_disorder} follows closely the strategy in the proof of Theorem~\ref{thm: exponential_decay_strong_disorder}, with some small differences which we shall stress.

The setting is as in the proof of Theorem~\ref{thm: exponential_decay_strong_disorder}, but we shall recall it for the sake of clarity.
Let us consider a polymer $\YY \ni x,y$ and an oriented tree $\textsc{T}_{\wp}$ on $\YY$, $\wp$ denoting the choice of orientations on the links of the tree. We denote by $\ell^{+}$ ($\ell^{-}$) the starting (ending) vertex of the oriented link $\ell \in \textsc{T}_{\wp}$. Links have to be oriented in order to select one of the two elements in $\hat{v}^{\pm}_{\ell}(\xi)$, see in Eq.~\eqref{eq: definition_mu_v_weak_disorder}. Furthermore, we introduce the sequences $\underline{\sigma} = \big\{ \sigma^{\varepsilon}_{\ell} \in \s \big\}_{\ell \in \textsc{T},\,\varepsilon = \pm}$ and $\underline{\sharp} = \big\{ \sharp_{\ell} \in \{B,F\} \big\}_{\ell \in \textsc{T}}$. We set $\xi^{\varepsilon}_{B,x,\sigma} = \kappa^{\varepsilon}_{x,\sigma}$ and $\xi^{\varepsilon}_{F,x,\sigma} = \eta^{\varepsilon}_{x,\sigma}$, and we shall henceforth write $|\YY|$ instead of $N+2 - \delta_{x,y} = |\YY|$. Let us define:
\begin{equation}
\label{eq: F_mathcal_weak}
\begin{split}
\widehat{\mathcal{F}}_{\textsc{t}_{\wp},\,\underline{\sharp}, \,\underline{\sigma}}^{\yy}(\xi) &:=
\Big ( \prod_{\ell \in \mathrm{T}_{\wp}} \,\xi^{+}_{\sharp_{\ell},x_{\ell^{+}},\sigma_{\ell}^{+}} \, \xi^{-}_{\sharp_{\ell},x_{\ell^{-}},\sigma_{\ell}^{-}} \Big )
\, \frac{\partial}{\partial\eta^{+}_{x}} \frac{\partial}{\partial\eta^{-}_{y}} \widehat{F_{\beta}}^{\yy}(\xi) \;.
\end{split}
\end{equation}
We can rewrite Eq.~\eqref{eq: G_N_definition} as
\begin{multline}
\label{eq: polarised_tree_expansion_green_function_weak}
G_{N}(x,y;z_{\pm} \pm \ii \eps) =  \frac{1}{N!}\sum_{\substack{\textsc{t} \,\,\, \mathrm{on} \\
 \,\{1,...,|\yy|\}}} \,\sum_{\wp, \,\underline{\sharp}, \, \underline{\sigma}}\,\,
\sum_{\substack{x_{1},...,x_{N} \in \Lambda' \\ \mathrm{distinct}}}  
\,
\int \dd p_{\textsc{t}}(s) \; \cdot
\\
\cdot \; \int \dd \widehat{\mu}^{\pm}_{\yy}(\xi,s) \,
\widehat{\mathcal{F}}_{\textsc{t}_{\wp},\,\underline{\sharp}, \, \underline{\sigma}}^{\yy}(\xi)  \,\Big (\prod_{\ell \in \textsc{T}_{\wp}}
\big(C^{E \pm \ii \eps}_{x_{\ell^{+}},x_{\ell^{-}}}\big)_{\sigma^{+}_{\ell},\sigma^{-}_{\ell} }\Big ) \;,
\end{multline}
which can be bounded as follows:
\begin{multline}
\sup _{\sigma,\sigma'}\big| \big(G_{N}(x,y;z_{\pm} \pm \ii \eps) \big)_{\sigma,\sigma'}\big| \\ \leq  \frac{4^{| \YY |}}{N!}\sum_{\substack{\textsc{t} \,\, \mathrm{on} 
 \,\{1,...,|\yy | \}}} \,
 \sup _{x_{1}, \dots, x_{N}} \sup _{\substack{\wp, \, \underline{\sharp}, \, \underline{\sigma},\, s \\ \sigma, \sigma'} } \Big|
\int \dd \widehat{\mu}^{\pm}_{\yy}(\xi,s) \,
\big(\widehat{\mathcal{F}}_{\textsc{t}_{\wp},\,\underline{\sharp}, \, \underline{\sigma}}^{\yy}(\xi)\big)_{\sigma,\sigma'} \Big|  \;
\\
\; \sup _{\wp} \,
\sum_{\underline{\sigma}}
 \sum_{\substack{x_{1},...,x_{N} \in \Lambda' \\ \mathrm{distinct}}}  
\prod_{\ell \in \textsc{T}_{\wp}}
\big| \big(C^{E \pm \ii \eps}_{x_{\ell^{+}},x_{\ell^{-}}}\big)_{\sigma^{+}_{\ell},\sigma^{-}_{\ell}}\big|  \;,
\label{eq: partial_bound_G_N}
\end{multline}
where the $\sup _{s}$ is taken over $s$ in the support of $\dd p_{\textsc{t}}$.
We shall first obtain a useful bound for the superintegral. The analysis differs slightly from the one of the proof of Theorem~\ref{thm: exponential_decay_strong_disorder} because we need to prove suitable decay bounds for $\big\| \big(\widehat{\mathcal{F}}_{\textsc{t}_{\wp},\,\underline{\sharp}, \, \underline{\sigma}}^{\yy}(\xi)\big)_{\sigma,\sigma'} \big \|$. We again see by inspection that 
 $\widehat{\mathcal{F}}_{\textsc{t}_{\wp},\,\underline{\sharp}, \, \underline{\sigma}}^{\yy}(\xi) $ is a local function: 
\begin{equation}
\label{eq: factorization_F_expansion}
\widehat{\mathcal{F}}_{\textsc{t}_{\wp},\,\underline{\sharp}, \, \underline{\sigma}}^{\YY}(\xi) = \pm \prod_{i = 1}^{|\YY|} \, \widehat{\mathcal{F}}_{x,y}^{\, \underline{d_{i}}}(\xi_{x_{i}}) \;,
\end{equation}
having set:
\begin{equation}
\label{eq: mathcal_F_local_function}
\widehat{\mathcal{F}}_{x,y}^{\,\underline{d}}(\xi_{x'}):=
\Big (
\prod_{\varepsilon,\sigma, \sharp}  \big (\xi_{\sharp,x',\sigma}^{\varepsilon}\big )^{d_{\sharp,\sigma}^{\varepsilon}}
\Big )
\,\left (\frac{\partial}{\partial \eta_{x}^{+}}\right )^{\delta_{x',x}} \left (\frac{\partial}{\partial \eta_{y}^{-}}\right )^{\delta_{x',y}} \widehat{F_{\beta} }(\xi_{x'}) \;,
\end{equation}
where $\underline{d}_{i}= \{(d_{i})_{\sharp,\sigma}^{\varepsilon}\}_{\varepsilon = \pm, \,\sigma \in \s,\,\sharp = B,F} $ are as in the proof of Theorem~\ref{thm: exponential_decay_strong_disorder}.
We estimate $\big\|\big( \widehat{\mathcal{F}}_{x,y}^{\, \underline{d_{i}}}(\xi_{x_{i}}) \big)_{\sigma,\sigma'}\big\|$ by means of the IDB on $F_{\beta}$. In fact, for any $m \in \mathbb{N}$, by Lemma \ref{lemma: bound_fourier_transform}:
\begin{multline}
\big\| (\kappa^{+}_{x'}\kappa^{-}_{x'})^{m}\big(\widehat{\mathcal{F}}_{x,y}^{\, \underline{d_{i}}}(\xi_{x'})\big)_{\sigma,\sigma'} \big\|
 \leq \Big \|(\kappa^{+}_{x'}\kappa^{-}_{x'})^{m}\Big( \prod_{\varepsilon,\sigma}  \big (\kappa_{x',\sigma}^{\varepsilon}\big )^{(d_{i}) _{B,\sigma}^{\varepsilon}} \widehat{F_{\beta}}(\xi_{x'})
\Big ) \Big \|
\\
 \leq  \int \dd \phi_{x'} \,\Big \| \Big[\Big ( \frac{\partial}{\partial\phi^{+}_{x'}}\Big ) \Big ( \frac{\partial}{\partial\phi^{-}_{x'}}\Big ) \Big]^{m}\Big (\prod_{\varepsilon,\sigma} \Big ( \frac{\partial}{\partial\phi^{\varepsilon}_{x'}}\Big )^{(d_{i})^{\varepsilon}_{B,\sigma}}\Big ) F_{\beta}(\Phi_{x'}) \Big \| \; ;
\end{multline}
thus, by the IDB on $F_{\beta}$ we obtain:
\begin{equation}
\big\| \widehat{\mathcal{F}}_{x,y}^{\, \underline{d_{i}}}(\xi_{x_{i}}) \big\|
\leq 2 \,
K  |\s|^{m}\,\frac{ ((d_{i,B} + 2m)!)^{p} \, M^{d_{i,B}}}{1 + M^{-2m}\,(\kappa_{x_{i}}^{+}\kappa_{x_{i}}^{-})^{m}}  \;,
\qquad
\forall m \in \mathbb{N} \;.
\end{equation}
Finally, since $M \geq  1$ 
\begin{equation}
 \sup _{\wp, \, \underline{\sharp}, \, \underline{\sigma} }
 \Big \| \widehat{\mathcal{F}}_{\textsc{t}_{\wp},\,\underline{\sharp}, \, \underline{\sigma} }^{\yy} (\xi)\Big \| \leq 
\prod_{i = 1}^{|\YY |} \; 2 K |\s|^{m}\, \frac{((d_{i}+2m)!)^{p} \, M^{d_{i}} }{1 + M^{-2m}\, (\kappa_{x_{i}}^{+} \kappa_{x_{i}}^{-})^{m}}, \qquad \forall m \in \mathbb{N} \;.
\label{eq: decay_global_fourier_transform}
\end{equation}
On the other hand, the non-local part in the superintegral is bounded as:
\begin{equation}
\label{simple_bound_gibbs_weight}
\sup _{ s} \,\big \| \widehat{\mu}^{\pm}_{\yy}(\xi,s) \big \| \leq \ee^{ \gamma \, \| C^{E \pm \ii \eps}\|_{\infty,1} \, |\YY |} \;,
\end{equation}
where $\| C^{E \pm \ii \eps}\|_{\infty,1} \leq \mathcal{C}\delta^{-1}$ is bounded thanks to \eqref{eq: decay_covariance_hopping_operator}, where for simplicity we have dropped the dependence of $\mathcal{C}$ on $\mathrm{D}$, the dimension of the lattice.
All in all, if we apply the bound in Eq.~\eqref{eq: decay_global_fourier_transform} with $m = |\s|+1$ and use $|\int \dd \xi _{\yy} f(\xi) | \leq \| f\|_{L^{1}(\mathcal{S}^{\YY},\mathscr{G}^{\YY})}$, we finally obtain:
\begin{multline}
\label{eq: superintegral_bound_fourier}
 \sup _{\substack{\wp, \, \underline{\sharp}, \, \underline{\sigma},\, s \\ \sigma ,\sigma'} } \Big|
\int \dd \widehat{\mu}^{\pm}_{\yy}(\xi,s) \,
\big(\mathcal{F}_{\textsc{t}_{\wp},\,\underline{\sharp}, \, \underline{\sigma}}^{\yy}(\xi)\big)_{\sigma,\sigma'} \Big|
\\
\leq  
  \Big(K'\,|\s|^{|\s| +1} \, M^{2|\s|}\,\ee^{ \gamma \, \| C^{E \pm \ii \eps}\|_{\infty,1}} \Big)^{|\YY |} \, \prod_{i = 1}^{|\YY|}
  ((d_{i}+2|\s|+2)!)^{p} \,M^{d_{i}} \;.
\end{multline}
for some $K' >K$. To estimate the second line in \eqref{eq: partial_bound_G_N}, the strategy is the same as in the proof of Theorem~\ref{thm: exponential_decay_strong_disorder}. We exploit the exponential decay of the covariance \eqref{eq: decay_covariance_hopping_operator}.
%:
%%
%\begin{equation}
%\sup _{\sigma,\sigma'} \big|\big(C^{E \pm \ii \eps}_{x_{i},x_{j}}\big)_{\sigma,\sigma'} \big| \leq \mathcal{C} \,\frac{\ee^{- \sqrt{\vphantom{E^{E}}\delta} \,|x_{i} - x_{j} |}} {1 + |x_{i} - x_{j} |^{D-2}}  \; ,
%\end{equation}
%%
If we define $ (C_{\theta}^{E \pm \ii \eps})_{x,y} := \ee^{(1+\theta) \sqrt{\delta}|x-y|/2} C^{E \pm\ii \eps}_{x,y} $, it follows that 
\begin{equation}
\begin{split}
& \big\| C_{\theta}^{E \pm \ii \eps}\big \|_{\infty,1} 
= \sup _{x \in \mathbb{Z}^{\mathrm{D}}}  \sum_{x'\in \mathbb{Z}^{\mathrm{D}}} \ee^{(1+\theta) \sqrt{\delta}|x'|/2} \,\sum_{\sigma,\sigma' \in \s} \big| \big(C_{x,x'}^{E \pm \ii \eps}\big)_{\sigma,\sigma'} \big| \leq  \mathcal{C}_{\theta} \, \delta^{-1} \;,
\\
& \big\| C_{\theta}^{E \pm \ii \eps}\big \|_{\infty,\infty} =  \sup _{x,x' \in \mathbb{Z}^{\mathrm{D}}}  \,\ee^{(1+\theta)\sqrt{\delta} |x'|/2} \,\sum_{\sigma,\sigma' \in \s} \big| \big(C_{x,x'}^{E \pm \ii \eps} \big)_{\sigma,\sigma'}\big| \leq \mathcal{C} \;,
\end{split}
\end{equation}
for some new $\mathcal{C}_{\theta} > \mathcal{C}$. Standard tree-stripping estimates (see Fig.~\ref{fig: trees}) based on the exponential decay of $C^{E \pm \ii \eps}$ give \cite{Brydges84}:
\begin{multline}
\label{eq: tree_like_estimate_C}
 \sum _{\underline{\sigma}}\sum_{\substack{x_{1},...,x_{N} \in \Lambda' \\ \mathrm{distinct}}}\prod_{\ell \in \textsc{T}_{\wp}}
\big| \big(C^{E \pm \ii \eps}_{x_{\ell^{+}},x_{\ell^{-}}}\big)_{\sigma^{+}_{\ell},\sigma^{-}_{\ell}} \big| 
\\
\leq  \ee^{-\theta \sqrt{\delta} |x - y|} \, 
 \sum _{\underline{\sigma}}\sum_{\substack{x_{1},...,x_{N} \in \Lambda' \\ \mathrm{distinct}}}\prod_{\ell \in \textsc{T}_{\wp}}
\big| \big((C^{E \pm \ii \eps}_{2\theta-1})_{x_{\ell^{+}},x_{\ell^{-}}}\big)_{\sigma^{+}_{\ell},\sigma^{-}_{\ell}} \big| 
\\
\leq 
\ee^{-\theta \sqrt{\delta} |x - y|} \, \big\| C_{\theta}^{E \pm \ii \eps}\big \|_{\infty,1}^{N} \, \big\| C_{\theta}^{E \pm \ii \eps}\big \|_{\infty,\infty}^{1- \delta_{x,y}} 
\prod_{i = 1}^{| \YY|} 
(C_{q,\theta})^{d_{i}-1} \, (d_{i}!)^{-q}
 \;.
\end{multline}
The details on the tree-stripping procedure are explained in \eqref{eq: detail_stripping_procedure}, text around it and in Fig.~\ref{fig: trees}.
Finally, we plug into \eqref{eq: partial_bound_G_N} the bounds \eqref{simple_bound_gibbs_weight} (recall that $\delta^{-1} \gamma \leq 1$), \eqref{eq: superintegral_bound_fourier} and \eqref{eq: tree_like_estimate_C} with $q= p+1$, and we use Cayley's theorem on the number of trees with fixed coordination numbers $\{ d_{i} \}_{i}$, see \cite{Brydges84}, to obtain:
\begin{equation}
\begin{split}
\big | \big(G_{N}(x,y;z_{\pm} \pm \ii \eps)\big)_{\sigma,\sigma'} \big |
& \leq  \ee^{-\theta \sqrt{\delta}|x-y| } 
\big(  K'\,|\s|^{|\s| +1} \, M^{2|\s|} \, \ee^{ \mathcal{C}}\big)^{| \YY|} 
\mathcal{C}
 (\mathcal{C}_{\theta} \, \delta^{-1})^{N}  \;
\\
&
\sum_{ \substack{\{ d_{i}\}_{i}, \, d_{i} \geq 1 \\  \sum_{i} d_{i} = 2| \YY| - 2}} \, \prod_{i = 1}^{|\YY |} \frac{((d_{i}+2|\s|+2)!)^{p} (C_{p+1,\theta})^{d_{i} -1} \,M^{d_{i}}}{(d_{i} - 1)! \, (d_{i}!)^{p+1}}\,
\\
& \leq (C_{K,M,p,\theta}/2)^{| \YY|} \,\delta^{-N} \,\ee^{-\theta \sqrt{\delta} |x - y|} 
\end{split}
\end{equation}
where 
$
C_{K,M,p,\theta} := 2 K'\,|\s|^{|\s| +1} \, M^{2|\s|+1} \ee^{\mathcal{C}} \mathcal{C}_{\theta}\sum_{n = 0}^{\infty} \frac{((n+2|\s|+2)!)^{p} \, (M\,C_{p+1,\theta})^{n}}{(n!)^{p+2}}$ $< \infty \;.
$
\qed
\medskip

%
%
%%%%%%%%%

\subsection*{Proof of Theorem \ref{thm: Lifshitz_LDOS_weak}}
To prove Theorem~\ref{thm: Lifshitz_LDOS_weak}, we need to further expand the super Gibbs' weight. We will need the following lemma.
\begin{lemma}[Super Lagrange Remainder]
\label{lemma: Lagrange_remainder}
Let $f: \mathcal{S}^{\XX} \to \mathscr{G}^{\XX}$. If $\| f(\Phi) \| < \infty$, then the following Lagrange estimate holds
\begin{equation}
\Big \| \ee^{f(\Phi)} - \sum_{j = 0}^{k} \frac{\big(f(\Phi)\big)^{j}}{j!} \Big \| \leq \frac{1}{(k+1)!}\sup _{t \in (0,1)} \Big \| \ee^{tf(\Phi)} \big( f(\Phi) \big)^{k+1}\Big \| \;.
\end{equation}
\end{lemma}
\begin{proof}
Since $\| f(\Phi) \|$ is finite, the function $t \mapsto \ee^{t f(\Phi)} = \sum_{n \geq 0}(t f(\Phi))^{n}/n! \in \mathscr{G}^{\XX}$ is analytic.
As usual, define the integral Lagrange remainder
\begin{equation}
R_{k}(\Phi):= \int_{0}^{1} \dd t \frac{(1-t)^{k}}{k!}  \Big( \frac{\dd }{\dd t}\Big)^{k} \, \ee^{t f(\Phi)}
= \ee^{f(\Phi)} - \sum_{j = 0}^{k} \frac{\big(f(\Phi)\big)^{j}}{j!} \;,
\end{equation}
which we estimate in Grassmann norm, hence the claim.
\end{proof}

For the sake of brevity, we shall write $G_{N}(E + \ii \eps) := G_{N}(0,0;E + \ii \eps)$.
By Proposition~\ref{thm: cluster_expansion_green_function_weak}
we have that for $E \in \mathbb{R}$
\begin{equation}
\rho_{\varepsilon,\Lambda}(E) = -
\frac{\gamma^{-1}}{\pi|\s|} \sum_{N = 0}^{|\Lambda|} \re \Big ( (\ii \gamma)^{N} \mathrm{Tr}_{\s} \,G_{N}(E + \ii \eps) \Big) \;,
\end{equation}
where, by the proof of Theorem~\ref{thm: exp_decay_weak} the following bound holds for $C_{K,M,p} := C_{K,M,p,\theta = 0}$:
\begin{equation}
\label{eq: bound_G_N_generic}
\sup _{\sigma}\big| \big(G_{N}(E + \ii \eps) \big)_{\sigma,\sigma} \big| \leq 
(C_{K,M,p}/2)^{N + 1}\, \delta^{-N}   \;,
\end{equation}
provided that $\delta \geq \gamma \, C_{M,K,p}$.
Let us now fix $ \overline{N} \in \mathbb{N}$ sufficiently large to be optimized later. To prove the claim, it suffices to improve bound \eqref{eq: bound_G_N_generic} for $N < \overline{N}$ as follows:
\begin{equation}
\label{eq: bound_G_N_improved}
\sup _{\sigma} \Big |  \re \Big ( (\ii \gamma)^{N}\big( G_{N}(E + \ii \eps)\big)_{\sigma,\sigma} \Big) \Big| 
\leq
(\overline{N}!)^{2p} (\gamma \,\delta^{-1}\widetilde{C}_{M,p})^{\overline{N}} 
  (\widetilde{C}_{K,M,p})^{N+1}
\end{equation}
for some constants $\widetilde{C}_{M,p}$ and $\widetilde{C}_{K,M,p}$.
Indeed, if bounds \eqref{eq: bound_G_N_generic} and \eqref{eq: bound_G_N_improved} hold true, then for any $\overline{N}$:
\begin{equation}
\big|\rho_{\varepsilon,\Lambda}(E) \big| \leq \gamma^{-1} C'_{K,M,p}
(\overline{N}!)^{2p} \big(C'_{K,M,p} \gamma\, \delta^{-1}\big)^{ \overline{N}} \;
\end{equation}
for some other constant $C'_{K,M,p}$. Therefore, by minimizing over $\overline{N} \in \mathbb{N}$ we obtain the statement for a suitable constant $\bar{C}_{K,M,p}$:
\begin{equation}
\big| \rho_{\Lambda}(E) \big|
\leq \gamma^{-1} C'_{K,M,p} \, \ee ^{-\bar{C}_{K,M,p} \,(\gamma^{-1}\, \delta)^{1/2p}} \;.
\end{equation}
Let us now prove \eqref{eq: bound_G_N_improved}. To begin, 
we Taylor expand the super Gibbs' weight $\widehat{\mu}_{\yy}^{+}(\xi,s)$:
\begin{equation}
\begin{split}
 \widehat{\mu}_{\yy}^{+}(\xi,s) &= \sum_{j = 0}^{\overline{N} - N -1} \frac{(\ii \gamma )^{j}}{j!} \Big(\frac{1}{2} \sum_{x,y \in \yy} s_{x,y} \hat{v}_{x,y}^{+}(\xi)\Big)^{j} 
 \, + \, \mathrm{Rem}_{_{\overline{N} - N}}\,\widehat{\mu}_{\yy}^{+}(\xi,s)
 \end{split}
\end{equation}
and accordingly introduce the splitting,
\begin{equation}
G_{N}(E + \ii \eps) = 
G_{\overline{N},N}(E + \ii \eps) + R_{\overline{N},N}(E + \ii \eps)
\end{equation}
where
\begin{equation}
\label{eq: definition_G_NN}
\begin{split}
& G_{\overline{N},N}(E + \ii \eps) := 
\\
& =
\frac{1}{N!}
\sum_{l = 0}^{\overline{N}-N-1} \frac{(\ii \gamma)^{l}}{l!}
 \sum_{\substack{\textsc{t} \,\,\, \mathrm{on} \\
 \,\{1,...,N+1\}}} \,\,
\sum_{\substack{x_{1},...,x_{N} \in \Lambda' \\ \mathrm{distinct}}}  
\int \dd p_{\textsc{t}}(s)  \int \dd \xi^{\yy} \;\cdot
\\
& \cdot \;  \Big( \sum_{x,y \in \yy}
 \frac{1}{2} s_{x,y} \,\hat{v}_{x,y}^{+}(\xi)
\Big)^{l} \,
 \prod_{\{i,j\} \in \textsc{t}} \hat{v}^{+}_{x_{i},x_{j}}(\xi) \,\, \frac{\partial}{\partial \eta^{+}_{x}} \frac{\partial}{\partial \eta^{-}_{y}}\widehat{F_{0}}^{\yy}(\xi)
 \\
 &
 = \frac{1}{N!}
\sum_{l = 0}^{\overline{N}-N-1} \frac{(\ii \gamma)^{l}}{l!}
 \sum_{\substack{\textsc{t} \,\,\, \mathrm{on} \\
 \,\{1,...,N+1\}}} \,\,
\sum_{\substack{x_{1},...,x_{N} \in \Lambda' \\ \mathrm{distinct}}}  
\int \dd p_{\textsc{t}}(s) \; \cdot 
\\
& \cdot \;  \Big( - \frac{1}{2}\sum_{x,y \in \yy}
s_{x,y} \, \hat{v}^{+}_{x,y}(\partial /\partial \Phi)
\Big)^{l} \,
 \prod_{\{i,j\} \in \textsc{t}} 
 \hat{v}^{+}_{x_{i},x_{j}}(\partial /\partial \Phi) \,\, \psi^{-}_{0}\psi^{+}_{0} \, F_{0}^{\yy}(\Phi) \Big|_{\Phi = 0}
\end{split}
\end{equation}
with $\hat{v}^{+}_{x,y}(\partial /\partial \Phi) :=  \partial /\partial \Phi_{x}^{+} C_{x,y}^{E + \ii \eps} \partial /\partial \Phi_{y}^{-} + \partial /\partial \Phi_{y}^{+} C_{y,x}^{E + \ii \eps} \partial /\partial \Phi_{x}^{-}$.
Since $\nu$ is even by assumption, $\hat{\nu}^{(2n+1)}(0)= \mathbb{E} \, \omega^{2n+1} = 0$ and thus $\psi^{-}_{0} \psi^{+}_{0}F^{\yy}_{0}(\Phi)$ has non-vanishing derivatives in $\Phi = 0$ only of order $2(2n+1)$, $n \in \mathbb{N}$.
As a consequence, only the terms such that $l + N$ is odd survive in the sum above. Finally, $ C_{x,y}^{E + \ii 0^{+}} \in \mathbb{R}$ for any $x,y$, and $\hat{\nu}^{(2n)}(0) \in \mathbb{R}$ imply that $G_{\overline{N},N}(E + \ii \eps)$ does not contribute to $\rho_{\Lambda}(E)$, that is
\begin{equation}
\re \Big ( (\ii \gamma)^{N} G_{\overline{N},N}(E + \ii 0^{+}) \Big) = 0 \;.
\end{equation}
We are left with bounding $\big(R_{\overline{N},N}(E + \ii \eps )\big)_{\sigma,\sigma}$:
\begin{equation}
\begin{split}
 & R_{\overline{N},N}(E + \ii \eps)  =
 \frac{1}{N!}\sum_{\substack{\textsc{t} \,\,\, \mathrm{on} \\
 \,\{1,...,N+1\}}} \,\sum_{\wp, \, \underline{\sharp}, \, \underline{\sigma}}\,\,
\sum_{\substack{x_{1},...,x_{N} \in \Lambda' \\ \mathrm{distinct}}}  
\,
\int \dd p_{\textsc{t}}(s) \; \cdot
\\
& \qquad \cdot \; \int \dd \xi_{\yy} 
\Big(\mathrm{Rem}_{\overline{N}-N-1}
 \widehat{\mu}_{\yy}^{+}(\xi,s)\Big) \,
\mathcal{F}_{\textsc{t}_{\wp},\,\underline{\sharp}, \, \underline{\sigma}}^{\yy}(\xi)  \,\Big (\prod_{\ell \in \textsc{T}_{\wp}}
\big(C^{E \pm \ii \eps}_{x_{\ell^{+}},x_{\ell^{-}}}\big)_{\sigma_{\ell}^{+}, \sigma^{-}_{\ell}}\Big ) \;,
\end{split}
\end{equation}
where the notation is again borrowed from the previous proof, see \eqref{eq: F_mathcal_weak} and the text above.
We introduce $\tilde{x} = (\tilde{x}^{+},\tilde{x}^{-}) \in \big(\YY^{\overline{N}-N}\big)^{ 2}$, $ \tilde{\sharp} \in \{B,F \}^{\overline{N}-N}$ and $\tilde{\sigma} = (\tilde{\sigma}^{+},\tilde{\sigma}^{-}) \in \big(\s^{\overline{N}-N}\big)^{2}$. For any $\tilde{x}$, $\tilde{\sharp}$ and $\tilde{\sigma}$, we define a sequence $\underline{\tilde{d}_{i}} := \{ (\tilde{d}_{i})^{\varepsilon}_{\sharp,\sigma}\}_{\varepsilon = \pm, \sharp = B,F}$ at any vertex $i \in \{1,...,N+1 \}$ of the tree,  where
$(\tilde{d}_{i})^{\varepsilon}_{\sharp,\sigma}:= \sum_{j = 1}^{\overline{N} - N} \delta_{x_{i},\tilde{x}_{j}^{\varepsilon}} \delta_{\sharp,\tilde{\sharp}_{j}} \delta_{\sigma,\tilde{\sigma}^{\varepsilon}_{j}}$. Furthermore, we define the following local function:
\begin{equation}
\label{eq: def_F_remainder_local_function}
\mathcal{F}_{\substack{\textsc{t}_{\wp},\,\underline{\sharp}, \, \underline{\sigma} \\ \tilde{x},\,\tilde{\sharp},\,\tilde{\sigma}}}^{\yy,}(\xi) := 
\Big (\prod_{i =1}^{N+1} \prod_{\varepsilon,\sharp,\sigma} \big(\xi^{\varepsilon}_{\sharp,x_{i},\sigma} \big)^{(\tilde{d}_{i})_{\sharp,\sigma}^{\varepsilon}} \Big)
 \mathcal{F}_{\textsc{t}_{\wp},\,\underline{\sharp},\,\underline{\sigma}}^{\yy}(\xi)
 = \pm \prod_{i = 1}^{N+1} \, \mathcal{F}_{x,y}^{\, \underline{d_{i}} + \underline{\tilde{d}_{i}}}(\xi_{x_{i}}) \;,
\end{equation}
where again the sign is unimportant for our purposes and where $\widehat{\mathcal{F}}_{x,y}^{\,\underline{d}}(\xi_{x'})$ was introduced in \eqref{eq: mathcal_F_local_function}. Finally, by Lemma~\ref{lemma: Lagrange_remainder}, we estimate $\mathrm{Rem}_{\overline{N} - N}$ in Grassmann norm as a super Lagrange remainder:
\begin{multline}
\big| \big(R_{\overline{N},N}(E + \ii \eps) \big)_{\sigma,\sigma}\big| 
\leq  \, \frac{4^{N+1}}{N!} \frac{(2\gamma)^{\overline{N}-N}}{(\overline{N}-N)!}
\\
\sum_{\substack{\textsc{t} \,\,\, \mathrm{on} \\
 \,\{1,...,N+1 \}}} \,
\sup _{\substack{x_{1}, \dots, x_{N} \\ s}} \sup _{\substack{\wp , \,\underline{\sharp},\,\underline{\sigma} \\  \tilde{x},\, \tilde{\sharp}, \, \tilde{\sigma}}} 
\int \dd \kappa_{\yy}  \Big \| \widehat{\mu}^{\pm}_{\yy}(\xi,s) \Big \| \Big\|\;
\Big(\mathcal{F}_{\substack{\textsc{t}_{\wp},\,\underline{\sharp}, \, \underline{\sigma} \\ \tilde{x},\,\tilde{\sharp},\,\tilde{\sigma}}}^{\yy}(\xi) \Big)_{\sigma,\sigma}\Big\| 
 \\
 \qquad\bigg ( \sup _{\wp} \, \sum_{\underline{\sigma}} \sum_{\substack{x_{1},...,x_{N} \in \Lambda' \\ \mathrm{distinct}}}  
\prod_{\ell \in \textsc{T}_{\wp}}
\big| \big(C^{E + \ii \eps}_{x_{\ell^{+}},x_{\ell^{-}}}\big)_{\sigma^{+}_{\ell},\sigma^{-}_{\ell}}\big|  \bigg) 
\\
 \bigg(
\sum_{\tilde{x},\, \tilde{\sigma}} \prod_{j = 1}^{\overline{N} -N} \big| \big(C^{E + \ii \eps}_{\tilde{x}^{+}_{j},\tilde{x}^{-}_{j}}\big)_{\tilde{\sigma}^{+}_{j},\tilde{\sigma}^{-}_{j}}\big| \bigg) \, 
 \;.
\end{multline}
In order to bound the integral, we follow the strategy used in the previous proof. Using Eq.~\eqref{eq: def_F_remainder_local_function}
we notice that
\begin{equation}
 \sup _{\substack{\wp , \,\underline{\sharp},\,\underline{\sigma} \\  \tilde{x},\, \tilde{\sharp}, \, \tilde{\sigma}}} 
 \Big \|\Big(\mathcal{F}_{\substack{\textsc{t}_{\wp},\,\underline{\sharp}, \, \underline{\sigma} \\ \tilde{x},\,\tilde{\sharp},\,\tilde{\sigma}}}^{\yy}(\xi) \Big)_{\sigma,\sigma} \Big \| \leq 
\prod_{i = 1}^{|\YY |} \; 2 K \, \frac{((d_{i}+\tilde{d}_{i}+2m)!)^{p} \, M^{d_{i} + \tilde{d}_{i}} }{1 + M^{-2m}\, (\kappa_{x_{i}}^{+} \kappa_{x_{i}}^{-})^{m}}, \qquad \forall m \in \mathbb{N} \;.
\label{eq: decay_global_fourier_transform_rem}
\end{equation}
Thus:
\begin{multline}
\sup _{\substack{x_{1}, \dots, x_{N} \\ s}} \sup _{\substack{\wp , \,\underline{\sharp},\,\underline{\sigma} \\  \tilde{x},\, \tilde{\sharp}, \, \tilde{\sigma}}} 
\int \dd \kappa_{\yy}  \Big \| \widehat{\mu}^{\pm}_{\yy}(\xi,s)\Big \| \; \Big \|
\Big(\mathcal{F}_{\substack{\textsc{t}_{\wp},\,\underline{\sharp}, \, \underline{\sigma} \\ \tilde{x},\,\tilde{\sharp},\,\tilde{\sigma}}}^{\yy}(\xi) \Big)_{\sigma,\sigma}\Big\| 
\\
\leq
  (K'_{M})^{|\YY |} \,
  \big((2\overline{N} - 2N)!\big)^{p} \prod_{i = 1}^{|\YY |}
  ((d_{i}+ 2|\s|+2)!)^{p} \,(2^{p}M)^{d_{i} + \tilde{d}_{i}} \;.
\end{multline}
for some constant $K'_{M}$ depending on $|\s|$ as well. Finally, by noticing that 
\begin{equation}
\sum_{\tilde{x}, \, \tilde{\sigma}} \prod_{j = 1}^{\overline{N} -N} \big| \big(C^{E + \ii \eps}_{\tilde{x}^{+}_{j},\tilde{x}^{-}_{j}} \big)_{\tilde{\sigma}^{+}_{j},\,\tilde{\sigma}^{-}_{j}}\big|  \leq \big( |\YY | \, |\s|\,\| C^{E + \ii \eps}\|_{\infty,1} \big)^{\overline{N} - N} \;,
\end{equation}
we obtain the following bound for some constant $\widetilde{C}_{M,p}$,
\begin{equation}
\begin{split}
\big| \big(R_{\overline{N},N}(E + \ii \eps) \big)_{\sigma,\sigma}\big|  \leq
(\overline{N}!)^{2p} (\delta^{-1}\widetilde{C}_{M,p})^{\overline{N}} 
 \gamma^{\overline{N}-N} (C_{K,2^{p} M,p, \theta = 0})^{N+1} \;,
\end{split}
\end{equation}
which implies the desired estimate \eqref{eq: bound_G_N_improved}.
\qed
\medskip

\noindent{\bf Acknowledgements.} 
I express my sincere gratitude to Marcello Porta for inspiring discussions and valuable guidance.
This work has been supported by the Swiss National Science Foundation via the grant ``Mathematical Aspects of Many-Body Quantum Systems'' and by the European Research Council (ERC) under  the European  Union's  Horizon  2020  research  and  innovation programme, ERC  Starting Grant  MaMBoQ, Grant Agreement
No.~802901.
I thank the anonymous referee for useful comments on a previous version of
the paper.

\appendix

\section{Disorder Distribution}

%In this section we discuss some connections between hypotheses (H2-$\mathrm{I}$), (H2-$\mathrm{II}$) and the disorder distribution.

We discuss two examples that connect the density $\nu$ with the IMB of Definition~\ref{def: bounds_strong_disorder}:
\medskip

\underline{\textit{Example} $\mathrm{I}.1$:}
We can state the following lemma.
\begin{lemma}
Let $\nu \in \mathscr{S}(\mathbb{R})$ be analytic in the strip $| \im \,t| < W$. Then, if $\gamma^{-1}z < W$, $F_{z}$ satisfies IMB with $M = (W/2 - \gamma^{-1}z/2 )^{-1/2}$ and $p = 1/2$. 
\end{lemma} 
\begin{proof}
The assumptions on $\nu$ imply that $\|\hat{\nu}(\Phi^{+}\Phi^{-}) \| \leq C \ee^{-W \phi^{+}\phi^{-}}$, for some universal constant $C$. Using that $|(\phi^{+}_{\sigma})^{n_{\sigma}^{+}}(\phi^{-}_{\sigma})^{n_{\sigma}^{-}}| \leq (n_{\sigma}/2)!\, M ^{-n_{\sigma}/2}\,\ee^{M\phi^{+}_{\sigma}\phi^{-}_{\sigma}}$, where $n_{\sigma} =n_{\sigma}^{+} + n_{\sigma}^{-} $ and $M = W/2 - \gamma^{-1}z/2 $ the claim follows. Here $K$ grows with $M^{-1}$.
\end{proof}
Notice that Example I.1 include Gaussian distributions.
\medskip

\underline{\textit{Example} $\mathrm{I}.2$:} The decay of the Fourier transform of test functions can be quantitatively characterised, see, e.g., \cite{Johnson}. For instance, if $\nu(t) = \ee^{-\frac{1}{1-t^{2}}}\mathbf{1}_{|t|\leq 1}$, then $F_{0}$ satisfies IMB with $M = 2$ and $p = 1$.
\bigskip

We then discuss two examples that connect the density $\nu$ with the IDB of Definition~\ref{def: bounds_weak_disorder}. Both examples include Gaussian distributions.
\medskip

\underline{\textit{Example} $\mathrm{II}.1$:} Introduce the following seminorm on the functions of a supervector:
\begin{equation}
\| f \|_{1,W} := 
\sup _{\zeta \in \mathbb{R}_{W}^{2}}
\int_{\mathbb{R}^{2\s}} \dd \phi \,\big\| f(\Phi + (\zeta,0))\big\| \;,
\end{equation}
where $ \mathbb{R}_{W}^{2\s} = \big \{  \phi \in \mathbb{C}^{2\s} \; | \; |\im \, \phi_{{\scriptscriptstyle{i}},\sigma}| \leq W , \; i = 1,2 \;, \;\,\sigma \in \s \big \} $. We can state the following lemma.
\begin{lemma}
Assume that for any $\alpha \in \mathbb{R}$, $\ee^{\alpha |t|} \nu (t)$ is bounded and that  $ \| F_{z} \|_{1,W} $ is finite for some $0 < W  \leq 1$. Then $F_{z}$ satisfies IDB with $K = \| F_{z} \|_{1,W}$, $M = W^{-1}$ and $p = 1$.
\end{lemma}
\begin{proof}
We notice that $\hat{\nu}$ is entire and so is $F_{z}$ in $\phi$. 
We thus apply multi-variable Cauchy integral formula in $\mathbb{R}_{W}^{2\s}$. Accordingly,
\begin{multline}
 \bigg[\prod_{i, \sigma} \Big(\frac{\partial}{\partial \phi_{i,\sigma}} \Big)^{n_{i,\sigma}} 
 \bigg] F_{z}(\Phi) 
 = \Big(\prod_{i,\sigma} \frac{n_{i,\sigma}!}{2\pi \ii} \Big) \oint_{(\partial D(0,W))^{2\s}} \dd w \, \frac{F_{z}\big(\Phi + (w,0)\big)}{\prod_{i,\sigma} w_{i,\sigma}^{n_{i,\sigma}+1}} \,,
\end{multline}
where $D(0,W) = \big\{ w \in \mathbb{C} \, \big| \,|w| \leq W \big\}$ and $\dd w = \bigtimes_{i,\sigma} \dd w_{i,\sigma}$; therefore:
\begin{equation}
\begin{split}
\int \dd \phi  \,\Big\|\bigg[\prod_{i, \sigma} \Big(\frac{\partial}{\partial \phi_{i,\sigma}} \Big)^{n_{i,\sigma}} 
 \bigg] F_{z}(\Phi) \Big \| 
\leq n! \, W^{-n}
\,\|F_{z} \|_{1,W} \;,
\end{split}
\end{equation}
where $n = \sum_{i,\sigma} n_{i,\sigma}$. 
The claim follows because $\partial /\partial\phi^{\varepsilon}_{\sigma} = 1/2 ( \partial /\partial \phi_{{\scriptscriptstyle{1}},\sigma}$ $ -\ii \varepsilon  \partial /\partial \phi_{{\scriptscriptstyle{2}},\sigma} )$.
\end{proof}
\medskip 

\underline{\textit{Example} $\mathrm{II}.2$:} Introduce the following seminorm on the functions of a supervector:
\begin{equation}
\triple{f}_{1,W} := \sup _{\zeta \in  ( \partial D(0,W))^{2\s}}
\int _{\mathbb{R}^{2\s}} \dd \phi \,\, \ee^{\,\sum_{\sigma \in \s} (\phi^{+}_{\sigma} \phi^{-}_{\sigma})^{1/2}} \, \|f\big(\Phi + (\zeta,0)\big) \| \;,
\end{equation}
where $D(0,W) \subset \mathbb{C}$ was defined above. We can state the following lemma.
\begin{lemma}
Assume $\ee^{W|t|} \nu(t)$ is bounded and $ \triple{F_{z}}_{1,W} $ is finite for some $0 < W  \leq 1$. Then $F_{z}$ satisfies factorial bounds with $K = |\s|!\,2^{|s|}\,\triple{F_{z}}_{1,W}$, $M = 4(W^{-1}+ \gamma^{-1}z)$, $p = 2$.
\end{lemma}
\begin{proof}
We notice that $\hat{\nu}(t)$ is holomorphic on the strip $| \im \, t| \leq W$ which we shall use to estimate its derivative. We think of $F_{z}(\Phi)$ as the composite function $f(\Phi^{+}\Phi^{-})$ and compute its derivative accordingly. For simplicity, we carry out the computation in the case of $|\s| = 1$ and we take, e.g.,  $n_{\sigma}^{+}\geq n_{\sigma}^{-}$ and $n_{\sigma} = n_{\sigma}^{+}+n_{\sigma}^{-}$. We accordingly bound the norm of the derivative as follows:
\begin{equation}
\Big\|\Big( \prod_{\varepsilon = \pm}  \Big(\frac{\partial}{\partial \phi^{\varepsilon}_{\sigma}} \Big)^{n^{\varepsilon}_{\sigma}}\Big)  F_{z}(\Phi) \Big\| \leq 2^{n_{\sigma}} (\phi^{+}_{\sigma}\phi_{\sigma}^{-})^{n_{\sigma}/2} 
\sum_{i = 0}^{n_{\sigma}^{-}} 
i! \|f^{(n_{\sigma}-i)}(\Phi^{+}\Phi^{-}) \| \;.
\end{equation}
We apply Cauchy integral representation to estimate the derivatives of $f(\cdot)$ and we use the bound  $(\phi^{+}_{\sigma}\phi_{\sigma}^{-})^{n_{\sigma}/2} \leq n_{\sigma}! \,\ee^{(\phi^{+}_{\sigma}\phi^{-}_{\sigma})^{1/2}}$. Integrating in $\dd \phi_{\sigma}$ and taking the superior over the contour variable gives the claim.
\end{proof}


\begin{thebibliography}{}

\bibitem{AbdesselamRivasseau}
A.~Abdesselam and V.~Rivasseau.
Trees, forests and jungles: A botanical garden for cluster expansions.
In: \textit{Constructive Physics Results in Field Theory, Statistical Mechanics and Condensed Matter Physics}, Springer (1995).

\bibitem{Aizenman}
M.~Aizenman.
Localization at weak disorder: some elementary bounds.
\textit{Rev. Math. Phys.} \textbf{6}, 1163 (1994).

\bibitem{AizenmanMolchanov}
M.~Aizenman and  S.~Molchanov.
Localization at Large Disorder
and at Extreme Energies: An Elementary Derivation. \textit{Commun. Math. Phys.} \textbf{157}, 245-278 (1993).

\bibitem{AizenmanSchenker}
M.~Aizenman, J.~H.~Schenker, R.~M.~Friedrich and D.~Hundertmark.
Finite-Volume Fractional-Moment Criteria
for Anderson Localization.
\textit{Commun. Math. Phys.} \textbf{224}, 219-253 (2001).

\bibitem{AizenmanWarzel}
M.~Aizenman and S.~Warzel.
\textit{Random Operatos. Disorder Effects
on Quantum Spectra and Dynamics.}
Graduate Studies in Mathemathics \textbf{168}, Americal Mathematical Society (2015).

\bibitem{FrestaPorta}
G.~Antinucci, L.~Fresta and M.~Porta.
A Supersymmetric Hierarchical Model for Weakly Disordered $3d$ Semimetals.
\textit{Ann. Henri Poincar\'e} 2020,  \url{https://doi.org/10.1007/s00023-020-00909-1}.

\bibitem{BattleFederbush}
G.~Battle and P.~ Federbush. 
A phase cell cluster expansion for Euclidean field theories. 
\textit{Ann. Phys. }\textbf{142}, 95-139 (1982).

\bibitem{BauerschmidtBrydges} 
R.~Bauerschmidt, D.~C.~Brydges and G.~Slade. {\it Introduction to a renormalisation group method.}

\bibitem{Berezin66}
F.~A.~Berezin.
\textit{The Methods of Second Quantization}.
Academic Press, (1966).

\bibitem{Berezin87}
F.~A.~Berezin.
\textit{Introduction to Superanalysis}.
Springer (1987).

\bibitem{BerezinMarinov}
F.~A.~Berezin and M.~S.~Marinov. Particle spin dynamics as the
Grassmann variant of classical mechanics. \textit{Ann. Phys.} \textbf{104}, 336-362 (1977).

\bibitem{Berline}
N.~Berline, E.~Getzler and M.~Vergne. Heat kernels and Dirac operators, Springer 
(1992).

\bibitem{BlauThompson}
M.~Blau and G.~Thompson. 
Localization and Diagonalization: A re-
view of functional integral techniques for low-dimensional gauge theories and topological field theories.
\textit{J. Math. Phys.} \textbf{36}, 2192-2236 (1995).

\bibitem{Bovier}
A.~Bovier. 
The density of states in the Anderson model at weak disorder: A renormalization group analysis of the hierarchical model. \textit{J Stat Phys.} \textbf{59}, 745-779 (1990).

\bibitem{BovierKleinPerez}
A.~Bovier, M.~Campanino, A.~Klein and J.~F.~Perez.
Smoothness of the Density of States in the Anderson Model at High Disorder. 
\textit{Commun. Math. Phys.} \textbf{114}, 439-461 (1988).

\bibitem{Brydges84}
D.~C.~Brydges,
A short course on cluster expansions. In:
\textit{Les Houches Summer School in Theoretical Physics, Session XLIII: Critical Phenomena, Random Systems, Gauge Theories} (1984).

\bibitem{BrydgesKennedy}
D.~Brydges and T.~Kennedy.
Mayer expansions and the Hamilton-Jacobi Equation.
\textit{Journ. Stat. Phys.} \textbf{48}, 19 (1987).

\bibitem{CampaninoKlein}
M.~Campanino and A.~Klein.
A supersymmetric transfer matrix and differentiability of the density of states in the one-dimensional Anderson model. \textit{Comm. Math. Phys.} \textbf{104} 227-241 (1986). 

\bibitem{ConstantinescuFelder}
F.~Constantinescu, G.~Felder, K.~Gawedzki, and A.~Kupiainen. Analyticity of density of states in a gauge-invariant model for disordered electronic systems. \textit{J. Stat. Phys.} \textbf{48}, 365-391 (1987).

\bibitem{ConstantinescuSpencer}
F.~Constantinescu, J.~Fr\"ohlich and T.~Spencer.
Analyticity of the density of states and replica method for random Schr\"odinger operators on a lattice.
\textit{J. Stat. Phys.} \textbf{34} (1984).
%
%\bibitem{Disertori}
%M. Disertori.  
%Density of states for GUE through supersymmetric approach.
%\textit{Rev. Math. Phys.} \textbf{16} 1191-1225 (2004).

\bibitem{DisertoriLager2}
M.~Disertori and M.~Lager.
Density of States for Random Band Matrices in Two Dimensions. 
\textit{Ann. Henri Poincar\'e} \textbf{18}, 2367-2413 (2017).

\bibitem{DisertoriLager}
M.~Disertori and M.~Lager. Supersymmetric Polar Coordinates with Applications to the Lloyd Model. \textit{Math. Phys. Anal. Geom.} \textbf{23}, 2 (2020).

\bibitem{DisertoriPinson}
M.~Disertori, H.~Pinson and T.~Spencer. Density of states for random band matrices.
\textit{Comm. Math. Phys.} \textbf{232} 83-124 (2002).

\bibitem{DisertoriSpencer} 
M.~Disertori and T.~Spencer. Anderson localization for a supersymmetric sigma model. {\it Comm. Math. Phys.} {\bf 300}, 659-671 (2010).

\bibitem{DisertoriSpencerZirnbauer} 
M.~Disertori, T.~Spencer and M.~Zirnbauer. Quasi-diffusion in a 3D supersymmetric hyperbolic sigma model. {\it Comm. Math. Phys.} {\bf 300}, 435-486 (2010).

\bibitem{EdwardsThouless}
S.~Edwards and D.~Thouless. 
Regularity of the density of states in Anderson's localized electron model. 
\textit{J. Phys. C} \textbf{4}, 453-457 (1971).

\bibitem{Efetov82}
K.~B.~Efetov.
Supersymmetry method in localization theory. \textit{Soy. Phys. JETP} \textbf{55}, 514-521
(1982).
%
%\bibitem{Efetov83}
%K.~B.~Efetov.
%Supersymmetry and theory of disordered metals. \textit{Adv. Phys.} \textbf{32}, 53-127 (1983).

\bibitem{Efetov10}
K.~B.~Efetov. Anderson localization and Supersymmetry. In: \textit{50 years of Anderson localization}, World Scientific (2010).

\bibitem{Elgart}
A.~Elgart. Lifshitz tails and localization in the three-dimensional Anderson model.
\textit{Duke Math. J.} \textbf{146}, 331-360 (2009).

\bibitem{Feldman}
J.~Feldman, H.~Kn\"orrer, and E.~Trubowitz. Convergence of Perturbation Expansions in Fermionic Models. Part 1: Nonperturbative Bounds. \textit{Commun. Math. Phys.} \textbf{247}, 195-242 (2004).

\bibitem{FroehlichSpencer}
J.~Fr\"ohlich and T.~Spencer.
Absence of diffusion in the Anderson Tight Binding Model for Large Disorder or Low Energy. \textit{Commun. Math. Phys.} \textbf{88}, 151-184 (1983).

\bibitem{GawedzkiKupiainen}
K.~Gaw\c{e}dzki and A.~Kupiainen. Gross-Neveu model through convergent perturbation expansions. \textit{Comm. Math. Phys}. \textbf{102} (1985).

\bibitem{GentileMastropietro}
G.~Gentile and V.~Mastropietro.
Renormalization group for one-dimensional fermions.
A review on mathematical results.
\textit{Phys. Rep.} \textbf{352} 273-437 (2001). 

\bibitem{GiulianiMastropietro}
A.~Giuliani and V.~Mastropietro.
The two-dimensional Hubbard model on the honeycomb lattice.
\textit{Commun. Math. Phys} \textbf{293} 301 (2010).

\bibitem{HormanderI}
L.~H\"ormander.
\textit{The Analysis of Linear Partial Differential Operators I. Distribution Theory and Fourier Analysis}. Springer (1990).

\bibitem{Johnson}
S.~G.~Johnson.
Saddle-point integration of $C^{\infty}$ ``bump'' functions.
arXiv:1508.04376.

\bibitem{KleinMartinelli}
A.~Klein, F.~Martinelli, and J.~F.~Perez.
A rigorous replica trick approach
to  Anderson  localization  in  one  dimension.  \textit{Comm.  Math.  Phys.} \textbf{106}
(1986).

\bibitem{KleinPerez}
A.~Klein and J.~F.~Perez.
On  the  density  of  states  for  random  potentials  in the presence of a uniform magnetic field. \textit{Nuclear Phys. B}
\textbf{251} 199-211 (1985).

\bibitem{Klopp}
F.~Klopp.
Weak Disorder Localization and Lifshitz Tails.
\textit{Commun. Math. Phys.} \textbf{232}, 125-155 (2002).

\bibitem{KloppWolff}
F.~Klopp and T.~Wolff. 
Lifshitz tails for 2-dimensional random Schr\"odinger operators.  \textit{J. Anal. Math.} \textbf{88}, 63 (2002).

\bibitem{Lesniewski}
A.~Lesniewski. 
Effective action for the Yukawa$_{2}$ quantum field theory, \textit{Comm. Math. Phys.} \textbf{108} 437-467 (1987).

\bibitem{Mastropietro}
V.~Mastropietro.
\textit{Non-perturbative Renormalization.}
World Scientific (2008).

\bibitem{ParisiSourlas}
G.~Parisi and N.~Sourlas.
Random Magnetic Fields, Supersymmetry, and Negative Dimensions.
\textit{Phys. Rev. Lett.} \textbf{43}, 744 (1979).

\bibitem{ParisiSourlas2}
G.~Parisi and N.~Sourlas. 
Supersymmetric field theories and stochastic differential equations.
Nuclear Physics B \textbf{206} 321-332 (1982).

\bibitem{Salmhofer}
M.~Salmhofer.
\textit{Renormalization: An Introduction}.
Springer-Verlag (1999).

\bibitem{SamuelI}
S.~Samuel.
The use of anticommuting variable integrals
in statistical mechanics. I. The computation
of partition functions.
\textit{J. Math. Phys.} \textbf{21}, 2806 (1980). 

\bibitem{SamuelII}
S.~Samuel.
The use of anticommuting variable integrals
in statistical mechanics. I. The computation
of correlation functions.
\textit{J. Math. Phys.} \textbf{21}, 2815 (1980).

\bibitem{SchaferWegner}
L.~Sch\"afer and F.~Wegner. 
Disordered system with n orbitals per site: Lagrange formulation, hyperbolic symmetry, and Goldstone modes.
\textit{Z. Phys. B} \textbf{38} 113-126 (1980).

\bibitem{SchwarzZaboronsky} 
A.~Schwarz and O.~Zaboronsky. 
Supersymmetry and localization. {\it Comm. Math. Phys.} {\bf 183}, 463-476 (1997).

\bibitem{Shamis}
M.~Shamis. 
Density of states for Gaussian unitary ensemble, Gaussian orthogonal ensemble, and interpolating ensembles through supersymmetric approach.
\textit{J. Math. Phys.} \textbf{54} (2013)

\bibitem{Shcherbina1} 
M.~Shcherbina and T.~Shcherbina. 
Transfer operator approach to 1d random band matrices. arXiv:1905.08252.

\bibitem{Shcherbina2} 
M. Shcherbina and T. Shcherbina. 
Universality for 1 d random band matrices. 	arXiv:1910.02999

\bibitem{Shcherbina3}
T.~Shcherbina.
Universality of the local regime for the block band matrices with a finite
number of blocks. 
\textit{J. Stat. Phys.} \textbf{155} 466-499 (2014).

\bibitem{Simon}
B.~Simon.
Kotani Theory for One Dimensional Stochastic Jacobi Matrices.
\textit{Commun. Math. Phys.} \textbf{89}, 227-234 (1983).

\bibitem{SjostrandWang1}
J.~Sj\"ostrand and W.-M.~Wang.
Supersymmetric measures and maximum principles in the complex
domain. Exponential decay of Green's functions.
\textit{Annales scientifiques de l'\'E.N.S. $4^{e}$ s\'erie},  \textbf{32}, 347-414 (1999).

\bibitem{SjostrandWang2}
J.~Sj\"ostrand and W.-M.~Wang.
Exponential decay of averaged Green functions for random
Schr\"odinger operators. A direct approach.
\textit{Annales scientifiques de l'\'E.N.S. $4^{e}$ s\'erie},  \textbf{32}, 415-431 (1999). 

\bibitem{Spencer}
T.~Spencer. 
Lifshitz tails and localization. \textit{Preprint} (1993).

\bibitem{SpencerNotes}
T.~Spencer.
SUSY Statistical Mechanics and Random
Band Matrices. In: \textit{Quantum Theory from Small to Large Scales: Lecture Notes of the Les Houches Summer School: Volume 95}
(2010), and In: \textit{Quantum Many Body Systems
Cetraro, Italy 2010, Editors: A. Giuliani, V. Mastropietro, J. Yngvason} (2012).

\bibitem{Wang}
W.-M.~Wang.
Localization and universality of Poisson statistics for the multidimensional Anderson model at weak disorder.
\textit{Invent. math.} \textbf{146}, 365-398 (2001).
 
\bibitem{Wegner79}
F.~Wegner.
The mobility edge problem: continuous symmetry and a conjecture.
\textit{Z. Phys. B} \textbf{35} 207-210 (1979).
 
\bibitem{Wegner81}
F.~Wegner.
Bounds on the Density of States of Disordered Systems.
\textit{Z. Physik B} \textbf{44} (1981). 
 
\bibitem{Wegner16}
F.~Wegner,
Supermathematics and its Applications to Statistical Physics, \textit{Lecture Notes in Physics}, \textbf{920} (2016).

\end{thebibliography}
\end{document}